%% file: main.tex
\documentclass[journal]{IEEEtran}

\usepackage{array}
\usepackage{amsthm}
\usepackage{amsmath}
\usepackage{amssymb}
\usepackage[vlined, ruled, linesnumbered]{algorithm2e}
\usepackage[export]{adjustbox}
\usepackage{booktabs, tabularx}
\usepackage[skip=4pt,font={bf}]{caption}  
\usepackage{comment}
\usepackage{enumitem}
\usepackage{graphicx}
\usepackage{hyperref}
\usepackage{cleveref}
\usepackage[right]{lineno}
\usepackage{listings}
\usepackage{multirow}
\usepackage[numbers,sort&compress]{natbib}  
\usepackage{subfig}
\usepackage{tikz}
\usepackage{times}
\usepackage{thmtools, thm-restate}
\usepackage{url}
\usepackage{wasysym}  
\usepackage{xcolor}
\usepackage[binary-units=true]{siunitx}

\def\techreport{true}

\def\tablename{Table}
\renewcommand{\thetable}{\arabic{table}}  

\crefname{ineq}{inequality}{inequalities}
\creflabelformat{ineq}{#2{\upshape(#1)}#3} 

\definecolor{red}{HTML}{E51400}  
\definecolor{blue}{HTML}{0050EF} 
\definecolor{green}{HTML}{008A00} 
\definecolor{purple}{HTML}{AA00FF} 
\definecolor{dark-red}{rgb}{0.4, 0.15, 0.15}
\definecolor{dark-blue}{rgb}{0.15, 0.15, 0.4}
\definecolor{medium-red}{rgb}{0.5, 0, 0}
\definecolor{medium-blue}{rgb}{0, 0, 0.5}
\definecolor{light-red}{rgb}{0.7, 0, 0}
\definecolor{light-blue}{rgb}{0, 0, 0.7}

\hypersetup{
  colorlinks,
  linkcolor={medium-red},
  citecolor={medium-blue},
  urlcolor={medium-blue}
}

\SetCommentSty{mycommfont}
\DontPrintSemicolon
\SetKw{KwAnd}{and}
\SetKwProg{myfunc}{Procedure}{}{}
\SetKwProg{myalgo}{Algorithm}{}{}

\usetikzlibrary{shapes.geometric}%
\usetikzlibrary{positioning}%
\usetikzlibrary{arrows}%
\usetikzlibrary{fit}%
\usetikzlibrary{decorations.markings}
\usetikzlibrary{decorations.shapes}
\tikzset{->, >=stealth, auto,
  font=\small,
  main node/.style={circle, draw, inner sep=0cm, minimum size=0.5cm,
    fill=white, align=center},
  seed node/.style={circle, draw, inner sep=0cm, minimum size=0.5cm,
    fill=green!20, align=center}
}
\pgfdeclarelayer{background}
\pgfdeclarelayer{foreground}
\pgfsetlayers{background,main,foreground}

\ifx\techreport\undefined
  \newcommand{\onlytech}[1]{\ignorespaces}
  \newcommand{\onlypaper}[1]{#1}
  \def\thm@space@setup{%
  \thm@preskip=0.75\topsep
  \thm@postskip=\thm@preskip 
  }
\else
  \newcommand{\onlytech}[1]{#1}
  \newcommand{\onlypaper}[1]{\ignorespaces}
  \makeatletter
  \def\@copyrightspace{\relax}
  \makeatother
\fi

\declaretheorem{theorem}
\declaretheorem{lemma}

\declaretheorem{definition}

\newcommand{\mynoindent}{\noindent}

\newcommand\argmax{\operatorname{arg\,max}}

\newcommand{\compilehidecomments}{false}
\ifthenelse{ \equal{\compilehidecomments}{true}}{%
	\newcommand{\wei}[1]{}
	\newcommand{\john}[1]{}
	\newcommand{\yishi}[1]{}
}{
  \newcommand{\wei}[1]{{\color{blue}  [\text{Wei:} #1]}}
  \newcommand{\john}[1]{{\color{green}  [\text{John:} #1]}}
  \newcommand{\yishi}[1]{{\color{red}  [\text{Yishi:} #1]}}
}

\newcommand{\rdelta}[1]{\big\lfloor{#1}\big\rfloor_\delta}
\newcommand{\bigrdelta}[1]{\Big\lfloor{#1}\Big\rfloor_\delta}
\newcommand{\rdeltav}[2]{\Big\lfloor{#1}\Big\rfloor_{#2}}
\newcommand{\indicator}[1]{\mathbb{I} ({#1})}
\newcommand{\edgep}[3]{{}p^{#3}_{#1,#2}}
\newcommand{\udelv}{u\backslash v}

\title{Boosting Information Spread: \\ An Algorithmic Approach}

\author{Yishi Lin, Wei Chen~\IEEEmembership{Member,~IEEE}, John C.S. Lui~\IEEEmembership{Fellow,~IEEE}
  \onlypaper{\thanks{Manuscript received.}}
  \thanks{Yishi Lin and John C.S. Lui are with the Department of Computer
    Science and Engineering, The Chinese University of Hong Kong, Hong Kong. (Emails: \{yslin,cslui\}@cse.cuhk.edu.hk)}
  \thanks{Wei Chen is with Microsoft Research, Beijing, China. (Email: weic@microsoft.com)}
  \onlytech{
  \thanks{This work has been submitted to the IEEE for possible publication. Copyright may be transferred without notice, after which this version may no longer be accessible.}
  }
  \onlypaper{
  \thanks{This research is supported in part by the GRF \#14630815.}
  }
}

\begin{document}

\ifx\techreport\undefined %
\setlength{\textfloatsep}{0.5\textfloatsep}
\setlength{\dbltextfloatsep}{0.5\dbltextfloatsep}
\setlength{\floatsep}{0.5\floatsep} %
\setlength{\dblfloatsep}{0.5\dblfloatsep}
\setlength{\abovedisplayskip}{0.5\abovedisplayskip}
\setlength{\belowdisplayskip}{0.5\belowdisplayskip} %
\captionsetup[figure]{skip=1pt}
\captionsetup[table]{skip=1pt}
\captionsetup[subfloat]{farskip=0pt,captionskip=0pt}
\fi
%
\maketitle

\input{abstract}
\input{introduction}
\input{related}  
\input{problem}
\input{generalgraph}

\input{bidirected}
\input{experiments}
\input{conclusion}

\bibliographystyle{IEEEtran}
\bibliography{mybib}

\ifx\techreport\undefined
\begin{appendix}
\input{appendix_generalgraph}
\input{appendix_bidirected}
\end{appendix}
\else
\appendices
\section{Proofs}
\input{appendix_problem}

\input{appendix_generalgraph}
\input{appendix_bidirected}
\section{DP-Boost for General Bidirected Trees}
\input{appendix_dpboost}
\fi


\end{document}

%% file: abstract.tex
\begin{abstract}

The majority of influence maximization (IM) studies
  focus on targeting influential \textit{seeders}
  to trigger substantial information spread in social networks.
Motivated by the observation that incentives 
  could ``boost'' users so that they are more likely to be influenced by friends,
  we consider a new and complementary \textit{$k$-boosting problem} which aims
  at finding $k$ users to boost so to trigger a maximized ``boosted'' influence spread.
The $k$-boosting problem is \textit{different}
  from the IM problem because 
  boosted users behave differently from seeders:
  \textit{boosted} users are initially uninfluenced and
  we only increase their probability to be influenced.
Our work also \textit{complements} the IM studies 
  because we focus on triggering larger influence spread
  on the basis of given seeders.
Both the NP-hardness of the problem and the non-submodularity of the 
  objective function pose challenges to the $k$-boosting problem.
To tackle the problem on general graphs,
  we devise two efficient algorithms with the 
  \textit{data-dependent} approximation ratio.
To tackle the problem on bidirected trees,
  we present an efficient greedy algorithm 
  and a dynamic programming that is a fully
  polynomial-time approximation scheme.
Extensive experiments using real social networks
  and synthetic bidirected trees verify the efficiency and
  effectiveness of the proposed algorithms.
In particular, on general graphs, we show that
  boosting solutions returned by our algorithms achieves boosts of influence
  that are up to several times higher than those achieved by boosting intuitive
  solutions with no approximation guarantee.
We also explore the ``budget allocation'' problem experimentally,
  demonstrating the beneficial of allocating
  the budget to \textit{both} seeders and boosted users.

\begin{IEEEkeywords}
  Influence maximization, Information boosting, Social networks, Viral marketing
\end{IEEEkeywords}

\end{abstract}

%% file: introduction.tex
\section{Introduction}
\label{sec:introduction}

\IEEEPARstart{W}{ith} the popularity of online social networks,
  \textit{viral marketing}, which is
  a marketing strategy to exploit online \textit{word-of-mouth} effects,
  has become a powerful tool for companies to promote sales.
In viral marketing campaigns,
  companies target influential users
  by offering free products or services
  with the hope of triggering a chain reaction of product adoption.
\textit{Initial adopters} or \textit{seeds} are often used
  interchangeably to refer to these targeted users.
Motivated by the need for effective viral marketing strategies,
  \textit{influence maximization} has become a fundamental research
  problem in the past decade.
The goal of influence maximization is usually to identify influential initial adopters~\cite{kempe2003maximizing,
  carnes2007maximizing,chen2009efficient,chen2010scalable,chen2010scalableLT,
  borgs2014maximizing,tang2014influence,tang2015influence}.

In practical marketing campaigns,
  companies often consider \textit{a mixture of multiple promotion strategies}.
Targeting influential users as initial adopters is one tactic,
  and we list some others as follows.
\begin{itemize}[leftmargin=*]
\item \textit{Customer incentive programs}:
Companies offer incentives such as
  coupons or product trials to attract potential customers.
Targeted customers are in general more likely to be
  influenced by their friends.
\item \textit{Social media advertising}:
Companies reach intended audiences via digital advertising. 
According to an advertising survey~\cite{nielsen2015global},
  owned online channels such as brand-managed sites are the second most trusted
  advertising formats, second only to recommendations from family and friends.
We believe that customers targeted by advertisements
  are more likely to follow their friends' purchases.
\item \textit{Referral marketing}:
Companies encourage customers to refer others to use the product by offering 
  rewards such as cash back.
In this case, targeted customers are more likely to influence their friends.
\end{itemize}
As one can see, these marketing strategies are able to ``{\em boost}'' the
  influence transferring through customers.
Furthermore, for companies, the cost of ``boosting'' a customer
  (e.g., the average redemption and distribution cost per coupon,
  or the advertising cost per customer)
  is much lower than the cost of nurturing an influential user as
  an initial adopter and a product evangelist.
Although identifying influential initial adopters have been actively studied,
  very little attention has been devoted to studying how to
  utilize incentive programs or other strategies
  to further increase the influence spread of initial adopters.

In this paper, we study the problem of finding $k$ boosted users
  so that when their friends adopt a product,
  they are more likely to make the purchase and 
  continue to influence others.
Motivated by the need for modeling boosted customers,
  we propose a novel \textit{influence boosting model}.
In our model,
  \textit{seed users} generate influence same as in the classical \textit{Independent
  Cascade} (IC) model.
In addition, we introduce the \textit{boosted user} as a new type of user.
They represent customers with incentives such as coupons.  
Boosted users are uninfluenced at the beginning of the influence propagation process,
However, they are more likely to be influenced by their friends 
  and further spread the influence to others.
In other words, they ``boost'' the influence transferring through them.
Under the influence boosting model, 
  we study \textit{how to boost the influence spread given initial adopters}.
More precisely, given initial adopters, 
  we are interested in identifying $k$ users among other users,
  so that the expected influence spread upon ``boosting'' them is maximized.
Because of the essential differences in behaviors between seed users and boosted users,
  our work is very different from influence maximization studies
  focusing on selecting seeds.

Our work also complements the studies of influence maximization problems. 
First, compared with nurturing an initial adopter,
  \textit{boosting} a potential customer usually incurs a lower cost.
For example,
  companies may need to offer free products to initial adopters,
  but only need to offer coupons to boost potential customers.
With both our methods that identify users to \textit{boost} and 
  influence maximization algorithms that select initial adopters,
  companies have more flexibility in allocating their marketing budgets.
Second, initial adopters are sometimes predetermined.
For example, they may be advocates of a particular brand or prominent bloggers 
  in the area.
In this case, our study suggests how to
  effectively utilize incentive programs or similar marketing strategies to
  take the influence spread to the next level.

\mynoindent\textbf{Contributions.}
\onlytech{
We study a novel problem of how to \textit{boost} the
  influence spread when the initial adopters are given.
}
We summarize our contributions as follows.

\begin{itemize}[leftmargin=*]
\item
We present the \textit{influence boosting model}, which
  integrates the idea of \textit{boosting} into the 
  \textit{Independent Cascade} model.
We formulate a \textit{$k$-boosting} problem that asks how to maximize the
  boost of influence spread under the \textit{influence boosting model}.
The $k$-boosting problem is NP-hard.
Computing the boost of influence spread is \#P-hard.
Moreover, the boost of influence spread does not possess the submodularity,
  meaning that the greedy algorithm does not provide performance guarantee.

\item
We present approximation algorithms \texttt{PRR-Boost} and
  \texttt{PRR-Boost-LB} for the $k$-boosting problem.
For the $k$-boosting problem on bidirected trees,
  we present a greedy algorithm \texttt{Greedy-Boost} based on
  a linear-time exact computation of the boost of influence spread
  and a fully polynomial-time approximation scheme (FPTAS)
  \texttt{DP-Boost} that returns near-optimal solutions.~\footnote{An FPTAS for a maximization problem is
    an algorithm that given any $\epsilon>0$, 
	it can approximate the optimal solution with a factor $1-\epsilon$,
  with running time polynomial to the input size and $1/\epsilon$.}
\texttt{DP-Boost} provides a benchmark for the greedy algorithm,
  at least on bi-directed trees,
  since it is very hard to find near optimal solutions in general cases.
Moreover, the algorithms on bidirected trees may be applicable to situations
where information cascades more or less follow a fixed tree architecture.

\item
We conduct extensive experiments using real social networks and synthetic bidirected trees. 
Experimental results show the efficiency and effectiveness of our proposed algorithms,
  and their superiority over intuitive baselines.
\end{itemize}

\mynoindent\textbf{Paper organization.}
\Cref{sec:background} provides background.
We describe the \textit{influence boosting model} and
  the $k$-boosting problem in \Cref{sec:problem_definition}.
We present building blocks of
  \texttt{PRR-Boost} and \texttt{PRR-Boost-LB} for the $k$-boosting problem
  in \Cref{sec:general},
  and the detailed algorithm design in \Cref{sec:algorithm}.
We present \texttt{Greedy-Boost} and \texttt{DP-Boost} for the $k$-boosting
  problem on bidirected trees in \Cref{sec:bidirected}.
We show experimental results in \Crefrange{sec:experiments}{sec:experiments_bid}.
\Cref{sec:conclusion} concludes the paper.

%% file: related.tex
\section{Background and related work}\label{sec:background}

In this section, we provide backgrounds about
  influence maximization problems and related works.

\mynoindent\textbf{Classical influence maximization problems.}
Kempe et al.~\cite{kempe2003maximizing} first formulated the \textit{influence
  maximization} problem that asks to select a set $S$ of $k$ nodes
  so that the expected influence spread is maximized under a predetermined
  influence propagation model.
The \textit{Independent Cascade} (IC) model is one classical model that
  describes the influence diffusion process~\cite{kempe2003maximizing}.
Under the IC model, given a graph $G=(V,E)$, influence probabilities
  on edges and a set $S\subseteq V$ of \textit{seeds},
  the influence propagates as follows.
Initially, nodes in $S$ are activated.
Each newly activated node $u$ influences its neighbor $v$ with probability $p_{uv}$.
The influence spread of $S$ is the expected
  number of nodes activated at the end of the influence diffusion process.
Under the IC model,
  the influence maximization problem is NP-hard~\cite{kempe2003maximizing}
  and computing the expected influence spread
  for a given $S$ is \#P-hard~\cite{chen2010scalable}.
A series of studies have been done to approximate the influence
  maximization problem under the IC model or other diffusion
  models~\cite{leskovec2007cost,chen2009efficient,chen2010scalable,
  goyal2011celf++,chen2013information,
  borgs2014maximizing,tang2014influence,tang2015influence,nguyen2016stopstare}.

\mynoindent\textbf{Influence maximization on trees.}
Under the IC model,
  tree structure makes the influence computation tractable.
To devise greedy ``seed-selection'' algorithms on trees,
  several studies presented various methods to compute
  the ``marginal gain'' of influence spread on trees~\cite{chen2010scalable,jung2012irie}.
Our computation of ``marginal gain of boosts'' on trees
  is more advanced than the previous methods:
It runs in linear-time,
  it considers the behavior of ``boosting'', and we assume that the benefits of
  ``boosting'' can be transmitted in both directions of an edge.
On bidirected trees, Bharathi et al.~\cite{bharathi2007competitive} described 
 an FPTAS for the classical influence maximization problem.
Our FPTAS on bidirected trees is different from theirs
  because ``boosting'' a node and targeting a node as a ``seed''
  have significantly different effects.

\mynoindent\textbf{Boost the influence spread.}
Several works studied how to recommend friends or inject links into social
  networks in order to boost the influence spread~\cite{chaoji2012recommendations,
    yang2013maximizing,antaris2014link,
    rafailidis2014little, rafailidis2015crossing, yang2016continuous}.
Lu et al.~\cite{lu2015competition} studied how to maximize
  the expected number of adoptions by targeting 
  initial adopters of a complementing product.
Chen et al.~\cite{chen2015combining} considered the \textit{amphibious influence
  maximization}.
They studied how to select a subset of seed content providers and a subset of seed
  customers so that the expected number of influenced customers is maximized.
Their model differs from ours in that they only consider influence originators selected
	from content providers, which are separated from the social network, and influence boost
	is only from content providers to consumers in the social network.
Yang et al.~\cite{yang2016continuous} studied how to offer discounts
  assuming that the probability of a customer being an
  initial adopter is a known function of the discounts offered to him.
They studied how to offer discounts to customers so 
  that the influence cascades triggered is maximized.
Different from the above studies, 
  we study how to \textit{boost the spread of influence} when seeds are given.
We assume that we can give incentives to some users (i.e., ``boost'' some users)
  so that they are more likely to be influenced by their friends, but they themselves would not
  become adopters without friend influence.
This article is an extended version of our conference paper~\cite{lin2017boosting} that 
  formulated the $k$-boosting problem and presented algorithms for it.
We add two new algorithms that tackle the $k$-boosting problem in bidirected
trees, and report new experimental results.

%% file: problem.tex
\section{Model and Problem Definition}\label{sec:problem_definition}

In this section, we first define the \textit{influence boosting model}
and the $k$-boosting problem.
Then, we highlight the challenges associated with solving the proposed problem.

\subsection{Model and Problem Definition}

Traditional studies of the influence maximization problem focus on how to
  identify a set of $k$ influential users (or \textit{seeds})
  who can trigger the largest influence diffusion.
In this paper,
  we aim to \textit{boost} the influence propagation
  assuming that \textit{seeds} are given.
We first define the \textit{influence boosting model}.

\begin{definition}[Influence Boosting Model]\label{def:boost_model}
Suppose we are given a directed graph $G\!=\!(V,E)$ with $n$ nodes and $m$ edges,
  two influence probabilities $p_{uv}$ and $p'_{uv}$ (with $p'_{uv} > p_{uv}$) on each edge $e_{uv}$,
  a set $S\subseteq V$ of seeds,
  and a set $B\subseteq V$ of boosted nodes.
Influence propagates in discrete time steps as follows.
If $v$ is not boosted,
  each of its newly-activated in-neighbor $u$ influences $v$ with probability $p_{uv}$.
If $v$ is a boosted node,
  each of its newly-activated in-neighbor $u$ influences $v$ with probability
  $p'_{uv}$.
\end{definition}

In \Cref{def:boost_model}, we assume that ``boosted'' users are
  more likely to be influenced.
Our study can also be adapted to the case where boosted users
  are more influential: if a newly-activated user $u$ is boosted,
  she influences her neighbor $v$ with probability $p'_{uv}$ instead of $p_{uv}$.
To simplify the presentation, we focus on the influence boosting
  model in \Cref{def:boost_model}.

\begin{figure}[ht]
\centering
\begin{minipage}[c]{0.24\textwidth}
\centering
\begin{tikzpicture}[
  baseline={-0.15in}, thick, node distance = 1cm and 1.2cm]

  \node[seed node] (s) {$s$};
  \node[main node] (0) [right =of s] {$v_0$};
  \node[main node] (1) [right =of 0] {$v_1$};

  \draw [->] (s) -- (0)
    node [midway, above=5mm] {\scriptsize $p_{s,{v_0}}\!=\!0.2$}
    node [midway, above=1mm] {\scriptsize $p_{s,{v_0}}'\!=\!0.4$};%
  \draw [->] (0) -- (1)
    node [midway, above=5mm] {\scriptsize $p_{v_0,v_1}\!=\!0.1$}
    node [midway, above=1mm] {\scriptsize $p_{v_0,v_1}'\!=\!0.2$};%

\end{tikzpicture}
\end{minipage}
\begin{minipage}[c]{0.24\textwidth}%
\scriptsize
\centering
\renewcommand{\arraystretch}{1.2}
\begin{tabular}[c]{|l|c|c|}
\hline
  $\boldsymbol{B}$ & $\boldsymbol{\sigma_S(B)}$ & $\boldsymbol{\Delta_S(B)}$ \\ \hline
$\emptyset$        & 1.22                       & 0.00                       \\ \hline
$\{v_0\}$          & 1.44                       & 0.22                       \\ \hline
$\{v_1\}$          & 1.24                       & 0.02                       \\ \hline
$\{v_0,v_1\}$      & 1.48                       & 0.26                       \\ \hline
\end{tabular}
\end{minipage}
\caption{Example of the influence boosting model ($S\!=\!\{s\}$).}
\label{fig:line3}
\end{figure}
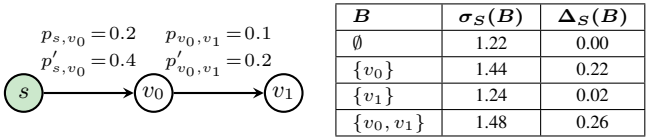

Let $\sigma_{S}(B)$ be the expected influence spread of $S$
  upon boosting nodes in $B$.
We refer to $\sigma_{S}(B)$ as the \textit{boosted influence spread}.
Let $\Delta_{S}(B)=\sigma_{S}(B)-\sigma_{S}(\emptyset)$.
We refer to $\Delta_{S}(B)$ as the \textit{boost of influence spread
  of $B$}, or simply the \textit{boost of $B$}.
Consider the example in \Cref{fig:line3}.
We have $\sigma_{S}(\emptyset)=1.22$,
  which is essentially the influence spread of $S$ in the
  IC model.
When we boost node $v_0$,
  we have $\sigma_{S}(\{v_0\})=1+0.4+0.04=1.44$, and
  $\Delta_{S}(\{v_0\})=0.22$.
We now formulate the \textit{$k$-boosting problem}.

\begin{definition} [$k$-Boosting Problem] 
Given a directed graph $G=(V,E)$, 
  influence probabilities $p_{uv}$ and $p'_{uv}$ on every edges $e_{uv}$,
  and a set $S\subseteq V$ of seed nodes,
  find a boost set $B\subseteq V$ with $k$ nodes,
  such that the boost of influence spread of $B$ is maximized.
That is, determine $B^*\subseteq V$ such that
\begin{align}
  B^*=\argmax _{B\subseteq V,|B|\leq k} \Delta_{S}(B).\label{eq:k_boosting_def}
\end{align}
\end{definition}

By definition, the $k$-boosting problem is very different 
  from the classical influence maximization problem.
In addition, boosting nodes that significantly increase the
  influence spread when used as additional seeds could be 
  extremely inefficient.
For example, consider the example in \Cref{fig:line3},
  if we are allowed to select one more seed, we should select $v_1$.
However, if we can boost a node, boosting $v_0$ is much better than boosting $v_1$.
\Cref{sec:experiments} provides more experimental results.

\subsection{Challenges of the Boosting Problem}

We now provide key properties of the $k$-boosting problem
  and show the challenges we face.
\Cref{th:hardness} summarizes the hardness of the $k$-boosting problem.

\begin{restatable}[Hardness]{theorem}{hardness}\label{th:hardness}
	The $k$-boosting problem is NP-hard.
    Computing $\Delta_{S}(B)$ given $S$ and $B$
	is \#P-hard.
\end{restatable}

\begin{proof} 
The NP-hardness is proved by a reduction from the NP-complete
  \textit{Set Cover} problem~\cite{karp1972reducibility}.
The \#P-hardness of the computation is proved by a reduction from the
  \#P-complete counting problem of \textit{$s$-$t$ connectedness} in directed
  graphs~\cite{valiant1979complexity}.
\onlytech{
The full analysis can be found in the appendix.
}
\onlypaper{
The full analysis is in our technical report~\cite{lin2016boosttechreport}.
}
\end{proof}

\mynoindent\textbf{Non-submodularity of the boost of influence.}
Because of the above hardness results, 
  we explore approximation algorithms to tackle the $k$-boosting problem.
In most influence maximization problems,
  the expected influence of the seed set $S$ (i.e., the objective function)
  is a monotone and submodular function of $S$.\footnote{
A set function $f$ is monotone if $f(S) \le f(T)$ for all $S\subseteq T$;
	it is submodular if $f(S\cup \{v\}) - f(S) \ge f(T\cup \{v\}) - f(T)$ for all $S\subseteq T$ and $v\not\in T$,
	and it is supermodular if $-f$ is submodular.}
Thus, a natural greedy algorithm returns a solution with an approximation
  guarantee~\cite{kempe2003maximizing,
  borgs2014maximizing, tang2014influence,
  tang2015influence, nguyen2016stopstare, lin2015analyzing}.
However, the objective function $\Delta_{S}(B)$ in our problem is
  \textit{neither submodular nor supermodular} on the set $B$ of boosted nodes.
On one hand, when we boost several nodes on different parallel paths from
  seed nodes, their overall boosting effect exhibits a \textit{submodular behavior}.
On the other hand, when we boost several nodes on a path starting from a
  seed node, their boosting effects can be cumulated along the path, generating
  a larger overall effect than the sum of their individual boosting effect.
This is in fact a \textit{supermodular behavior}.
\onlytech{
To illustrate, consider the graph in \Cref{fig:line3},
  we have
  $\Delta_{S}(\{v_0,v_1\})-\Delta_{S}(\{v_0\})=0.04$,
  which is larger than
  $\Delta_{S}(\{v_1\})-\Delta_{S}(\emptyset)=0.02$.
In general, the boosted influence has a complicated interaction between
	supermodular and submodular behaviors when the boost set grows, and
	is neither supermodular nor submodular.
}
The non-submodularity of $\Delta_{S}(\cdot)$
  indicates that the boosting set returned by the greedy
  algorithm may not have the $(1-1/e)$-approximation guarantee.
Therefore,
  the non-submodularity of the objective function poses an additional challenge.

%% file: generalgraph.tex
\section{Boosting on General Graphs: Building Blocks}\label{sec:general}

In this section,
  we present three building blocks for solving the $k$-boosting problem:
  (1) a state-of-the-art influence maximization framework,
  (2) the \textit{Potentially Reverse Reachable Graph}
  for estimating the boost of influence spread,
  and (3) the \textit{Sandwich Approximation} strategy~\cite{lu2015competition}
  for maximizing non-submodular functions.
Our algorithms \texttt{PRR-Boost} and \texttt{PRR-Boost-LB} integrate
the three building blocks.
We will present their detailed algorithm design in the next section.

\subsection{State-of-the-art influence maximization techniques}
One state-of-the-art influence maximization framework is the \textit{Influence 
  Maximization via Martingale} (\texttt{IMM}) method~\cite{tang2015influence}
  based on the idea of \textit{Reverse-Reachable Sets}
  (RR-sets)~\cite{borgs2014maximizing}.
We utilize the \texttt{IMM} method in this work,
  but other similar frameworks based on RR-sets (e.g.,
  \texttt{SSA}/\texttt{D-SSA}~\cite{nguyen2016stopstare}) could also be applied.

\mynoindent\textbf{RR-sets.}
An \textit{RR-set} for a node $r$ is a random set $R$ of nodes,
  such that for any seed set $S$,
  the probability that $R\cap S\neq \emptyset$ equals the probability
  that $r$ can be activated by $S$ in a random diffusion process.
Node $r$ may also be selected uniformly at random from $V$, and the
  RR-set will be generated accordingly with $r$.
One key property of RR-sets is that 
  the expected influence of $S$ equals to
  $n\cdot \mathbb{E}[\indicator{R\cap S\neq \emptyset}]$ for all
  $S\subseteq V$,
  where $\indicator{\cdot}$ is the indicator function
  and the expectation is taken over the randomness of $R$.

\mynoindent\textbf{General \texttt{IMM} algorithm.}
The \texttt{IMM} algorithm has two phases.
The \textit{sampling} phase generates a sufficiently large number of random RR-sets
  such that the estimation of the influence spread is ``accurate enough''.
The \textit{node selection} phase greedily selects $k$ seed nodes
  based on their estimated influence spread.
If generating a random RR-set takes time $O(EPT)$,
  \texttt{IMM} returns a $(1-1/e-\epsilon)$-approximate solution with
  probability at least $1-n^{-\ell}$, and runs in
  $O(\frac{EPT}{OPT}\cdot(k+\ell)(n+m)\log n/\epsilon^2)$ expected time,
  where $OPT$ is the optimal expected influence.

\subsection{Potentially Reverse Reachable Graphs}
We now describe how we estimate the boost of influence.
The estimation is based on the concept of the \textit{Potentially Reverse
  Reachable Graph} (PRR-graph) defined as follows.

\begin{definition}[Potentially Reverse Reachable Graph]
Let $r$ be a node in $G$.
  A Potentially Reverse Reachable Graph (PRR-graph) $R$ for a node $r$
  is a random graph generated as follows.
We first sample a deterministic copy $g$ of $G$.
In the deterministic graph $g$,  
  each edge $e_{uv}$ in graph $G$ is
  ``live'' in $g$ with probability $p_{uv}$,
  ``live-upon-boost'' with probability $p'_{uv}-p_{uv}$,
  and ``blocked'' with probability $1-p'_{uv}$.
The PRR-graph $R$ is the minimum subgraph of $g$ containing all paths from
  seed nodes to $r$ through non-blocked edges in $g$.
We refer to $r$ as the ``root node''.
When $r$ is also selected from $V$ uniformly at random,
    we simply refer to the generated PRR-graph as a random PRR-graph (for a random root).
\end{definition}

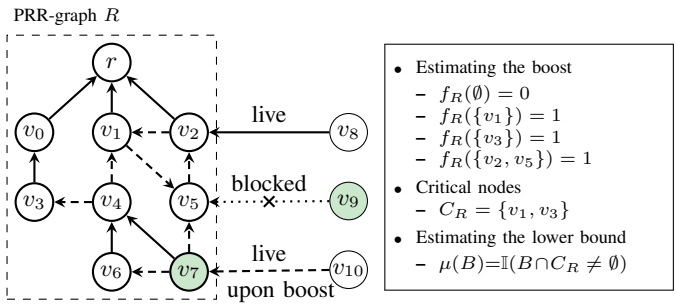
\begin{figure}[htb]
\centering
\begin{tikzpicture}[->, node distance=0.4cm and 0.5cm,
decoration={
  markings, mark= at position 0.5 with
  {\draw [solid, -](-2pt,-2pt) -- (2pt,2pt);
   \draw [solid, -](2pt,-2pt) -- (-2pt,2pt);}
}
]

  \node[main node, thick] (r) {$r$};
  \node[main node, thick] (1) [below =of r] {$v_1$};
  \node[main node, thick] (4) [below =of 1] {$v_4$};
  \node[main node, thick] (6) [below =of 4] {$v_6$};
  \node[main node, thick] (0) [left =of 1] {$v_0$};
  \node[main node, thick] (3) [below =of 0] {$v_3$};
  \node[main node, thick] (2) [right =of 1] {$v_2$};
  \node[main node, thick] (5) [below =of 2] {$v_5$};
  \node[seed node, thick] (7) [below =of 5] {$v_7$};
  \node[draw, fit=(r) (0) (7),
  inner sep=0.1cm, 
  dashed,
  label={[xshift=-0.6cm, yshift=0cm] {\scriptsize PRR-graph $R$}}] (R) {};

  \node[main node] (8) [right =1.6cm of 2] {$v_8$};
  \node[seed node] (9) [below =of 8] {$v_9$};
  \node[main node] (10) [below =of 9] {$v_{10}$};

  \draw[thick] (0) -- (r);
  \draw[thick] (1) -- (r);
  \draw[thick] (2) -- (r);
  \draw[thick] (3) -- (0);
  \draw[thick] (6) -- (4);
  \draw[thick] (7) -- (4);
  \draw[thick] (8) -- (2) node [midway, above] {live};

  \draw[thick, densely dashed] (2) -- (1);
  \draw[thick, densely dashed] (4) -- (1);
  \draw[thick, densely dashed] (5) -- (2);
  \draw[thick, densely dashed] (4) -- (3);
  \draw[thick, densely dashed] (1) -- (5);
  \draw[thick, densely dashed] (7) -- (5);
  \draw[thick, densely dashed] (7) -- (6);
  \draw[thick, densely dashed] (10) -- (7)
  node [midway, above] {live}
  node [midway, below] {\quad upon boost};

  \draw[thick, dotted, postaction={decorate}]
  (9) -- (5) node [midway, above] {blocked};

  \node [right =2.1cm of R] {\framebox{\parbox{3.5cm}{
      \scriptsize
      \baselineskip=3pt
      \begin{itemize}[noitemsep, topsep=3pt, leftmargin=*]
      \item  Estimating the boost
      \begin{itemize}[noitemsep, topsep=3pt, leftmargin=*]
        \item $f_R(\emptyset)=0$\\
        \item $f_R(\{v_1\})=1$\\
        \item $f_R(\{v_3\})=1$\\
        \item $f_R(\{v_2,v_5\})=1$\\
        \end{itemize}
      \item Critical nodes
        \begin{itemize}[noitemsep, topsep=3pt, leftmargin=*]
        \item 
          $C_R=\{v_1,v_3\}$\\
        \end{itemize}
      \item Estimating the lower bound
        \begin{itemize}[noitemsep, topsep=3pt, leftmargin=*]
        \item $\mu(B){=}\indicator{B\!\cap\! C_R\neq \emptyset}$
        \end{itemize}
      \end{itemize}
      \par
      }
     }};
\end{tikzpicture}
\caption{Example of a Potentially Reverse Reachable Graph.}
\label{fig:prr_graph}
\end{figure}

\Cref{fig:prr_graph} shows an example of a PRR-graph $R$.
\onlytech{
The directed graph $G$ contains $12$ nodes and $16$ edges.
}
Node $r$ is the \textit{root} node.
Shaded nodes are seed nodes.
Solid, dashed and dotted arrows with crosses represent live, live-upon-boost
  and blocked edges, respectively.
The PRR-graph for $r$ is the subgraph in the dashed box.
\onlytech{
It contains $9$ nodes and $13$ edges.
}
Nodes and edges outside the dashed box do not belong to the PRR-graph,
  because they are not on any paths from seed nodes to $r$ that only
  contain non-blocked edges.
By definition, a PRR-graph may contain loops.
For example, the PRR-graph in \Cref{fig:prr_graph} contains a loop
  among nodes $v_1$, $v_5$, and $v_2$.

\mynoindent\textbf{Estimating the boost of influence.}
Let $R$ be a given PRR-graph with root $r$.
By definition, every edge in $R$ is either \textit{live} or \textit{live-upon-boost}.
Given a path in $R$, 
  we say that it is \textit{live} if and only if it contains only live
  edges.
Given a path in $R$ and a set of boosted nodes $B\subseteq V$,
  we say that the path is \textit{live upon boosting $B$}
  if and only if the path is not a \textit{live} one,
  but every edge $e_{uv}$ on it is either \textit{live} or
  \textit{live-upon-boost} with $v\! \in\! B$.
For example, in \Cref{fig:prr_graph},
  the path from $v_3$ to $r$ is \textit{live},
  and the path from $v_7$ to $r$ via $v_4$ and $v_1$ is
  \textit{live upon boosting $\{v_1\}$}.
Define $f_R(B):2^V\!\rightarrow\!\{0,1\}$ as:
    $f_R(B)=1$ if and only if, in $R$,
  (1) there is no \textit{live} path from seed nodes to $r$; and
  (2) a path from a seed node to $r$ is \textit{live upon boosting $B$}.
\onlytech{
Intuitively, in the deterministic graph $R$,
  $f_R(B)=1$ if and only if the root node is inactive without boosting,
  and active upon boosting nodes in $B$.
}
In \Cref{fig:prr_graph}, 
  if $B=\emptyset$,
  there is no live path from the seed node $v_7$ to $r$ upon boosting $B$.
Therefore, we have $f_R(\emptyset)=0$.
Suppose we boost a single node $v_1$.
There is a live path from the seed node $v_7$ to $r$ that is live upon boosting
$\{v_1\}$, and thus we have $f_R(\{v_1\})=1$.
Similarly, we have $f_R(\{v_3\})=f_R(\{v_2,v_5\})=1$.
Based on the above definition of $f_R(\cdot)$, we have the following lemma.
\begin{restatable}{lemma}{lemmaprrdelta}\label{lemma:prr_delta}
  For any $B\subseteq V$, we have
  $n\cdot \mathbb{E}[f_R(B)]=\Delta_{S}(B)$,
  where the expectation is taken over the randomness of $R$.
\end{restatable}

\begin{proof}
For a random PRR-graph $R$ whose root node is randomly selected,
  $\Pr[f_R(B)=1]$ equals the difference between probabilities that a random node
  in $G$ is activated given that we boost $B$ and $\emptyset$.
\end{proof}

\noindent Let $\mathcal{R}$ be a set of independent random PRR-graphs, define
\begin{align}
\hat{\Delta}_{\mathcal{R}}(B) =\frac{n}{|\mathcal{R}|}\cdot \sum_{R\in\mathcal{R}}f_R(B),\forall B\subseteq V.
\end{align}
By Chernoff bound,
  $\hat{\Delta}_{\mathcal{R}}(B)$ closely estimates $\Delta_{S}(B)$
  for any $B\subseteq V$ if $|\mathcal{R}|$ is sufficiently large.

\subsection{Sandwich Approximation Strategy}
To tackle the non-submodularity of function $\Delta_{S}(\cdot)$,
  we apply the \textit{Sandwich Approximation} (SA)
  strategy~\cite{lu2015competition}.
First, we find submodular lower and upper bound functions of
  $\Delta_{S}$, denoted by $\mu$ and $\nu$.
Then, we select node sets $B_\Delta$, $B_\mu$ and $B_\nu$ by greedily maximizing
  $\Delta_{S}$, $\mu$ and $\nu$ under the cardinality constraint of $k$.
Ideally, we return
  $B_\text{sa}=\argmax_{B\in \{B_\mu,B_\nu,B_\Delta\}}\Delta_{S}(B)$
  as the final solution.
Let the optimal solution of the $k$-boosting problem be $B^{*}$
  and let $OPT=\Delta_{S}(B^{*})$.
Suppose $B_\mu$ and $B_\nu$ are $(1-1/e-\epsilon)$-approximate solutions for
  maximizing $\mu$ and $\nu$, we have
\begin{align}
\Delta_{S}(B_\text{sa}) &\geq \frac{\mu(B^{*})}{\Delta_{S}(B^{*})}
  \cdot (1-1/e-\epsilon) \cdot OPT, \label[ineq]{eq:sa_lb} \\
\Delta_{S}(B_\text{sa}) &\geq \frac{\Delta_{S}(B_\nu)}{\nu(B_\nu)}
  \cdot (1-1/e-\epsilon) \cdot OPT.
\end{align}
Thus, to obtain a good approximation guarantee,
  at least one of $\mu$ and $\nu$ should be close to $\Delta_{S}$.
In this work,
  we derive a submodular lower bound $\mu$ of $\Delta_{S}$
  using the definition of PRR-graphs.
Because $\mu$ is significantly closer to
  $\Delta_{S}$ than any submodular upper bound we have tested,
  we only use the lower bound function $\mu$ and the ``lower-bound side''
  of the SA strategy with approximation guarantee in \Cref{eq:sa_lb}.

\mynoindent\textbf{Submodular lower bound.}
We now derive a submodular lower bound of $\Delta_{S}$.
Let $R$ be a PRR-graph with the root node $r$.
Let $C_R=\{v|f_R(\{v\})=1\}$.
We refer to nodes in $C_R$ as \textit{critical nodes} of $R$.
Intuitively,
  the root node $r$ becomes activated if we boost any node in $C_R$.
For any node set $B\subseteq V$,
  define $f_R^{-}(B)=\indicator{B\cap C_R\neq\emptyset}$.
By definition of $C_R$ and $f_R^{-}(\cdot)$,
  we have $f^{-}_R(B)\leq f_R(B)$ for all $B\subseteq V$.
Moreover, because the value of $f_R^{-}(B)$ is based on whether the node set $B$
  intersects with a fixed set $C_R$,
  $f^{-}_R(B)$ is a submodular function on $B$.
For any $B\subseteq V$, define
  $\mu(B)=n\cdot \mathbb{E}[f^{-}_R(B)]$
  where the expectation is taken over the randomness of $R$.
\Cref{lemma:submodularlb} shows the properties of the function $\mu$.

\begin{restatable}{lemma}{lemmasubmodularlb}\label{lemma:submodularlb}
  We have $\mu(B)\leq \Delta_{S}(B)$ for all $B\subseteq V$.
  Moreover, $\mu(B)$ is a submodular function of $B$.
\end{restatable}

\begin{proof}
For all $B\subseteq V$,
  we have $\mu(B)\leq \Delta_{S}(B)$
  because we have $f^{-}_R(B)\leq f_R(B)$ for any PRR-graph $R$.
Moreover, $\mu(B)$ is submodular on $B$
  because $f^{-}_R(B)$ is submodular on $B$ for any PRR-graph $R$.
\end{proof}

Our experiments show that $\mu$ is close to $\Delta_{S}$
  especially for small $k$ (e.g., less than a thousand).
Define
\begin{align*}
  \hat{\mu}_{\mathcal{R}}(B)=\frac{n}{|\mathcal{R}|}\cdot
  \sum_{R\in\mathcal{R}}f^{-}_R(B),\forall B\subseteq V.
\end{align*}
Because $f_R^{-}(B)$ is submodular on $B$ for any PRR-graph $R$,
  $\hat{\mu}_{\mathcal{R}}(B)$ is submodular on $B$.
Moreover, by Chernoff bound, $\hat{\mu}_{\mathcal{R}}(B)$ is close to
  $\mu(B)$ when $|\mathcal{R}|$ is sufficiently large.

\onlytech{
\mynoindent\textbf{Remarks on the lower bound function $\mu(B)$.}
The lower bound function
     $\mu(B)$ does correspond to a physical diffusion model, as we now explain.
Roughly speaking, $\mu(B)$ is the influence spread in a diffusion model
    with the boost set $B$, and the
    constraint that at most one edge on
    the influence path from a seed node to an activated node can be boosted.
More precisely, on every edge $e_{uv}$ with $v\in B$,
    there are three possible outcomes when $u$
    tries to activate $v$:
    (a) \textit{normal activation}:
    $u$ successfully activates $v$ without relying on the boost of $v$
    (with probability $p_{uv}$),
    (b) \textit{activation by boosting}: $u$ successfully activates $v$ but relying on the boost of $v$
    (with probability $p'_{uv}-p_{uv}$); and
    (c) \textit{no activation}: $u$ fails to activate $v$ (with probability $1-p'_{uv}$).
In the diffusion model, each activated node records whether it is normally activated
    or activated by boosting.
Initially, all seed nodes are normally activated.
If a node $u$ is activated by boosting, we disallow $u$
    to activate its out-neighbors by boosting.
Moreover, when $u$ normally activates $v$, $v$ inherits the status from $u$
    and records its status as activated by boosting.
However, if later $u$ can be activated again by another in-neighbor normally, $u$ can resume
    the status of being normally activated, resume trying to activate its out-neighbors
    by boosting.
Furthermore, this status change recursively propagates to $u$'s out-neighbors that
    were normally activated by $u$ and inherited the ``activated-by-boosting'' status
    from $u$, so that they now have the status of ``normally-activated'' and
    can activate their
    out-neighbors by boosting.
All the above mechanisms are to insure that the chain of activation from any seed to any
    activated node uses at most one activation by boosting.

Admittedly, the above model is convoluted, while the PRR-graph description
    of $\mu(B)$ is more direct and is easier to analyze.
Indeed, we derived the lower bound function $\mu(B)$ directly from the concept
  of PRR-graphs first,
    and then ``reverse-engineered'' the above model from the PRR-graph model.
Our insight is that by fixing the randomness in the original influence diffusion
    model, it may be easier to
    derive submodular lower-bound or upper-bound functions.
Nevertheless, we believe the above model also provides some intuitive understanding
    of the lower bound model --- it is precisely submodular because it disallows
    multiple activations by boosting in any chain of activation sequences.
}

\onlypaper{
\mynoindent\textbf{Remarks on function $\mu(B)$.}
Function $\mu(B)$ does correspond to some physical diffusion model.
Roughly speaking, $\mu(B)$ is the influence spread in a diffusion model
    with the boost set $B$, and the constraint that at most one edge on
    the influence path from a seed node to an activated node can be boosted.
Compared with the convoluted diffusion model corresponding to $\mu(B)$,
  the PRR-graph description of $\mu(B)$ is more direct and is easier to analyze.
Our insight is that by fixing the randomness in the original diffusion
model, it may be easier to
    derive submodular lower-bound or upper-bound functions.
}

\section{Boosting On General Graphs: Algorithm Design}\label{sec:algorithm}

In this section, we first present how we generate random PRR-graphs.
Then we obtain overall algorithms by
  integrating the general \texttt{IMM} algorithm with PRR-graphs and the
  \textit{Sandwich Approximation} strategy.

\subsection{Generating PRR-graphs}

We classify PRR-graphs into three categories.
Let $R$ be a PRR-graph with root node $r$.
(1) \textit{Activated}: If there is a \textit{live} path
    from a seed node to $r$;
(2) \textit{Hopeless}:
There is no path from seeds to $r$ with at most $k$
  \textit{non-live} edges;
(3) \textit{``Boostable''}: not the above two categories.
If $R$ is not \textit{boostable} (i.e. case (1) or (2)),
  we have $f_R(B)=f_R^-(B)=0$ for all $B\subseteq V$.
Therefore, for ``non-boostable'' PRR-graphs,
  we only count their occurrences
  and we terminate the generation of them once we know they are not \textit{boostable}.
\Cref{algo:generation_prr} depicts generation of a random PRR-graph in two phases.
The first phase (\Crefrange{algo:prr_gen_start}{algo:prr_gen_end})
  generates a PRR-graph $R$.
If $R$ is boostable, the second phase compresses $R$ to reduce its size.
\Cref{fig:prr_graph_gen} shows the results of two phases,
  given that the status sampled for every edge is same as that
  in \Cref{fig:prr_graph}.

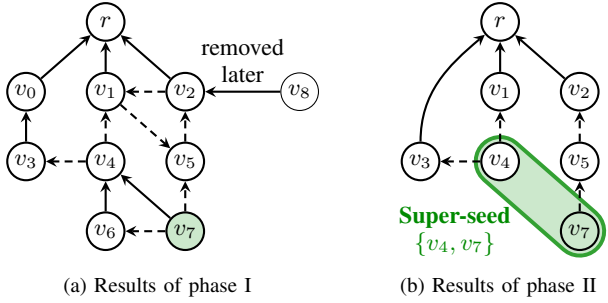
\begin{figure}[ht]
\centering
\subfloat[Results of phase I]{
\label{fig:prr_graph_1}
\begin{tikzpicture}[baseline, ->, node distance=0.4cm and 0.53cm,
decoration={
  markings, mark= at position 0.5 with
  {\draw [solid, -](-2pt,-2pt) -- (2pt,2pt);
   \draw [solid, -](2pt,-2pt) -- (-2pt,2pt);}
}
]

  \node[main node, thick] (r) {$r$};
  \node[main node, thick] (1) [below =of r] {$v_1$};
  \node[main node, thick] (4) [below =of 1] {$v_4$};
  \node[main node, thick] (6) [below =of 4] {$v_6$};
  \node[main node, thick] (0) [left =of 1] {$v_0$};
  \node[main node, thick] (3) [below =of 0] {$v_3$};
  \node[main node, thick] (2) [right =of 1] {$v_2$};
  \node[main node, thick] (5) [below =of 2] {$v_5$};
  \node[seed node, thick] (7) [below =of 5] {$v_7$};

  \node[main node] (8) [right =1cm of 2] {$v_8$};

  \draw[thick] (0) -- (r);
  \draw[thick] (1) -- (r);
  \draw[thick] (2) -- (r);
  \draw[thick] (3) -- (0);
  \draw[thick] (6) -- (4);
  \draw[thick] (7) -- (4);
  \draw[thick] (8) -- (2)
  node [midway, above, align=center] {removed\\ later};

  \draw[thick, densely dashed] (2) -- (1);
  \draw[thick, densely dashed] (4) -- (1);
  \draw[thick, densely dashed] (5) -- (2);
  \draw[thick, densely dashed] (4) -- (3);
  \draw[thick, densely dashed] (1) -- (5);
  \draw[thick, densely dashed] (7) -- (5);
  \draw[thick, densely dashed] (7) -- (6);

  \begin{pgfonlayer}{background}
    \filldraw[draw=white, -, rounded corners=2em, line width=2em, 
          cap=round, opacity=0] (4.center) -- (7.center);
    \filldraw[draw=white, -, rounded corners=2em, line width=1.7em, 
          cap=round, opacity=0] (4.center) -- (7.center);
    \node[white, align=center, opacity=0] (super) [left =of 7] {\textbf{Super-seed}\\ $\{v_4,v_7\}$};
  \end{pgfonlayer}
\end{tikzpicture}
}\quad\quad
\subfloat[Results of phase II]{
\label{fig:prr_graph_2}
\begin{tikzpicture}[baseline, ->, node distance=0.4cm and 0.53cm,
decoration={
  markings, mark= at position 0.5 with
  {\draw [solid, -](-2pt,-2pt) -- (2pt,2pt);
   \draw [solid, -](2pt,-2pt) -- (-2pt,2pt);}
}
]

  \node[main node, thick] (r) {$r$};
  \node[main node, thick] (1) [below =of r] {$v_1$};
  \node[main node, thick] (4) [below =of 1] {$v_4$};
  \node[main node, thick] (3) [left =of 4] {$v_3$};
  \node[main node, thick] (2) [right =of 1] {$v_2$};
  \node[main node, thick] (5) [below =of 2] {$v_5$};
  \node[seed node, thick] (7) [below =of 5] {$v_7$};

  \draw[thick] (1) -- (r);
  \draw[thick] (2) -- (r);
  \draw[thick] (3) to [out=90,in=220] (r);

  \draw[thick, densely dashed] (4) -- (1);
  \draw[thick, densely dashed] (5) -- (2);
  \draw[thick, densely dashed] (4) -- (3);
  \draw[thick, densely dashed] (7) -- (5);

  \begin{pgfonlayer}{background}
    \filldraw[draw=green!80, -, rounded corners=2em, line width=2em, 
          cap=round] (4.center) -- (7.center);
    \filldraw[draw=green!20, -, rounded corners=2em, line width=1.7em, 
          cap=round] (4.center) -- (7.center);
    \node[green, align=center] (super) [left =of 7] {\textbf{Super-seed}\\ $\{v_4,v_7\}$};
  \end{pgfonlayer}
\end{tikzpicture}
}
\caption{Generation of a PRR-Graph. (Solid and dashed
  arrows represent live and live-upon-boost edges respectively.)}
\label{fig:prr_graph_gen}
\end{figure}

\begin{algorithm}[htb]
\SetKw{KwAnd}{and}%

Select a random node $r$ as the root node\;\label{algo:prr_gen_start}
\lIf{$r\in S$}{\KwRet{$R$ is activated}}
Create a graph $R$ with a singleton node $r$\;
Create a double-ended queue $Q$ with $(r,0)$\;
Initialize $d_r[r]\leftarrow 0$ and $d_r[v]\leftarrow +\infty,\forall v\neq r$\;
\While{$Q$ is not empty} {
  $(u,d_{ur}) \leftarrow Q$.dequeue\_front()\;
  \lIf(\tcp*[f]{we've processed $u$}){$d_{ur}>d_r[u]$}{continue}\label{algo:prr_gen_outdated}
  \For{each non-blocked incoming edge $e_{vu}$ of $u$} {
    $d_{vr}\leftarrow\indicator{e_{vu} \text{ is \textit{live-upon-boost}}} + d_{ur}$\;
    \lIf(\tcp*[f]{pruning}){$d_{vr}>k$}{continue}\label{algo:prr_gen_prune}
    Add $e_{vu}$ to $R$\;
    \If{$d_{vr}<d_r[v]$}{
      $d_r[v] \leftarrow d_{vr}$\;
      \If{$v\in S$}{
        \lIf{$d_r[v]=0$}{\KwRet{$R$ is activated}\label{algo:prr_gen_activated}}
      }
      \lElseIf{$d_{vr}\!=\!d_{ur}$}{$Q$.enqueue\_front($(v,d_{vr})$)}
      \lElse{$Q$.enqueue\_back($(v,d_{vr})$)}
    }
  }
}
\lIf{there is no seed in $R$}{\KwRet{$R$ is hopeless}}\label{algo:prr_gen_end}
Compress the boostable $R$ to reduce its size\;\label{algo:prr_gen_compress}
\KwRet{a compressed boostable $R$}
\caption{Generating a random PRR-graph $(G,S,k)$}\label{algo:generation_prr}
\end{algorithm}

\mynoindent\textbf{Phase I:~Generating a PRR-graph.}
Let $r$ be a random node.
We include into $R$ all \textit{non-blocked} paths
  from seed nodes to $r$ with at most $k$ \textit{live-upon-boost} edges via a 
  a backward Breadth-First Search (BFS) from $r$.
The status of each edge (i.e., \textit{live, live-upon-boost, blocked})
  is sampled when we first process it.
The detailed backward BFS is as follows.
Define the \textit{distance} of a path from $u$ to $v$ as the number of
  \textit{live-upon-boost} edges on it.
Then, the \textit{shortest distance} from $v$ to $r$ is the minimum number of
  nodes we have to boost so that at least a path from $v$ to $r$ becomes live.
For example, in \Cref{fig:prr_graph_1}, the shortest distance from
  $v_7$ to $r$ is one.
We use $d_r[\cdot]$ to maintain the 
  shortest distances from nodes to the root node $r$.
Initially, we have $d_r[r]=0$ and we enqueue $(r,0)$ into a double-ended queue $Q$.
We repeatedly dequeue and process a node-distance pair $(u, d_{ur})$
  from the head of $Q$, until the queue is empty.
Note that the distance $d_{ur}$ in a pair $(u,d_{ur})$
  is the shortest known distance from $u$ to $r$ when the pair was enqueued.
Thus we may find $d_{ur}>d_r[u]$ in \Cref{algo:prr_gen_outdated}.
Pairs $(u,d_{ur})$ in $Q$ are in the ascending order of the distance $d_{ur}$
  and there are at most two different values of distance in $Q$.
Therefore, we process nodes in the ascending order of their shortest
  distances to $r$.
When we process a node $u$,
  for each of its non-blocked incoming edge $e_{vu}$,
  we let $d_{vr}$ be the shortest distance from $v$ to $r$ via $u$.
If $d_{vr}>k$, all paths from $v$ to $r$ via $u$ are impossible to
  become live upon boosting at most $k$ nodes,
  therefore we ignore $e_{vu}$ in \Cref{algo:prr_gen_prune}.
This is in fact a ``pruning'' strategy, because
  it may reduce unnecessary costs in the generation step.
The pruning strategy is effective for small values of $k$.
For large values of $k$, only a small number of paths need to be pruned
  due to the small-world property of real social networks.
If $d_{vr}\leq k$,
  we insert $e_{vu}$ into $R$, update $d_r[v]$ and enqueue $(v,d_{vr})$ if necessary.
During the generation,
  if we visit a seed node $s$ and its shortest distance to $r$ is zero,
  we know $R$ is \textit{activated} and we terminate the generation
  (\Cref{algo:prr_gen_activated}).
If we do not visit any seed node during the backward BFS,
  $R$ \textit{is hopeless}
  and we terminate the generation (\Cref{algo:prr_gen_end}).

\mynoindent\textbf{Remarks.}
At the end of phase I,
  $R$ may include nodes and edges not belonging to it (e.g., 
  \textit{non-blocked} edges not on any non-blocked paths from seeds to the
  root).
These \textit{extra} nodes and edges will be removed in the compression phase.
For example, \Cref{fig:prr_graph_1} shows the results of the first phase,
  given that we are constructing a PRR-graph $R$ according to the root node
  $r$ and sampled edge status shown in \Cref{fig:prr_graph}.
At the end of the first phase, $R$ also includes the \textit{extra} edge
  from $v_8$ to $v_2$ and they will be removed later.

\mynoindent\textbf{Phase II:~Compressing the PRR-graph.}
When we reach \Cref{algo:prr_gen_compress}, $R$ is \textit{boostable}.
In practice, we observe that we can remove and merge a significant
  fraction of nodes and edges from $R$ (i.e., \textit{compress} $R$),
  while keeping values of $f_R(B)$ and $f_R^{-}(B)$ for all $|B|\leq k$ same as before.
Therefore, we compress all boostable PRR-graphs to prevent the memory usage
  from becoming a bottleneck.
\Cref{fig:prr_graph_2} shows the compressed result of \Cref{fig:prr_graph_1}.
\textit{First},
  we merge nodes $v_4$ and $v_7$ into a single ``super-seed'' node,
  because they are activated without boosting any node.
Then, we remove node $v_6$ and its incident edges,
  because they are not on any paths from the super-seed node to the root node $r$.
Similarly, we remove the extra node $v_8$ and the extra edge from $v_8$ to $v_2$.
\textit{Next},
  observing that there are live paths from nodes $v_0,v_1,v_2$ and $v_3$
  to root $r$, we remove all outgoing edges of them,
  and add a direct live edge from each of these nodes to $r$.
After doing so, we remove node $v_0$ because 
  it is not on any path from the super-seed node to $r$.
Now, we describe the compression phase in detail.

The \textit{first part} of the compression merges
  nodes into the super-seed node.
We run a forward BFS from seeds in $R$ and compute 
  the shortest distance $d_S[v]$ from seeds to $v$ for every node $v$ in $R$.
Let $X=\{v|d_S[v]=0\}$,
  we have $f_R(B)=f_R(B\backslash X)$ for all $B\subseteq V$ because
  whether we boost any subset of $X$ has nothing to do with
  whether the root node of $R$ is activated.
Thus, we merge nodes in $X$ as a single super-seed node:
we insert a super-seed node $x$ into $R$;
for every node $v$ in $X$, we remove its incoming edges
  and redirect every of its outgoing edge $e_{vu}$ to $e_{xu}$.
Finally, we clean up nodes and edges not on any paths from the
  super-seed node to the root node $r$.

In the \textit{second part}, we add live edges so that nodes connected to $r$
  through live paths are directly linked to $r$.
We also clean up nodes and edges 
  that are not necessary for later computation.
For a node $v$, let $d'_r[v]$ be the shortest distance from $v$ to $r$
  without going through the super-seed.
If a node $v$ satisfies $d'_r[v]+d_S[v]>k$, every path going through $v$
  cannot be live with at most $k$ nodes boosted, therefore we remove
  $v$ and its adjacent edges.
If a non-root node $v$ satisfies $d'_r[v]=0$,
  we remove its outgoing edges and add a live edge from $v$ to $r$.
In fact, in a boostable $R$, if a node $v$ satisfies $d'_r[v]=0$,
  we must have $d_r[v]=0$ in the first phase.
In our implementation,
  if a node $v$ satisfies $d_r[v]=0$,
  we in fact clear outgoing edges of $v$ and add the live edge $e_{vr}$ to $R$
  in the first phase.
Finally, we remove ``non-super-seed'' nodes with no incoming edges.

\mynoindent\textbf{Time complexity.}
The cost for the first phase of the PRR-graph generation is linear to the number
  of edges visited during the generation.
The compression phase runs linear to the number of uncompressed edges generated
  in the generation phase.
\Cref{sec:experiments} shows the average number of uncompressed edges
  in boostable PRR-graphs in several social networks.

\subsection{PRR-Boost Algorithm}
We obtain our algorithm, \texttt{PRR-Boost}, by integrating PRR-graphs,
  the \texttt{IMM} algorithm and the \textit{Sandwich Approximation} strategy.
\Cref{algo:prr_boost} depicts \texttt{PRR-Boost}. 

\begin{algorithm}[htb]
\SetKw{KwAnd}{and}%
\SetKwFunction{enqueuefront}{enqueue\_front}%
\SetKwFunction{enqueueback}{enqueue\_back}%
\SetKwFunction{dequeuefront}{dequeue\_front}%

$\ell' = \ell\cdot (1+\log 3/\log n)$\label{algo:prr_boost_lb_start}\;
$\mathcal{R}\leftarrow$ \texttt{SamplingLB}$(G,S,k,\epsilon,\ell')$
\tcp*{sampling in \texttt{IMM}~\cite{tang2015influence} using the PRR-graph generation
		of Algo.~\ref{algo:generation_prr}}\label{algo:prr_boost_sampling}

$B_{\mu}\leftarrow$ \texttt{NodeSelectionLB}$(\mathcal{R}, k)$
\tcp*{maximize $\mu$}
\label{algo:prr_boost_lb_end}\label{algo:prr_boost_ns_lb}

$B_{\Delta}\leftarrow$ \texttt{NodeSelection}$(\mathcal{R}, k)$
\tcp*{maximize $\Delta_{S}$}
\label{algo:prr_boost_ns_true}

$B_\text{sa}=\argmax_{B\in \{B_{\Delta}, B_{\mu}\}}
  \hat{\Delta}_{\mathcal{R}}(B)$\label{algo:prr_boost_select}\;
\KwRet{$B_\text{sa}$}

\caption{PRR-Boost$(G,S,k,\epsilon,\ell)$}\label{algo:prr_boost}
\end{algorithm}

\Crefrange{algo:prr_boost_lb_start}{algo:prr_boost_lb_end} utilize the 
  \texttt{IMM} algorithm~\citep{tang2015influence} with the PRR-graph generation given in 
  \Cref{algo:generation_prr}   
  to maximize the lower bound $\mu$ of $\Delta_S$ under the cardinality
  constraint of $k$.
\Cref{algo:prr_boost_ns_true}
  greedily selects a set $B_{\Delta}$ of nodes
  with the goal of maximizing $\Delta_{S}$,
  and we reuse PRR-graphs in $\mathcal{R}$ to estimate $\Delta_{S}(\cdot)$.
Finally, between $B_{\mu}$ and $B_{\Delta}$,
  we return the set with a larger estimated boost of influence 
  as the final solution.

\mynoindent\textbf{Approximation ratio.}
Let $B_{\mu}^*$ be the optimal solution for maximizing $\mu$ under the
  cardinality constraint of $k$, and let $OPT_\mu=\mu(B_{\mu}^*)$.
By the analysis of the \texttt{IMM} method,
  we have the following lemma.

\begin{restatable}{lemma}{lemmaprrboostlb}\label{lemma:prr_boost_lb}
In \Cref{algo:prr_boost},
  define $\epsilon_1=\epsilon\cdot \alpha/((1-1/e)\alpha+\beta)$
  where $\alpha=\sqrt{\ell'\log n+\log 2}$,
  and $\beta = \sqrt{(1-1/e)\cdot(\log\tbinom{n}{k}+\ell'\log n+\log 2)}$.
With a probability of at least $1-n^{-\ell'}$,
  the number of PRR-graphs generated in \Cref{algo:prr_boost_sampling}
  satisfies
  \begin{align}
    |\mathcal{R}| \geq \frac{(2-2/e)\cdot n
                       \cdot \log\big(\tbinom{n}{k}\cdot 2n^{\ell'} \big)}
                      {{(\epsilon-(1-1/e)\cdot\epsilon_1)}^2\cdot OPT_{\mu}}
      \quad\big(\text{Th.2 by \cite{tang2015influence}}\big).
      \label[ineq]{eq:prr_boost_lb1}
  \end{align}
Given that \Cref{eq:prr_boost_lb1} holds,
  with probability at least $1-n^{-\ell'}$,
  the set $B_{\mu}$ returned by \Cref{algo:prr_boost_ns_lb} satisfies
  \begin{align}
    n\!\cdot\! \hat{\mu}_{\mathcal{R}}(B_\mu)
      \geq (1-1/e)(1-\epsilon_1)\cdot OPT_{\mu}
      \quad\big(\text{Th.1 by \cite{tang2015influence}}\big).
      \label[ineq]{eq:prr_boost_lb2}
  \end{align}
\end{restatable}

Ideally, we should select 
  $B_\text{sa}=\argmax_{B\in\{B_{\mu},B_{\Delta}\}}\Delta_{S}(B)$.
Because of the \#P-hardness of computing $\Delta_{S}(B)$ for any given $B$,
  we select $B_\text{sa}$ between 
  $B_{\mu}$ and $B_{\Delta}$ with the larger \textit{estimated}
  boost of influence in \Cref{algo:prr_boost_select}.
The following lemma shows that boosting $B_\text{sa}$ leads to a large expected boost.

\begin{restatable}{lemma}{lemmaprrboostselect}\label{lemma:prr_boost_select}
Given that \Crefrange{eq:prr_boost_lb1}{eq:prr_boost_lb2} hold,
  with a probability of at least $1-n^{-\ell'}$, we have
  $\Delta_{S}(B_\text{sa})
\geq (1-1/e-\epsilon)\cdot OPT_{\mu} \geq (1-1/e-\epsilon)\cdot \mu(B^*)$.
\end{restatable}

\begin{proof}
(outline) Let $B$ be a boost set with $k$ nodes, we say that $B$ is a \textit{bad} set if
$\Delta_{S}(B)<(1-1/e-\epsilon)\cdot OPT_{\mu}$.
Let $B$ be an arbitrary \textit{bad} set with $k$ nodes.
If we return $B$, we must have $\hat{\Delta}_{\mathcal{R}}(B)>\hat{\Delta}_{\mathcal{R}}(B_{\mu})$,
and we can prove that 
  $\Pr[\hat{\Delta}_{\mathcal{R}}(B) > \hat{\Delta}_{\mathcal{R}}(B_{\mu})] \leq n^{\!-\!\ell'}/\tbinom{n}{k}$.
Because there are at most $\binom{n}{k}$ \textit{bad} sets with $k$ nodes,
  and because $OPT_{\mu}\geq \mu(B^*)$, 
  \Cref{lemma:prr_boost_select} holds.
The full proof can be found in the appendix.
\end{proof}

\mynoindent\textbf{Complexity.}
Let $EPT$ be the expected number of edges explored for generating a random PRR-graph.
Generating a random PRR-graph runs in $O(EPT)$ expected time,
  and the expected number of edges in a random PRR-graph is at most $EPT$.
Denote the \textit{size} of $\mathcal{R}$ as the total number
  of edges in PRR-graphs in $\mathcal{R}$.
\Crefrange{algo:prr_boost_lb_start}{algo:prr_boost_lb_end} of \texttt{PRR-Boost}
  are essentially the \texttt{IMM} method with the goal of maximizing $\mu$.
By the analysis of the general \texttt{IMM} method,
  both the expected complexity of the sampling step in
  \Cref{algo:prr_boost_sampling}
  and the size of $\mathcal{R}$ are
  $O(\frac{EPT}{OPT_{\mu}}\cdot(k+\ell)(n+m)\log n/\epsilon^2)$.
The node selection in \Cref{algo:prr_boost_ns_lb} corresponds to the greedy algorithm for maximum coverage,
  thus runs in time linear to the
  \textit{size} of $\mathcal{R}$.
The node selection in \Cref{algo:prr_boost_ns_true} runs in 
  $O(\frac{EPT}{OPT_{\mu}}\cdot k(k+\ell)(n+m)\log n/\epsilon^2)$ expected time,
  because updating $\hat{\Delta}_{\mathcal{R}}(B\cup\{v\})$ for all $v\notin B\cup S$
  takes time linear to the \textit{size} of $\mathcal{R}$
  after we insert a node into $B$.

From \Crefrange{lemma:prr_boost_lb}{lemma:prr_boost_select},
  we have $\Delta_{S}(B_\text{sa}) \geq (1-1/e-\epsilon)\cdot \mu(B^*)$
  with probability at least $1-3n^{-\ell'}=1-n^{-\ell}$.
Together with the above complexity analysis,
  we have the following theorem about \texttt{PRR-Boost}.

\begin{theorem}\label{th:prr_boost}
With a probability of at least $1-n^{-\ell}$,
  \texttt{PRR-Boost} returns a 
  $(1-1/e-\epsilon)\cdot\frac{\mu(B^*)}{\Delta_{S}(B^*)}$-approximate solution.
Moreover, it runs in 
  $O(\frac{EPT}{OPT_{\mu}}\cdot k \cdot(k+\ell)(n+m)\log n/\epsilon^2)$
  expected time.
\end{theorem}

The approximation ratio given in \Cref{th:prr_boost} depends on the ratio of
  $\frac{\mu(B^*)}{\Delta_{S}(B^*)}$,
	which should be close to one if the lower bound function $\mu(B)$ is close to 
	the actual boost of influence $\Delta_{S}(B)$,
  when $\Delta_{S}(B)$ is large.
\Cref{sec:experiments} demonstrates that $\mu(B)$ is indeed close to 
	$\Delta_{S}(B)$ in real datasets.

\subsection{The PRR-Boost-LB Algorithm}
\texttt{PRR-Boost-LB} is a simplification of \texttt{PRR-Boost}
  where we return the node set $B_{\mu}$ as the final solution.
Recall that the estimation of $\mu$ only relies on the critical node set $C_R$
  of each boostable PRR-graph $R$.
In the first phase of the PRR-graph generation,
  if we only need to obtain $C_R$,
  there is no need to explore incoming edges of a node $v$ if $d_r[v]>1$.
Moreover, in the compression phase,
  we can obtain $C_R$ right after computing $d_S[\cdot]$ and we can
  terminate the compression earlier.
The sampling phase of \texttt{PRR-Boost-LB}
  usually runs faster than that of \texttt{PRR-Boost},
  because we only need to generate $C_R$ for each boostable PRR-graph $R$.
In addition, the memory usage is significantly lower than that for \texttt{PRR-Boost},
  because the averaged number of ``critical nodes'' in a random boostable
  PRR-graph is small in practice.
In summary, compared with \texttt{PRR-Boost},
  \texttt{PRR-Boost-LB} has the same approximation factor
  but runs faster than \texttt{PRR-Boost}.
We will compare \texttt{PRR-Boost} and \texttt{PRR-Boost-LB} by experiments
  in \Cref{sec:experiments}.

\subsection{Discussion: The Budget Allocation Problem}
A question one may raise is
  what is the best strategy if companies could freely decide
  how to allocation budget on both seeding and boosting.
A heuristic method combing influence maximization algorithms 
  and \texttt{PRR-Boost} is as follows.
We could test different budget allocation strategy.
For each allocation, we first identify seeds using any influence
  maximization algorithm, then we find boosted user by \texttt{PRR-Boost}.
Finally, we could choose the budget allocation strategy leading to the largest
  boosted influence spread among all tested ones.
In fact, the budget allocation problem could be much harder than the
  $k$-boosting problem itself, and its full treatment is beyond the scope of this study and
  is left as a  future work.

%% file: bidirected.tex
\section{Boosting on Bidirected Trees}
\label{sec:bidirected}

In this section,
  we study the $k$-boosting problem
  where influence propagates on \textit{bidirected trees}.
On bidirected trees,
  the computation of the boost of influence spread becomes tractable.
We are able to devise an efficient greedy algorithm and
  an approximation algorithm with a near-optimal approximation ratio.
This demonstrates that the hardness of the $k$-boosting problem is partly
  due to the graph structure, and when we restrict to tree structures,
  we are able to find near-optimal solutions.
Moreover, using near-optimal solutions as benchmarks
  enables us to verify that a greedy node selection method on trees
  in fact returns near-optimal solutions in practice. 
Besides, our efforts on trees will help to designing 
  heuristics for the $k$-boosting problem on general graphs,
  or for other related problems in the future.

\mynoindent\textbf{Bidirected trees.}
A directed graph $G$ is a \textit{bidirected tree} if and only if
  its underlying undirected graph (with 
  directions and duplicated edges removed) is a tree.
For simplicity of notation,
  we assume that every two adjacent nodes are connected by two edges,
  one in each direction.
We also assume that nodes are activated with probability less than one, because
  nodes that will be activated for sure could be identified in linear time 
  and they could be treated as seeds.
\Cref{fig:bidirected_example} shows an example of a bidirected tree.
The existence of bidirected edges brings challenges
  to the algorithm design,
  because the influence may flow from either direction between a pair of
  neighboring nodes.
 
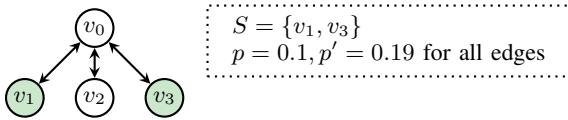
\begin{figure}[htb]
  \centering
  \begin{tikzpicture}[thick, node distance = 0.4cm and 0.4cm]

    \node[main node] (0) {$v_0$};
    \node[main node] (2) [below =of 0] {$v_2$};
    \node[seed node] (1) [left =of 2] {$v_1$};
    \node[seed node] (3) [right =of 2] {$v_3$};
    \node[draw, dotted, below right = -5mm and 13mm of 0]{
      \begin{tabular}{l}$S=\{v_1,v_3\}$\\$p=0.1,p'=0.19$ for all edges\end{tabular}
    };

    \draw [<->] (1) -- (0);
    \draw [<->] (0) -- (2);
    \draw [<->] (3) -- (0);

  \end{tikzpicture}
  \caption{A bidirected tree with four nodes and six directed edges.}\label{fig:bidirected_example}
\end{figure}

In the remaining of this section,
  we first present how to compute the exact boosted influence spread
  on bidirected trees, and a greedy algorithm \texttt{Greedy-Boost} based on it.
Then, we present a rounded dynamic programming \texttt{DP-Boost},
  which is a \textit{fully polynomial-time approximation scheme}.
\texttt{Greedy-Boost} is efficient but does not provide the approximation guarantee.
\texttt{DP-Boost} is more computationally expensive but guarantees a near-optimal approximation ratio.
We leave all proofs of this section in the appendix.

\subsection{Computing the boosted influence spread}\label{sec:computing_boost_tree}
In this part,
  we first discuss how to compute the boosted influence spread
  in a bidirected tree.
The computation serves as a building block for the greedy algorithm 
  that iteratively selects nodes with the maximum marginal gain of
  the boosted influence spread.

We separate the computation into \textit{three steps}.
(1)
We refer to the probability that a node gets activated (i.e., influenced) 
  as its ``\textit{activation probability}''.
For every node $u$,
  we compute the increase of its activation probability
  when it is inserted into $B$.
(2)
If we regard a node $u$ as the root of the tree,
  the remaining nodes could be categorized into multiple ``subtrees'',
  one for each neighbor of $u$.
For every node $u$,
  we compute intermediate results that help us to 
  determine the increase of influence spread
  in each such ``subtree'' if we insert $u$ into $B$.
(3)
Based on the previous results, we compute
$\sigma_S(B)$ and $\sigma_S(B\cup\{u\})$ for every node $u$.
If necessary, we are able to obtain $\Delta_{S}(B\cup\{u\})$ from
  $\sigma_{S}(B\cup\{u\})-\sigma_{S}(\emptyset)$.

\mynoindent\textbf{Notations.}
We use $\edgep{u}{v}{B}$ to denote the influence probability of an
  edge $e_{uv}$, given that we boost nodes in $B$.
Similarly, let $\edgep{u}{v}{b}$ be the influence probability of $e_{uv}$,
  where $b\in\{0,1\}$ indicates whether $v$ is boosted.
We use $N(u)$ to denote the set of neighbors of node $u$.
Given neighboring nodes $u$ and $v$,
  we use $G_{\udelv}$ to denote the subtree of $G$
  obtained by first removing node $v$
  and then removing all nodes not connected to $u$.
To avoid cluttered notation,
  we slightly abuse the notation and keep using $S$ and $B$ to denote
  seed users and boosted users in $G_{\udelv}$,
  although some nodes in $S$ or $B$ may not be in $G_{\udelv}$.

\mynoindent\textbf{Step I:~Activation probabilities.}
For a node $u$,
  we define $ap_B(u)$ as the activation probability of $u$ when we boost $B$.
For $v\in N(u)$, we define $ap_B(\udelv)$ as the activation probability
  of node $u$ in $G_{\udelv}$ when we boost $B$.
For example, in \Cref{fig:bidirected_example}, suppose \mbox{$B=\emptyset$},
  we have $ap_B(v_0)=1-{(1-p)}^{2}=0.19$ and $ap_B(v_0\backslash v_1)=p=0.1$.
By the above definition, we have the following lemma.

\begin{restatable}{lemma}{lemmaapuv}\label{lemma:apuv}
Suppose we are given a node $u$.
If $u$ is a seed node (i.e., $u\in S$),
  we have $ap_B(u)=1$ and $ap_B(\udelv)=1$ for all $v\in N(u)$.
Otherwise, we have
\onlypaper{
\begin{align}
  ap_B(u)&=1{-}\prod_{v\in N(u)}\left( 1{-}ap_B(v\backslash u)\!\cdot\! \edgep{v}{u}{B} \right),\label{eq:apu}\\
  ap_B(\udelv)&=1{-}\prod_{w\in N(u)\backslash\{v\}}
                \left(1{-}ap_B(w\backslash u)\!\cdot\! \edgep{w}{u}{B}\right),
  \forall v\in N(u),\label{eq:apuv_slow}\\
  ap_B(\udelv) &= 1{-}\big(1{-}ap_B(u\backslash w)\big)\!\cdot\!
  \frac{1{-}ap_B(w\backslash u)\!\cdot\!\edgep{w}{u}{B}}
                 {1{-}ap_B(v\backslash u)\!\cdot\!\edgep{v}{u}{B}}, 
  \forall v,w\in N(u), v\neq w.\label{eq:apuv_fast}
\end{align}
}
\onlytech{
\begin{align}
  ap_B(u)=&1{-}\prod_{v\in N(u)}\left( 1{-}ap_B(v\backslash u)\!\cdot\! \edgep{v}{u}{B} \right),\label{eq:apu}\\
  ap_B(\udelv)=&1{-}\prod_{w\in N(u)\backslash\{v\}}
                \left(1{-}ap_B(w\backslash u)\!\cdot\! \edgep{w}{u}{B}\right),\nonumber\\
  &\forall v\in N(u),\label{eq:apuv_slow}\\
  ap_B(\udelv) =& 1{-}\big(1{-}ap_B(u\backslash w)\big)\!\cdot\!
  \frac{1{-}ap_B(w\backslash u)\!\cdot\!\edgep{w}{u}{B}}
                 {1{-}ap_B(v\backslash u)\!\cdot\!\edgep{v}{u}{B}}, \nonumber\\ 
  &\forall v,w\in N(u), v\neq w.\label{eq:apuv_fast}
\end{align}
}
\end{restatable}

\begin{algorithm}[htb]
\makeatletter \if@twocolumn
  \SetInd{0.5em}{0.5em}
\fi \makeatother
  \SetKwFunction{getapuv}{ComputeAP}
  Initialize $ap_B(\udelv)$ as ``not computed'' for all $u$ and $v\in N(u)$\label{algo:ap_uv_begin}\;
  \ForEach{$u\in V$}{
    \ForEach{$v\in N(u)$}{
      \getapuv{$u$, $v$}\tcp*{compute $ap_B(\udelv)$}
    }
  }\label{algo:ap_uv_end}
  \ForEach{$u\in V$}{\label{algo:apu_begin}
    \lIf{$u\in S$}{
      $ap_B(u) \leftarrow 1$ 
    }\lElse{
      $ap_B(u) \leftarrow 1 - \prod_{v\in N(u)}(1-ap_B(v\backslash u)\cdot\edgep{v}{u}{B})$
    }
  }\label{algo:apu_end}
  \myfunc{\getapuv{$u$, $v$}}{
    \If{we have not computed $ap_B(u\backslash v)$}{\label{algo:line_apuv_test}
      \lIf{$u\in S$}{$ap_B(u\backslash v)\leftarrow 1$}\label{algo:line_apuv_trivial}
      \ElseIf{we have not computed $ap_B(u\backslash w)$ for any $w\in N(u)\backslash\{v\}$}{\label{algo:line_apuv_slow_begin}
        \lForEach{$w\in N(u)\backslash\{v\}$}{ \getapuv{$w$, $u$}}%
        $ap_B(u\backslash v){\leftarrow}1{-}\prod_{w\in
          N(u)\backslash\{v\}}(1{-}ap_B(w\backslash u) \edgep{w}{u}{B})$\label{algo:line_apuv_slow_eq}\;
      }\label{algo:line_apuv_slow_end}
      \Else{\label{algo:line_apuv_fast_begin}
        Suppose we have computed $ap_B(u\backslash w)$ for a node $w\in N(u)\backslash\{v\}$\;
        \getapuv{$w$, $u$}\;
        \vspace{-2mm}
        $ap_B(u\backslash v)\leftarrow 1-(1-ap_B(u\backslash w))\cdot
        \frac{1-ap_B(w\backslash u)\cdot\edgep{w}{u}{B}}
        {1-ap_B(v\backslash u)\cdot\edgep{v}{u}{B}}$\label{algo:line_apuv_fast_eq}\;
      }\label{algo:line_apuv_fast_end}
    }
  }
  \caption{Computing Activation Probabilities}\label{algo:ap_uv}
\end{algorithm}

\Cref{algo:ap_uv} depicts how we compute activation probabilities.
\Crefrange{algo:ap_uv_begin}{algo:ap_uv_end}
  initialize and compute $ap_B(\udelv)$ for all neighboring nodes $u$ and $v$.
\Crefrange{algo:apu_begin}{algo:apu_end} compute $ap_B(u)$ for all nodes $u$.
The recursive procedure \texttt{ComputeAP($u,v$)} for computing $ap_B(\udelv)$ works as follows.
\Cref{algo:line_apuv_test} guarantees that we do not re-compute $ap_B(\udelv)$.
\Cref{algo:line_apuv_trivial} handles the trivial case where node $u$ is a seed.
\Crefrange{algo:line_apuv_slow_begin}{algo:line_apuv_slow_end} compute
  the value of $ap_B(\udelv)$ using \Cref{eq:apuv_slow}.
\Crefrange{algo:line_apuv_fast_begin}{algo:line_apuv_fast_end} compute
  $ap_B(\udelv)$ more efficiently using \Cref{eq:apuv_fast},
  taking advantages of the known $ap_B(u\backslash w)$ and $ap_B(v\backslash u)$.
Note that in \Cref{algo:line_apuv_fast_end},
  the value of $ap_B(v\backslash u)$ must have been computed, because we have computed
  $ap_B(u\backslash w)$, which relies on the value of $ap_B(v\backslash u)$.
For a node $u$, given the values of $ap_B(w\backslash u)$ for all $w\in N(u)$,
  we can compute $ap_B(u\backslash v)$ for all $v\in N(u)$ in $O(|N(u)|)$.
Then, for a node $u$, given values of $ap_B(w\backslash u)$ for all $w\in N(u)$,
  we can compute $ap_B(u)$ in $O(|N(u)|)$.
Therefore, the time complexity of \Cref{algo:ap_uv} is $O(\sum_{u}|N(u)|)=O(n)$,
  where $n$ is the number of nodes in the bidirected tree.

\mynoindent\textbf{Step II:~More intermediate results.}
Given that we boost $B$,
  we define $g_B(\udelv)$ as the ``gain'' of the influence spread in $G_{\udelv}$
  when we add node $u$ into the current seed set $S$.
Formally, $g_B(\udelv)$ is defined as $g_B(\udelv)=\sigma^{G_{\udelv}}_{S\cup \{u\}}(B)
    -\sigma^{G_{\udelv}}_{S}(B)$,
  where $\sigma^{G_{\udelv}}_{S}(B)$ is the boosted influence spread
  in $G_{\udelv}$ when the seed set is $S$ and we boost $B$.
In \Cref{fig:bidirected_example}, we have
  $G_{v_0\backslash v_1}=G\backslash\{e_{01},e_{10}\}$.
Suppose $B=\emptyset$, when we insert $v_0$ into $S$,
  the boosted influence spread in $G_{v_0\backslash v_1}$ increases from $1.11$ to $2.1$,
  thus $g_B(v_0,v_1)=0.99$.
We compute $g_B(\udelv)$ for all neighboring nodes $u$ and $v$ using the
  formulas in the following lemma.
  
\begin{restatable}{lemma}{lemmaguv}\label{lemma:guv}
Suppose we are given a node $u$.
If $u$ is a seed node (i.e., $u\in S$), we have $g_B(\udelv)=0$.
Otherwise, for any $v\in N(u)$, we have
\onlypaper{
  \begin{align}\label{eq:guv}
    g_B(u\backslash v)=&\Big(1-ap_B(\udelv)\Big) 
    \Big(1+ \sum_{w\in N(u)\backslash\{v\}}
               \frac{\edgep{u}{w}{B}\cdot g_B(w\backslash u)}
                    {1-ap_B(w\backslash u)\cdot \edgep{w}{u}{B}}\Big).
  \end{align}
}
\onlytech{
  \begin{align}\label{eq:guv}
    g_B(u\backslash v)=&\Big(1-ap_B(\udelv)\Big) \cdot \nonumber \\
    & \Big(1+ \sum_{w\in N(u)\backslash\{v\}}
               \frac{\edgep{u}{w}{B}\cdot g_B(w\backslash u)}
                    {1-ap_B(w\backslash u)\cdot \edgep{w}{u}{B}}\Big).
  \end{align}
}
Moreover, for $v,w\in N(u)$ and $v\neq w$, we have
\onlypaper{
  \begin{align}\label{eq:guv_fast}
    g_B(&\udelv) =  \big(1-ap_B(\udelv)\big)\cdot 
              \Big(\frac{g_B(u\backslash w)}{1-ap_B(u\backslash w)} 
    + \frac{\edgep{u}{w}{B}\cdot g_B(w\backslash u)}{1-ap_B(w\backslash u)
                 \cdot \edgep{w}{u}{B}}
               - \frac{\edgep{u}{v}{B}\cdot g_B(v\backslash u)}{1-ap_B(v\backslash u)
                 \cdot \edgep{v}{u}{B}}\Big).
  \end{align}
}
\onlytech{
  \begin{align}\label{eq:guv_fast}
    g_B(&\udelv) =  \big(1-ap_B(\udelv)\big)\cdot 
              \Big(\frac{g_B(u\backslash w)}{1-ap_B(u\backslash w)} 
                \nonumber \\
    & + \frac{\edgep{u}{w}{B}\cdot g_B(w\backslash u)}{1-ap_B(w\backslash u)
                 \cdot \edgep{w}{u}{B}}
               - \frac{\edgep{u}{v}{B}\cdot g_B(v\backslash u)}{1-ap_B(v\backslash u)
                 \cdot \edgep{v}{u}{B}}\Big).
  \end{align}
}
\end{restatable}

\Cref{eq:guv} shows how to compute $g_B(\udelv)$ by definition.
\Cref{eq:guv_fast} provides a faster way to compute $g_B(\udelv)$,
  taking advantages of the previously computed values.
Using similar algorithm in \Cref{algo:ap_uv},
  we are able to compute $g_B(\udelv)$ for all $u$ and $v\in N(u)$ in $O(n)$.

\mynoindent\textbf{Step III:~The final computation.}
Recall that $\sigma_{S}(B)$ is the expected influence spread upon boosting $B$,
  we have $\sigma_{S}(B)=\sum_{v\in V}ap_B(v)$.
The following lemma shows how we compute $\sigma_{S}(B\cup\{u\})$.

\begin{restatable}{lemma}{lemmasigmau}\label{lemma:sigmau}
Suppose we are given a node $u$.
If $u$ is a seed node or a boosted node (i.e., $u\in B\cup S$),
  we have $\sigma_{S}(B\cup\{u\})=\sigma_S(B)$.
Otherwise, we have
\onlypaper{
\begin{align}
\sigma_{S}(B\cup\{u\})& = \sigma_{S}(B)+ \Delta ap_B(u)
    + \sum_{v\in N(u)} \edgep{u}{v}{B}\cdot \Delta ap_{B}(\udelv) \cdot g_B(v\backslash u),\label{eq:sigmau}
\end{align}
}
\onlytech{
\begin{align}
\sigma_{S}(B\cup\{u\})& = \sigma_{S}(B)+ \Delta ap_B(u)\nonumber\\
  & + \sum_{v\in N(u)} \edgep{u}{v}{B}\cdot \Delta ap_{B}(\udelv) \cdot g_B(v\backslash u),\label{eq:sigmau}
\end{align}
}
where
  $\Delta ap_{B}(u):=ap_{B\cup\{u\}}(u)-ap_B(u)
    =1-\prod_{v\in N(u)}\big(1-ap_B(v\backslash u)\cdot p_{v,u}'\big)-ap_B(u)$ and
  $\Delta ap_{B}(\udelv):=ap_{B\cup\{u\}}(\udelv)-ap_B(\udelv)
  =1-\prod_{w\in N(u)\backslash\{v\}}\big(1-ap_B(w\backslash u)\cdot p_{w,u}'\big)-ap_B(\udelv)$.
\end{restatable}

The intuition behind \Cref{eq:sigmau} is as follows.
Let $V_{v\backslash u}\subseteq V$ be the set of nodes in $G_{v\backslash
  u}\subseteq G$.
When we insert a node $u$ into $B$,
  $\Delta ap_B(u)$ is the increase of the activation probability of $u$ itself,
  and
  $\edgep{u}{v}{B}\cdot \Delta ap_{B}(\udelv)\cdot g_B(v\backslash u)$
  is the increase of the number of influenced nodes in $V_{v\backslash u}$.
The final step computes $\sigma_{S}(B)$ and
  $\sigma_{S}(B\cup\{u\})$ for all nodes $u$ in $O(n)$.

\mynoindent\textbf{Putting it together.}
Given a bidirected tree $G$ and a set of currently boosted nodes $B$,
  we are interested in computing $\sigma_{S}(B)$ and $\sigma_{S}(B\cup\{u\})$
  for all nodes $u$.
For all $u\in V$ and $v\in N(u)$,
  we compute $ap_B(u)$ and $ap_B(\udelv)$ in the first step,
  and we compute $g_B(\udelv)$ in the second step.
In the last step, we finalize the computation of $\sigma_{S}(B)$ and
  $\sigma_{S}(B\cup\{u\})$ for all nodes $u$.
Each of the three steps runs in $O(n)$, where $n$ is the number of nodes.
Therefore, the total time complexity of all three steps is $O(n)$.
The above computation also allows us to compute $\Delta_{S}(B\cup\{u\})$
  for all $u$.
To do so, we have to compute $ap_{\emptyset}(\cdot)$ in extra,
  then we have $\sigma_{S}(\emptyset)=\sum_{v}ap_{\emptyset}(v)$
  and $\Delta_{S}(B\cup\{u\}) = \sigma_{S}(B\cup\{u\})-\sigma_{S}(\emptyset)$.

\mynoindent\textbf{\textit{Greedy-Boost.}}
Based on the computation of $\sigma_S(B\cup\{u\})$ for all nodes $u$,
  we have a greedy algorithm \texttt{Greedy-Boost}
  to solve the $k$-boosting problem on bidirected trees.
In \texttt{Greedy-Boost},
  we iteratively insert into set $B$ a node
  $u$ that maximizes $\sigma_{S}(B\cup\{u\})$, until $|B|=k$.
\texttt{Greedy-Boost} runs in $O(kn)$.

\subsection{A Rounded Dynamic Programming}\label{sec:bidirected_fptas}

In this subsection, we present a rounded dynamic programming \texttt{DP-Boost},
  which is a \textit{fully polynomial-time approximation scheme}.
\texttt{DP-Boost} requires that the tree has a root node,
  any node could be assigned as the root node.
Denote the root node by $r$.
For ease of presentation,
  in this subsection, 
  we assume that every node of the tree has at most two children.
\onlypaper{
We leave details about \texttt{DP-Boost} for general
  bidirected trees in our technical report~\cite{lin2016boosttechreport}.
}
\onlytech{
We leave details about \texttt{DP-Boost} for general
  bidirected trees in the appendix.
}

\mynoindent\textbf{Exact dynamic programming.}
We first define a bottom-up exact dynamic programming.
For notational convenience,
  we assume that $r$ has a \textit{virtual parent} $r'$ and $p_{r'r}=p'_{r'r}=0$.
Given a node $v$, let $V_{T_v}$ be the set of nodes in the subtree of $v$.
Define $g(v,\kappa,c,f)$ as the maximum expected boost of nodes in $V_{T_v}$
  under the following conditions.
(1) \textit{Assumption:} The parent of node $v$ is activated with probability
  $f$ if we remove all nodes in $V_{T_v}$ from $G$.
(2) \textit{Requirement:} We boost at most $\kappa$ nodes in $V_{T_v}$,
  and node $v$ is activated with probability $c$ if we remove nodes not in $V_{T_v}$.
It is possible that for some node $v$,
  the second condition could never be satisfied
  (e.g., $v$ is a seed but $c<1$).
In that case, we define $g(v,\kappa,c,f):=-\infty$.

By definition, $\max_c g(r,k,c,0)$
  is the maximum boost of the influence spread upon boosting at most $k$ nodes
  in the tree.
However, the exact dynamic programming is infeasible in practice because
  we may have to calculate $g(v,\kappa,c,f)$ for exponentially many choices
  of $c$ and $f$.
To tackle this problem,
  we propose a \textit{rounded dynamic programming}
  and call it \texttt{DP-Boost}.

\mynoindent\textbf{High level ideas.}
Let $\delta\in(0,1)$ be a \textit{rounding parameter}.
We use $\lfloor{x}\rfloor_{\delta}$ to denote the value of $x$ rounded down
  to the nearest multiple of $\delta$.
We say $x$ is \textit{rounded} if and only if it is a multiple of $\delta$.
For simplicity, we consider $1$ as a \textit{rounded} value.
The high level idea behind \texttt{DP-Boost} is that
  we compute a \textit{rounded version} of $g(v,\kappa,c,f)$
  only for \textit{rounded values} of $c$ and $f$.
Then, the number of calculated entries would be polynomial in $n$ and $1/\delta$.
Let $g'(v,\kappa,c,f)$ be the rounded version of $g(v,\kappa,c,f)$, \texttt{DP-Boost} guarantees that
(1) $g'(v,\kappa,c,f)\leq g(v,\kappa,c,f)$;
(2) $g'(v,\kappa,c,f)$ gets closer to $g(v,\kappa,c,f)$ when $\delta$ decreases.


\Cref{def:dpboost} defines \texttt{DP-Boost}.
An important remark is that
  $g'(\cdot)$ is equivalent to the definition of $g(\cdot)$ 
  if we ignore all the rounding (i.e., assuming $\rdelta{x}=x,\forall x$).

\begin{definition}[\texttt{DP-Boost}]\label{def:dpboost}
  Let $v$ be a node. Denote the parent node of $v$ by $u$.
  \begin{itemize}[leftmargin=*]
  \item \textit{Base case.} Suppose $v$ is a leaf node.
  If $c\neq \indicator{v\in S}$, let $g'(v,\kappa,c,f)=-\infty$;
  otherwise, let 
  \onlypaper{
    $g'(v,\kappa,c,f)=
      \max \big\{1{-}(1{-}c)(1{-}f\cdot \edgep{u}{v}{\indicator{\kappa>0}}){-}ap_{\emptyset}(v),0\big\}$.
  }
  \onlytech{
  \begin{align*}
    g'(v,\kappa,c,f)=
      \max \big\{1{-}(1{-}c)(1{-}f\cdot \edgep{u}{v}{\indicator{\kappa>0}}){-}ap_{\emptyset}(v),0\big\}.
  \end{align*}
  }
  
  \item \textit{Recurrence formula.} Suppose $v$ is an internal node.
  If $v$ is a seed node, we let $g'(v,\kappa,c,f)=-\infty$ for $c\neq 1$, and
  otherwise let
  \begin{align*} 
    g'(v,\kappa,1,f) = \max_{\kappa=\sum\kappa_{v_i}} \sum_i g'(v_i,\kappa_{v_i},c_{v_i},1).
  \end{align*}
  If $v$ is a non-seed node, we use $C'(v,\kappa,c,f)$ to denote the set of
  \textit{consistent subproblems} of $g'(v,\kappa,c,f)$.
  Subproblems $(\kappa_{v_i},c_{v_i},f_{v_i},\forall i)$
  are consistent with $g'(v,\kappa,c,f)$ if they satisfy the following conditions:
  \onlypaper{
    $b=\kappa-\sum_i \kappa_{v_i}\in\{0,1\}$,
    $c{=}\rdelta{1\!-\!\prod _i \big(1\!-\!c_{v_i}\cdot \edgep{v_i}{v}{b} \big)}$,
    $f_{v_i}{=}\rdelta{1\!-\!(1\!-\!f\cdot \edgep{u}{v}{b})\prod _{j\neq i}
    \big(1\!-\!c_{v_j}\cdot \edgep{v_j}{v}{b}\big)}, \forall i$.
  }
  \onlytech{
  \begin{align*}
    & b=\kappa-\sum_i \kappa_{v_i}\in\{0,1\},
    c{=}\rdelta{1\!-\!\prod _i \big(1\!-\!c_{v_i}\cdot \edgep{v_i}{v}{b} \big)},\\
    &f_{v_i}{=}\rdelta{1\!-\!(1\!-\!f\cdot \edgep{u}{v}{b})\prod _{j\neq i}
    \big(1\!-\!c_{v_j}\cdot \edgep{v_j}{v}{b}\big)}, \forall i.
  \end{align*}
  }
  If $C'(v,\kappa,c,f)=\emptyset$, let $g'(v,\kappa,c,f)=\!-\!\infty$;
  otherwise, let 
  \begin{align*} 
    g'(v,\kappa,c,f) = \!\!\! \max_{\substack{(\kappa_{v_i},f_{v_i},c_{v_i},\forall i)\\
    \in C'(v,\kappa,c,f),\\b=k-\sum_i \kappa_{v_i}}}
    \left(\substack{\sum_i g'(v_i,\kappa_{v_i},c_{v_i},f_{v_i})+ \\
    \max\{1{-}(1{-}c)(1{-}f\cdot \edgep{u}{v}{b}){-}ap_{\emptyset}(v),0\}} \right).
  \end{align*}
  \end{itemize}
\end{definition}

\mynoindent\textbf{Rounding and relaxation.}
In \texttt{DP-Boost},
  we compute $g'(v,\kappa,c,f)$ only for rounded $c$ and $f$.
Therefore, when we search subproblems of $g'(v,\kappa,c,f)$
   for a  internal node $v$,
  we may not find any consistent subproblem
  if we keep using requirements of $c$ and $f_{v_i}$ in the exact dynamic programming
  (e.g., $c=1-\prod_i (1-c_{v_i}\cdot\edgep{v_i}{v}{b})$).
In \Cref{def:dpboost},
  we slightly relax the requirements of $c$ and $f_{v_i}$
  (e.g., $c=\rdelta{1-\prod_i (1-c_{v_i}\cdot\edgep{v_i}{v}{b})}$).
Our relaxation guarantees that
  $g'(v,\kappa,c,f)$ is at most $g(v,\kappa,c,f)$.
The rounding and relaxation together may result in a loss of the
  boosted influence spread of the returned boosting set.
However, as we shall show later, the loss is bounded.

\mynoindent\textbf{Algorithm and complexity.}
In \texttt{DP-Boost}, we first determine the rounding parameter $\delta$.
We obtain a lower bound $LB$ of the optimal boost of influence
  by \texttt{Greedy-Boost} in $O(kn)$.
Define $p^{(k)}(u\rightsquigarrow v)$ as the probability that node $u$
  can influence node $v$
  given that we boost edges with top-$k$ influence probability along the path.
The rounding parameter is then determined by
\begin{align}
  \delta = \frac{\epsilon\cdot \max(LB,1)}{\sum_{u\in V}\sum_{v\in V}p^{(k)}(u\rightsquigarrow v)}.
\end{align}
We obtain the denominator of $\delta$ via depth-first search starting from
  every node, each takes time $O(kn)$.
Thus, we can obtain $\delta$ in $O(kn+kn^2)=O(kn^2)$.
With the rounding parameter $\delta$,
  \texttt{DP-Boost} computes the values of $g'(\cdot)$ bottom-up.
For a leaf node $v$, 
  it takes $O(k/\delta^2)$ to
  compute entries $g'(v,\kappa,c,f)$ for all $\kappa$, rounded
  $c$ and rounded $f$. 
For an internal node $v$,
  we enumerate over all combinations of $f$, $b\in\{0,1\}$,
  and $\kappa_{v_i}$, $c_{v_i}$ for children $v_i$.
For each combination,
  we can uniquely determine the values for $\kappa$, $c$ and $f_{v_i}$ for all children $v_i$,
  and update $g'(v,\kappa,c,f)$ accordingly.
For an internal node $v$,
  the number of enumerated combinations is $O(k^2/\delta^3)$,
  hence we can compute all $k/\delta^2$ entries $g'(v,\ldots)$
  in $O(k^2/\delta^3)$.
The total complexity of \texttt{DP-Boost} is $O(kn^2+n\cdot k^2/\delta^3)$.
In the worst case, we have $O(1/\delta)=O(n^2/\epsilon)$.
Therefore, the complexity for \texttt{DP-Boost} 
  is $O(k^2n^7/\epsilon^3)$.
To conclude, we have the following theorem about \texttt{DP-Boost}.
The approximation guarantee is proved in the appendix.

\begin{restatable}{theorem}{thtreedp}\label{th:tree_dp}
Assuming the optimal boost of influence is at least one,
\texttt{DP-Boost} is a fully-polynomial time approximation scheme,
it returns a $(1-\epsilon)$-approximate solution in $O(k^2n^7/\epsilon^3)$.
\end{restatable}

\mynoindent\textbf{Refinements.}
In the implementation,
  we compute possible ranges of $c$ and $f$ for every node,
  and we only compute $g(v,k,c,f)$ for rounded values of $c$ and $f$
  within those ranges.
Let $c^L_v$ (resp.\ $f^L_v$) be the lower bound of possible values of $c_v$
  (resp.\ $f_v$) for node $v$.
For a node $v$, we let $c^L_v=1$ if $v$ is a seed,
  let $c^L_v=0$ if $v$ is a  leaf node,
  and let $c^L_v=\rdelta{1-\prod_i(1-c^L_{v_i}\cdot p_{v_i,v})}$.
  otherwise.
For the root node $r$, we let $f^L_r=0$.
For the $i$-th child $v_i$ of node $v$, let 
  $f^L_{v_i} = \rdelta{1-(1-f^L_u\cdot p_{uv})\prod_j(1-c^L_{v_j}\cdot p_{v_j,v})}$,
  where $u$ is the parent of $v$.
The upper bound of values of $c$ and $f$ for every node is computed
  assuming all nodes are boosted,
  via a method similar to how to we compute the lower bound.

\mynoindent\textbf{Remarks.}
In \texttt{DP-Boost} for general bidirected trees,
  the bottom-up algorithm for computing $g'(v,\dots)$ is more complicated.
Given that the optimal boost of influence is at least one,
  \texttt{DP-Boost}
  returns a $(1-\epsilon)$-approximate solution in $O(k^2n^9/\epsilon^3)$.
If the number of children of every node is upper-bounded by a constant,
  \texttt{DP-Boost} runs in $O(k^2n^7/\epsilon^3)$.
\onlytech{
Please refer to the appendix for details.
}
\onlypaper{
We refer to interested readers to our technical report~\cite{lin2016boosttechreport}.
}

%% file: experiments.tex
\section{Experiments on General Graphs}\label{sec:experiments}

We conduct extensive experiments using real social networks 
  to test \texttt{PRR-Boost} and \texttt{PRR-Boost-LB}.
Experimental results demonstrate their efficiency and effectiveness,
  and show their superiority over intuitive baselines.
All experiments were conduct on a Linux machine with
  an Intel Xeon E5620@2.4GHz CPU and \SI{30}{GB} memory.
In \texttt{PRR-Boost} and \texttt{PRR-Boost-LB},
  the generation of PRR-graphs and the estimation of objective functions
  are parallelized with OpenMP and executed using eight threads.

\onlypaper{
\begin{table}[ht]
  \centering
  \caption{Statistics of datasets and seeds (all directed)}
  \label{table:dataset}
  \footnotesize
  \begin{tabular}{@{}lrrrr@{}}
  \toprule
            \textbf{Description}        & \textbf{Digg} & \textbf{Flixster} & \textbf{Twitter} & \textbf{Flickr} \\
  \midrule
    number of nodes ($n$)              & \SI{28}{K} & \SI{96}{K} & \SI{323}{K} & \SI{1.45}{M} \\
  number of edges ($m$)               & \SI{200}{K} & \SI{485}{K}& \SI{2.14}{M} & \SI{2.15}{M} \\
  average influence probability       & 0.239 & 0.228 & 0.608 & 0.013 \\
  influence of $50$ influential seeds & \SI{2.5}{K}  & \SI{20.4}{K} & \SI{85.3}{K} & \SI{2.3}{K}  \\
  \bottomrule
  \end{tabular}
\end{table}
}
\onlytech{
\begin{table}[ht]
  \centering
  \caption{Statistics of datasets and seeds (all directed)}
  \label{table:dataset}
  \footnotesize
  \begin{tabular}{@{}lrrrr@{}}
  \toprule
            \textbf{Description}        & \textbf{Digg} & \textbf{Flixster} & \textbf{Twitter} & \textbf{Flickr} \\
  \midrule
    number of nodes ($n$)              & \SI{28}{K} & \SI{96}{K} & \SI{323}{K} & \SI{1.45}{M} \\
  number of edges ($m$)               & \SI{200}{K} & \SI{485}{K}& \SI{2.14}{M} & \SI{2.15}{M} \\
  average influence probability       & 0.239 & 0.228 & 0.608 & 0.013 \\
  influence of $50$ influential seeds & \SI{2.5}{K}  & \SI{20.4}{K} & \SI{85.3}{K} & \SI{2.3}{K}  \\
  influence of $500$ random seeds   & \SI{1.8}{K}  & \SI{12.5}{K} & \SI{61.8}{K} & \SI{0.8}{K}  \\
  \bottomrule
  \end{tabular}
\end{table}
}

\mynoindent\textbf{Datasets.}
We use four real social networks:
  \textit{Flixster}~\citep{jamali2010matrix},
  \textit{Digg}~\citep{hogg2012social},
  \textit{Twitter}~\citep{hodas2014simple}, and
  \textit{Flickr}~\citep{cha2009measurement}.
All dataset have both directed social connections among its users,
  and actions of users with timestamps (e.g., rating movies,
  voting for stories, re-tweeting URLs, marking favorite photos).
We learn influence probabilities on edges using a widely accepted
  method by Goyal et al.~\cite{goyal2010learning}.
We remove edges with zero influence probability
  and keep the largest weakly connected component.
\Cref{table:dataset} summaries our datasets.

\mynoindent\textbf{Boosted influence probabilities.}
To the best of our knowledge, no existing work quantitatively studies
  how influence among people changes respect to different kinds of
  ``boosting strategies''.
Therefore, we assign the boosted influence probabilities as follows.
For every edge $e_{uv}$ with an influence probability of $p_{uv}$,
  let the boosted influence probability $p'_{uv}$ be $1-{(1-p_{uv})}^\beta$ ($\beta>1$).
We refer to $\beta$ as the \textit{boosting parameter}.
Due to the large number of combinations of parameters,
  we fix $\beta=2$ unless otherwise specified.
Intuitively, $\beta=2$ indicates that every activated neighbor of a
  boosted node $v$ has two independent chances to activate $v$.
We also provide experiments showing the impacts of $\beta$.

\mynoindent\textbf{Seed selection.}
\onlypaper{
We use the \texttt{IMM} method~\citep{tang2015influence} to select $50$
  influential nodes.
\Cref{table:dataset} summaries the expected influence spread of selected seeds.
We also conduct experiments with randomly selected seeds. 
The setting maps to the situation where some users become seeds spontaneously.
The experimental results given influential seeds and random seeds provide
  similar insights.
Due to space limits, we leave the detailed results where seed users are randomly
  selected in our technical report.~\cite{lin2016boosttechreport}.
}
\onlytech{
We select seeds in two ways.
(i) We use the \texttt{IMM} method~\citep{tang2015influence} to select $50$
  influential nodes.
In practice, the selected seeds typically correspond to highly influential
  customers selected with great care. 
\Cref{table:dataset} summaries the expected influence spread of
  selected seeds.
(ii) We randomly select five sets of $500$ seeds.
The setting maps to the situation where some users become seeds spontaneously.
\Cref{table:dataset} shows the average expected influence of
  five sets of selected seeds.
}

\mynoindent\textbf{Baselines.}
As far as we know, no existing algorithm is applicable
  to the $k$-boosting problem.
Thus, we compare our proposed algorithms
  with several heuristic baselines, as listed below.
\begin{itemize}[leftmargin=*]
\item \textit{HighDegreeGlobal}:
Starting from an empty set $B$,
  \textit{HighDegreeGlobal} iteratively adds a node with the highest
  \textit{weighted degree} to $B$, until $k$ nodes are selected.
We use four definitions of the \textit{weighted degree},
  for a node $u\notin (S\cup B)$, they are:
(1) the sum of influence probabilities on outgoing edges (i.e., $\sum_{e_{uv}}p_{uv}$);
(2) the ``discounted'' sum of influence probabilities on outgoing edges
(i.e., $\sum_{e_{uv}, v\notin B}p_{uv}$);
(3) the sum of the boost of influence probabilities on incoming edges
(i.e., $\sum_{e_{vu}}[p_{vu}'-p_{vu}]$);
(4) the ``discounted'' sum of the boost of influence probabilities on incoming edges
(i.e., $\sum_{e_{vu}, v\notin B}[p_{vu}'-p_{vu}]$).
Each definition outperforms others in some experiments.
We report the maximum boost of influence among four solutions
  as the result.
\item\textit{HighDegreeLocal}:
The only difference between \textit{HighDegreeLocal} and
\textit{HighDegreeGlobal} is that, we first consider nodes close to seeds.
We first try to select $k$ nodes among neighbors of seeds.
If there is not enough nodes to select, we continue to select among nodes that 
  are two-hops away from seeds, and we repeat until $k$ nodes are selected.
We report the maximum boost of influence among four solutions selected
  using four definitions of the \textit{weighted degree}.
\item\textit{PageRank}:
  We use the \textit{PageRank} baseline for influence maximization
    problems~\cite{chen2010scalable}.
  When a node $u$ has influence on $v$, it implies that node
    $v$ ``votes'' for the rank of $u$.
  The transition probability on edge $e_{uv}$ is $p_{vu}/\rho(u)$,
    where $\rho(u)$ is the summation of influence probabilities on all incoming
    edges of $u$.
  The restart probability is $0.15$.
  We compute the \textit{PageRank} iteratively until two consecutive iteration
    differ for at most $10^{-4}$ in $L_1$ norm.
\item\textit{MoreSeeds}:
  We adapt the \texttt{IMM} framework to select $k$ more seeds with the
    goal of maximizing the \textit{increase} of the expected influence spread.
    We return the selected $k$ seeds as the boosted nodes.
\end{itemize}
We do not compare our algorithms to the greedy algorithm with Monte-Carlo simulations.
Because it is extremely computationally expensive even for
  the classical influence maximization~\citep{kempe2003maximizing,tang2014influence}.

\mynoindent\textbf{Settings.}
For \texttt{PRR-Boost} and \texttt{PRR-Boost-LB},
  we let $\epsilon=0.5$ and $\ell=1$ so that both algorithms
  return $(1-1/e-\epsilon)\cdot\frac{\mu(B^*)}{\Delta_{S}(B^*)}$-approximate
  solution with probability at least $1-1/n$.
To enforce fair comparison, for all algorithms,
  we evaluate the boost of influence spread by $20,000$ Monte-Carlo simulations.

\onlytech{
\subsection{Influential seeds}
In this part, we report results where the seeds 
  are $50$ influential nodes. 
The setting here maps to the real-world situation where the
  \textit{initial adopters} are highly influential users selected with great care. 
We run each experiment five times and report the average results.
}
\onlypaper{
\subsection{Performance evaluation}
In this part, we evaluate the performance of our proposed algorithms.
We report results where the seeds are $50$ influential nodes.
We run each experiment five times and report the average results.
}

\mynoindent\textbf{Quality of solution.}
\Cref{fig:quality_boost} compares
  the boost of the influence spread of solutions returned by different algorithms.
\texttt{PRR-Boost} always return the best solution,
  and \texttt{PRR-Boost-LB} returns solutions
  with slightly lower but comparable quality.
Moreover, 
  both \texttt{PRR-Boost} and \texttt{PRR-Boost-LB} outperform other baselines.
In addition, \textit{MoreSeeds} returns solutions with the lowest quality.
This is because nodes selected by \textit{MoreSeeds} are typically in the part of graph not
	 covered by the existing seeds so that they could generate larger marginal influence.
In contrast, boosting nodes are typically close to existing seeds to make the boosting result
	more effective. 
Thus, our empirical result further demonstrates that $k$-boosting problem differs significantly
	from the influence maximization problem.
 
\begin{figure}[t]
    \centering
    \captionsetup[subfigure]{font=small,oneside,margin={0.6cm,0.0cm}}
    \subfloat{\includegraphics[width=0.5\textwidth]
      {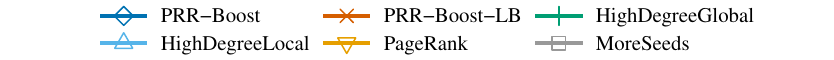}}\\
    \setcounter{subfigure}{0}%
    \subfloat[\textit{Digg}]
    {\includegraphics[width=0.23\textwidth]
      {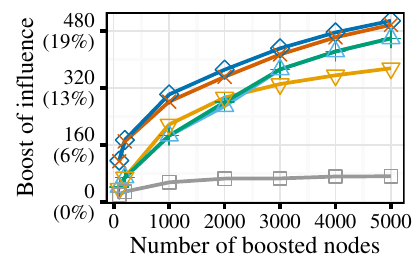}\label{fig:quality_digg}}
    \subfloat[\textit{Flixster}]
    {\includegraphics[width=0.23\textwidth]
      {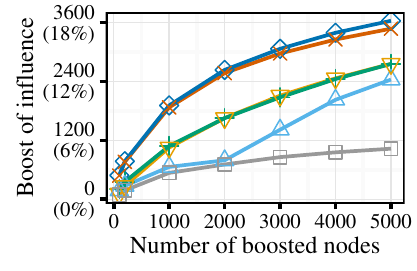}\label{fig:quality_flixster}}
    \makeatletter \if@twocolumn \\ \fi \makeatother
    \subfloat[\textit{Twitter}]
    {\includegraphics[width=0.23\textwidth]
      {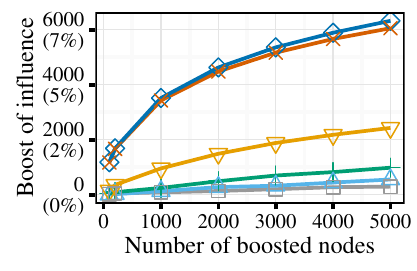}\label{fig:quality_twitter}}
    \subfloat[\textit{Flickr}]
    {\includegraphics[width=0.23\textwidth]
      {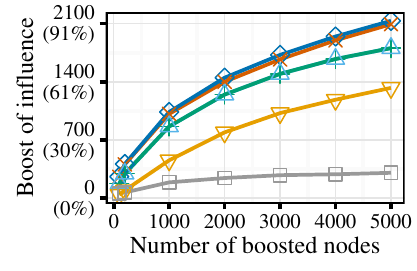}\label{fig:quality_flickr}}
    \caption{Boost of the influence versus $k$\onlytech{ (influential seeds)}.}
    \label{fig:quality_boost}
\end{figure}

\mynoindent\textbf{Running time.}
\Cref{fig:all_time} shows the running time of 
  \texttt{PRR-Boost} and \texttt{PRR-Boost-LB}.
The running time of both algorithm increases when $k$
  increases.
This is mainly because the number of random PRR-graphs required increases when $k$ increases.
\Cref{fig:all_time} also shows that the running time is
  in general proportional to the number of nodes and edges for \textit{Digg},
  \textit{Flixster} and \textit{Twitter}, \textit{but not for Flickr}.
This is mainly because of the 
  significantly smaller average influence probabilities on \textit{Flickr}
  as shown in \Cref{table:dataset}, and the accordingly significantly
  lower expected cost for generating a random PRR-graph (i.e., $EPT$)
  as we will show shortly in \Cref{table:memory}.
In \Cref{fig:all_time}, we also label the speedup of
  \texttt{PRR-Boost-LB} compared with \texttt{PRR-Boost}.
Together with \Cref{fig:quality_boost},
  we can see that \texttt{PRR-Boost-LB}
  returns solutions with quality comparable to \texttt{PRR-Boost}
  but runs faster.
Because our approximation algorithms consistently outperform all
  heuristic methods with no performance guarantee in all tested cases,
  we do not compare the running time of our algorithms with
  heuristic methods to avoid cluttering the results.

\begin{figure}[t]
    \centering
    \captionsetup[subfigure]{font=small,oneside,margin={0.6cm,0.0cm}}
    \subfloat{\includegraphics[width=0.5\textwidth]{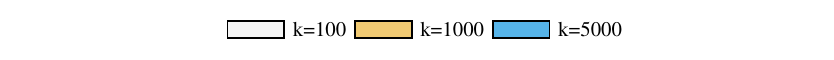}}\\
    \setcounter{subfigure}{0}%
    \subfloat[\texttt{PRR-Boost}]
    {\includegraphics[width=0.23\textwidth]{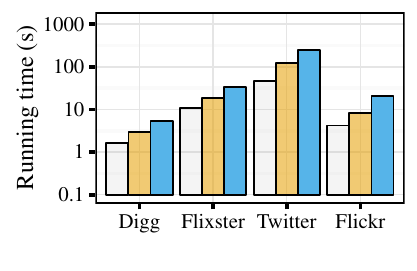}}
    \subfloat[\texttt{PRR-Boost-LB}]
    {\includegraphics[width=0.23\textwidth]{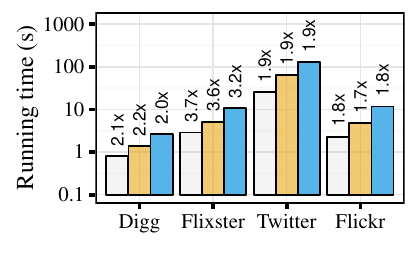}}
    \caption{Running time\onlytech{ (influential seeds)}.}\label{fig:all_time}
\end{figure}

\mynoindent\textbf{Effectiveness of the compression phase.}
\Cref{table:memory} shows the ``compression ratio'' of PRR-graphs and
  memory usages of \texttt{PRR-Boost} and \texttt{PRR-Boost-LB},
  demonstrating the importance of compressing PRR-graphs.
The \textit{compression ratio} is the ratio
  between the average number of uncompressed edges and average number of edges
  after compression in \textit{boostable} PRR-graphs.
Besides the total memory usage,
  we also show in parenthesis the memory usage for storing boostable PRR-graphs,
  which is measured as the additional memory usage starting from the generation
  of the first PRR-graph.
For example, for the \textit{Digg} dataset and $k=100$,
  for boostable PRR-graphs,
  the average number of uncompressed edges is $1810.32$,
  the average number of compressed edges is $2.41$, and
  the compression ratio is $751.59$.
Moreover, the total memory usage of \texttt{PRR-Boost} is $0.07$GB
  with around $0.01$GB being used to storing ``boostable'' PRR-graphs.
The compression ratio is high in practice for two reasons.
\textit{First}, many nodes visited
  in the first phase cannot be reached by seeds.
\textit{Second}, among the remaining nodes,
  many of them can be merged into the super-seed node,
  and most non-super-seed nodes will be removed because they are not on
  any paths to the root node without going through the super-seed node. 
The high compression ratio and the memory used for storing
  compressed PRR-graphs show that
  the compression phase is indispensable.
For \texttt{PRR-Boost-LB}, the memory usage is \textit{much lower}
  because we only store ``critical nodes'' of boostable PRR-graphs.
In our experiments with $\beta=2$,
  each boostable PRR-graph only has a few critical nodes on average,
  which explains the low memory usage of \texttt{PRR-Boost-LB}.
If one is indifferent about the slightly difference between the quality
  of solutions returned by \texttt{PRR-Boost-LB} and \texttt{PRR-Boost},
  we suggest to use \texttt{PRR-Boost-LB} because of
  its lower running time and lower memory usage.

\begin{table}[t]
\makeatletter \if@twocolumn
\fontsize{7}{7}\selectfont
\fi \makeatother
\centering
\caption{Memory usage and compression ratio\onlytech{ (influential seeds)}. Numbers in parentheses are additional memory usage for boostable PRR-graphs.}
\label{table:memory}
\begin{tabular}{@{}lllcc@{}}
\toprule
  \multirow{2}{*}{$k$} & \multirow{2}{*}{\textbf{Dataset}} & \multicolumn{2}{c}{\textbf{PRR-Boost}} & \textbf{PRR-Boost-LB} \\
  \cmidrule(l){3-4} \cmidrule(l){5-5}
  && \textbf{Compression Ratio} & \textbf{Memory (GB)} & \textbf{Memory (GB)} \\ \midrule
  \multirow{4}{*}{100}
&\textit{Digg} &  1810.32 / 2.41 =  751.79 & 0.07 (0.01) & 0.06 (0.00) \\ 
  &\textit{Flixster} &  3254.91 / 3.67 =  886.90 & 0.23 (0.05) & 0.19 (0.01) \\ 
  &\textit{Twitter} & 14343.31 / 4.62 = 3104.61 & 0.74 (0.07) & 0.69 (0.02) \\ 
  &\textit{Flickr} &   189.61 / 6.86 =   27.66 & 0.54 (0.07) & 0.48 (0.01) \\ 
  \midrule
  \multirow{4}{*}{5000}
&\textit{Digg} &  1821.21 / 2.41 =  755.06 & 0.09 (0.03) & 0.07 (0.01) \\ 
  &\textit{Flixster} &  3255.42 / 3.67 =  886.07 & 0.32 (0.14) & 0.21 (0.03) \\ 
  &\textit{Twitter} & 14420.47 / 4.61 = 3125.37 & 0.89 (0.22) & 0.73 (0.06) \\ 
  &\textit{Flickr} &   189.08 / 6.84 =   27.64 & 0.65 (0.18) & 0.50 (0.03) \\
\bottomrule
\end{tabular}
\end{table}

\mynoindent\textbf{Approximation factors in the \textit{Sandwich Approximation}.}
Recall that the approximate ratio of \texttt{PRR-Boost} and
  \texttt{PRR-Boost-LB} depends on the
  ratio $\frac{\mu(B^*)}{\Delta_{S}(B^*)}$.
The closer to one the ratio is,
  the better the approximation guarantee is.
With $B^*$ being unknown due to the NP-hardness of the problem,
  we show the ratio when the boost is relatively large.
We obtain $300$ sets of $k$ boosted nodes by replacing a random number of nodes
  in $B_{sa}$ by other non-seed nodes,
  where $B_{sa}$ is the solution returned by \texttt{PRR-Boost}.
For a given set $B$,
  we use PRR-graphs generated for finding $B_{sa}$ to estimate 
  $\frac{\mu(B)}{\Delta_{S}(B)}$.
\Cref{fig:quality_lb} shows the ratios for generated sets $B$
  as a function of $\Delta_{S}(B)$ for varying $k$.
Because we intend to show the ratio when the boost of influence is large,
  we do not show points corresponding to sets whose boost of influence is less
  than $50\%$ of $\Delta_{S}(B_{sa})$.
For all datasets,
  the ratio is above $0.94$, $0.83$ and $0.74$ for $k=100,1000,5000$, respectively.
The ratio is closer to one when $k$ is smaller,
  and we now explain this.
In practice, most boostable PRR-graphs have ``critical nodes''.
When $k$ is small, say $100$,
  \texttt{PRR-Boost} and \texttt{PRR-Boost-LB} tend to return node sets $B$
  so that every node in $B$ is a critical node in many boostable PRR-graphs.
\onlytech{
For example, for \textit{Twitter}, when $k=100$,
  among PRR-graphs that have critical nodes and are activated upon boosting $B_{sa}$,
  above $98\%$ of them have their critical nodes boosted (i.e., in $B_{sa}$).
}
Meanwhile, many root node $r$ of PRR-graphs without critical nodes may stay inactive.
For a given PRR-graph $R$, if $B$ contains critical nodes of $R$ or if
  the root node of $R$ stays inactive upon boosting $B$,
  $f^{-}_R(B)$ does not underestimate $f_R(B)$.
Therefore, when $k$ is smaller,
  the ratio of $\frac{\mu(B)}{\Delta_{S}(B)}=\frac{\mathbb{E}[f_R^{-}(B)]}{\mathbb{E}[f_R(B)]}$ 
  tends to be closer to one.
When $k$ increases, we can boost more nodes,
  and root nodes of PRR-graphs without critical nodes may be activated,
  thus the approximation ratio tends to decrease.
\onlytech{
For example, for \textit{Twitter}, when $k$ increases from $100$ to $5000$,
  among PRR-graphs whose root nodes are activated upon boosting $B_{sa}$,
  the fraction of them having critical nodes decreases from around $98\%$ to
  $88\%$.
Accordingly, the ratio of $\mu(B_{sa})/\Delta_S(B_{sa})$
  decreased by around $9\%$ when $k$ increases from $100$ to $5000$.
}

\begin{figure}[t]
  \centering
  \includegraphics[width=0.5\textwidth]{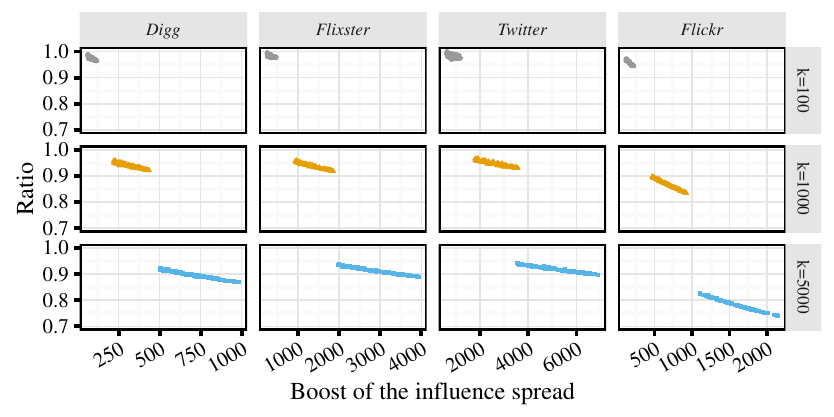}
  \caption{Sandwich Approximation: $\frac{\mu(B)}{\Delta_{S}(B)}$
    \onlytech{ (influential seeds)}. }\label{fig:quality_lb}
\end{figure}

\mynoindent\textbf{Effects of the boosted influence probabilities.}
In our experiments,
  the larger the \textit{boosting parameter} $\beta$ is,
  the larger the boosted influence probabilities on edges are.
\Cref{fig:beta} shows the effects of $\beta$ on the boost of influence and
  the running time when $k=1000$.
For other values of $k$, the results are similar.
In \Cref{fig:beta_a}, 
  the optimal boost increases when $\beta$ increases.
When $\beta$ increases, for \textit{Flixster} and \textit{Flickr},
  \texttt{PRR-Boost-LB} returns solution with quality comparable to those
  returned by \texttt{PRR-Boost}.
For \textit{Twitter},
  we consider the slightly degenerated performance of \texttt{PRR-Boost-LB} acceptable
  because \texttt{PRR-Boost-LB} runs significantly faster.
\Cref{fig:beta_b} shows the running time for
  \texttt{PRR-Boost} and \texttt{PRR-Boost-LB}.
When $\beta$ increases,
  the running time of \texttt{PRR-Boost} increases accordingly,
  but the running time of \texttt{PRR-Boost-LB} remains almost unchanged.
Therefore, compared with \texttt{PRR-Boost},
  \texttt{PRR-Boost-LB} is more scalable
  to \textit{larger boosted influence probabilities on edges}.
In fact, when $\beta$ increases,
  a random PRR-graph tends to include more nodes and edges.
The running time of \texttt{PRR-Boost} increases mainly because the
  cost for PRR-graph generation increases.
However, when $\beta$ increases,
  we observe that the cost for obtaining ``critical nodes'' for a random
  PRR-graph does not change much, thus the running time of
  \texttt{PRR-Boost-LB} remains almost unchanged.
\Cref{fig:beta_ratio} shows the approximation ratio of the sandwich
approximation strategy with varying boosting parameters.
We observe that, for every dataset, when we increase the boosting parameter, 
  the ratio of $\frac{\mu(B)}{\Delta_{S}(B)}$ for large $\Delta_{S}(B)$ 
  remains almost the same.
This suggests that both our proposed algorithms remain effective when we
  increase the boosted influence probabilities on edges. 

\begin{figure}[t]
    \captionsetup[subfigure]{font=small,oneside,margin={0.6cm,0.0cm}}
    \centering
    \subfloat{\includegraphics[width=0.5\textwidth]
      {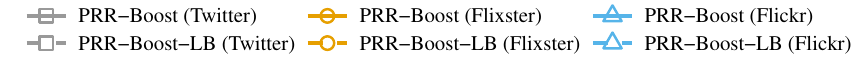}}\\
    \setcounter{subfigure}{0}%
    \subfloat[\textit{Boost of influence}]
    {\includegraphics[width=0.23\textwidth]
      {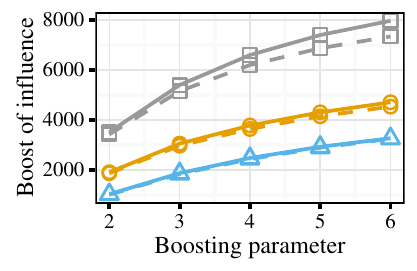}\label{fig:beta_a}}
    \subfloat[\textit{Running time}]
    {\includegraphics[width=0.23\textwidth]
      {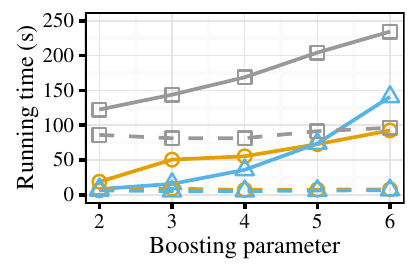}\label{fig:beta_b}}
    \caption{Effects of the boosting parameter (\onlytech{influential seeds, }$k=1000$).}\label{fig:beta}
\end{figure}

\begin{figure}[t]
  \centering
  \includegraphics[width=0.5\textwidth]{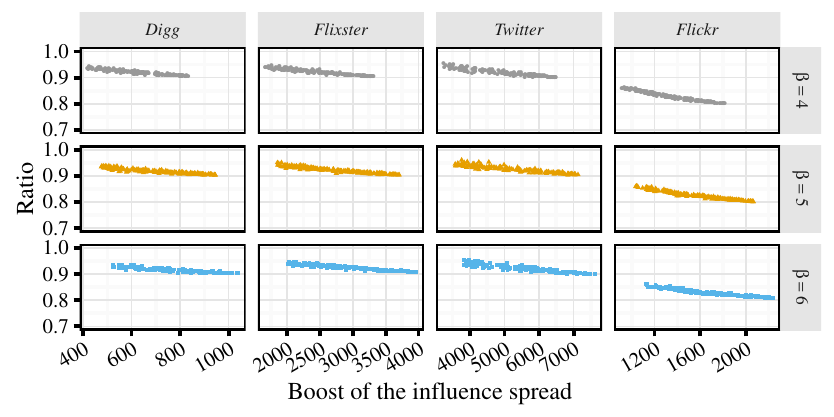}
  \caption{Sandwich Approximation with varying boosting parameter:
    $\frac{\mu(B)}{\Delta_{S}(B)}$ (\onlytech{influential seeds, }$k=1000$).}\label{fig:beta_ratio}
\end{figure}

\onlytech{
\subsection{Random seeds}

In this part,
  we select five sets of $500$ random nodes as seeds
  for each dataset.
The setting here maps to the real situation where some users
  become \textit{seeds} spontaneously.
All experiments are conducted on five sets of random seeds, and we report
the average results.

\mynoindent\textbf{Quality of solution.}
We select up to $5000$ nodes and compare our algorithms with baselines.
From \Cref{fig:rd_quality_boost},
  we can draw conclusions similar to those drawn from
  \Cref{fig:quality_boost}
  where the seeds are highly influential users.
Both \texttt{PRR-Boost} and \texttt{PRR-Boost-LB} 
  outperform all baselines.

\begin{figure}[t]
    \centering
    \captionsetup[subfigure]{font=small,oneside,margin={0.6cm,0.0cm}}
    \subfloat{\includegraphics[width=0.5\textwidth]
      {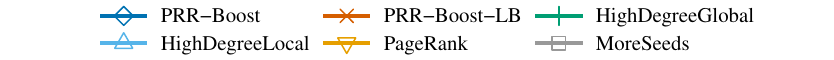}}\\
    \setcounter{subfigure}{0}%
    \subfloat[\textit{Digg}]
    {\includegraphics[width=0.23\textwidth]
      {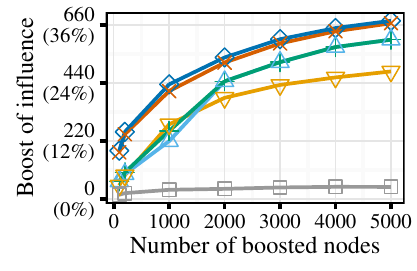}\label{fig:rd_quality_digg}}
    \subfloat[\textit{Flixster}]
    {\includegraphics[width=0.23\textwidth]
      {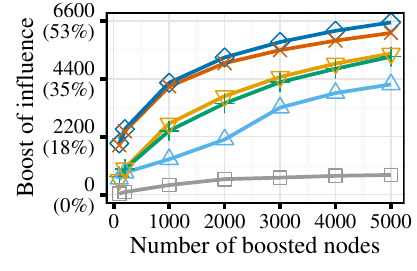}\label{fig:rd_quality_flixster}}
    \makeatletter \if@twocolumn \\ \fi \makeatother

    \subfloat[\textit{Twitter}]
    {\includegraphics[width=0.23\textwidth]
      {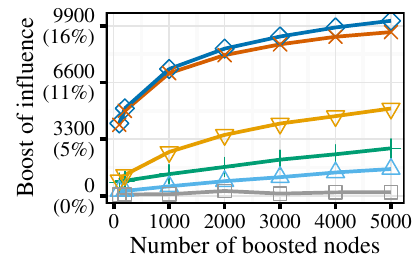}\label{fig:rd_quality_twitter}}
    \subfloat[\textit{Flickr}]
    {\includegraphics[width=0.23\textwidth]
      {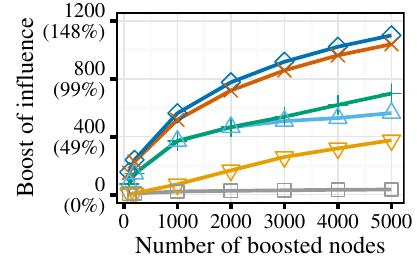}\label{fig:rd_quality_flickr}}
    \caption{Boost of the influence versus $k$ (random seeds).}
    \label{fig:rd_quality_boost}
\end{figure}

\mynoindent\textbf{Running time.}
\Cref{fig:rd_all_time} shows the running time of 
  \texttt{PRR-Boost} and \texttt{PRR-Boost-LB}, and the 
  speedup of \texttt{PRR-Boost-LB} compared with \texttt{PRR-Boost}.
\Cref{fig:rd_all_time2} shows that 
  \texttt{PRR-Boost-LB} runs up to three times faster than \texttt{PRR-Boost}.
Together with \Cref{fig:rd_quality_boost},
  \texttt{PRR-Boost-LB} is in fact both efficient and effective
  given randomly selected seeds.

\begin{figure}[t]
    \centering
    \captionsetup[subfigure]{font=small,oneside,margin={0.6cm,0.0cm}}
    \subfloat{\includegraphics[width=0.5\textwidth]{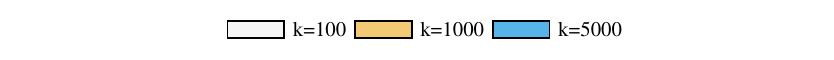}}\\
    \setcounter{subfigure}{0}%
    \subfloat[\texttt{PRR-Boost}]
    {\includegraphics[width=0.23\textwidth]{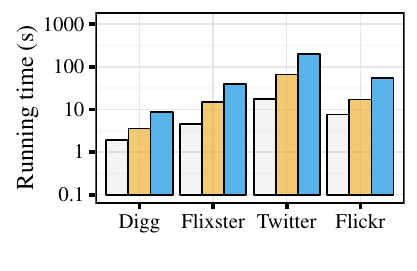}}
    \subfloat[\texttt{PRR-Boost-LB}]
    {\includegraphics[width=0.23\textwidth]{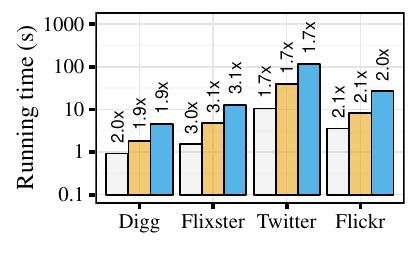}\label{fig:rd_all_time2}}
    \caption{Running time (random seeds).}\label{fig:rd_all_time}
\end{figure}

\mynoindent\textbf{Effectiveness of the compression phase.}
\Cref{table:rd_memory} shows the compression ratio of \texttt{PRR-Boost},
  and the memory usage of both proposed algorithms.
Given randomly selected seed nodes,
  the compression step of PRR-graphs is also very effective.
Together with \Cref{table:memory}, we can conclude that the compression
  phase is an indispensable step for both cases where the seeds are highly
  influence users or random users.

\begin{table}[t]
\fontsize{7}{7}\selectfont
\centering
\caption{Memory usage and compression ratio (random seeds).}
\label{table:rd_memory}
\begin{tabular}{@{}lllcc@{}}
\toprule
  \multirow{2}{*}{$k$} & \multirow{2}{*}{\textbf{Dataset}} & \multicolumn{2}{c}{\textbf{PRR-Boost}} & \textbf{PRR-Boost-LB} \\
  \cmidrule(l){3-4} \cmidrule(l){5-5}
  && \textbf{Compression Ratio} & \textbf{Memory (GB)} & \textbf{Memory (GB)} \\ \midrule
  \multirow{4}{*}{100}
&\textit{Digg} &  3069.15 /  5.61 = 547.28 & 0.07 (0.01) & 0.06 (0.00) \\ 
  &\textit{Flixster} &  3754.43 / 25.83 = 145.37 & 0.24 (0.06) & 0.19 (0.01) \\ 
  &\textit{Twitter} & 16960.51 / 56.35 = 300.96 & 0.78 (0.11) & 0.68 (0.01) \\ 
  &\textit{Flickr} &   701.84 / 18.12 =  38.73 & 0.56 (0.09) & 0.48 (0.01) \\ 
  \midrule
  \multirow{4}{*}{5000}
&\textit{Digg} &  3040.94 /  5.59 = 544.19 & 0.12 (0.06) & 0.07 (0.01) \\ 
  &\textit{Flixster} &  3748.74 / 25.86 = 144.94 & 0.71 (0.53) & 0.21 (0.03) \\ 
  &\textit{Twitter} & 16884.86 / 57.29 = 294.72 & 1.51 (0.84) & 0.72 (0.05) \\ 
  &\textit{Flickr} &   701.37 / 18.10 =  38.75 & 1.00 (0.53) & 0.50 (0.03) \\ 
\bottomrule
\end{tabular}
\end{table}

\mynoindent\textbf{Approximation factors in the \textit{Sandwich Approximation}.}
The approximate ratio of \texttt{PRR-Boost} and
  \texttt{PRR-Boost-LB} depends on the
  ratio $\frac{\mu(B^*)}{\Delta_{S}(B^*)}$.
We use the same method to generate different sets of boosted nodes $B$ as in the
  previous sets of experiments.
\Cref{fig:rd_quality_lb} shows the ratios for generated sets $B$
  as a function of $\Delta_{S}(B)$ for $k\in\{100,1000,5000\}$.
For all four datasets,
  the ratio is above $0.76$, $0.62$ and $0.47$ for $k=100,1000,5000$, respectively.
As from \Cref{fig:quality_lb},
  the ratio is closer to one when $k$ is smaller.
Compared with \Cref{fig:quality_lb}, we observe that the ratios in
  \Cref{fig:rd_quality_lb} are lower.
The main reason is that,
  along with many PRR-graphs with critical nodes,
  many PRR-graphs without critical nodes are also boosted.
For example, for \textit{Twitter}, when $k=5000$,
  among PRR-graphs whose root nodes are activated upon boosting $B_{sa}$,
  around $25\%$ of them do not have critical nodes,
  and around $3\%$ of them have critical nodes but their critical nodes are not in $B_{sa}$.
Note that,
  although the approximation guarantee of our proposed algorithms decreases as
  $k$ increases,
  \Cref{fig:rd_quality_boost} shows that our proposed algorithms still
  outperform all other baselines.

\begin{figure}[ht]
  \centering
  \includegraphics[width=0.5\textwidth]{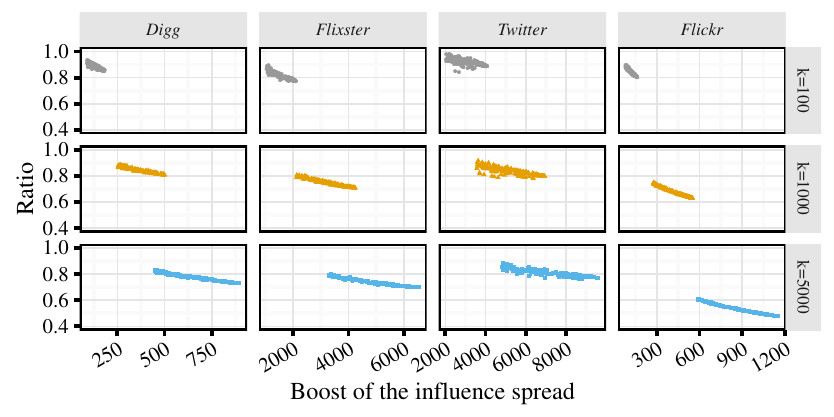}
  \caption{Sandwich Approximation: $\frac{\mu(B)}{\Delta_{S}(B)}$ 
    (random seeds).}\label{fig:rd_quality_lb}
\end{figure}
}

\subsection{Budget allocation between seeding and boosting}
In this part, we vary both the number of seeders and the number of boosted nodes.
Under the context of viral marketing,
  this corresponds to the situation where a company can decide both the
  number of free samples and the number of coupons they offer.
Intuitively, targeting a user as a seeder (e.g., offering a free
  product and rewarding for writing positive opinions)
  must cost more than boosting a user (e.g., offering a discount
  or displaying ads).
In the experiments,
  we assume that we can target $100$ users as seed nodes with all the budget.
Moreover, we assume that targeting a seeder costs $100$ to $800$ times
  as much as boosting a user.
For example, suppose targeting a seeder costs $100$ times as much as
boosting a user:
we can choose to spend $20\%$ of our budget on targeting initial adopters
  (i.e., finding $20$ seed users and boosting $8000$ users);
or, we can spend $80\%$ of the budget on targeting initial adopters
  (i.e, finding $80$ seeds and boosting $2000$ users). 
We explore how the expected influence spread changes, when we decrease the
  number of seed users and increase the number of boosted users.
Given the budget allocation (i.e.,  the number of seeds and the number boosted users),
  we first identify a set of influential seeds using the \texttt{IMM} method,
  then we use \texttt{PRR-Boost} to select the set of nodes we boost.
Finally, we use $20,000$ Monte-Carlo simulations to estimate the expected
  boosted influence spread.

\begin{figure}[ht]
    \centering
    \captionsetup[subfigure]{font=small,oneside,margin={0.6cm,0.0cm}}
    \subfloat{\includegraphics[width=0.5\textwidth]
      {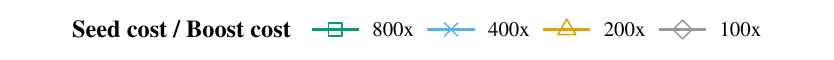}}\\
    \setcounter{subfigure}{0}%
    \subfloat[\textit{Flixster}]
    {\includegraphics[width=0.23\textwidth]
      {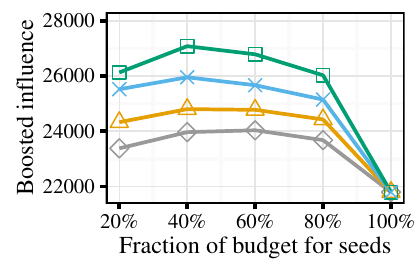}\label{fig:bd_flixster}}
    \subfloat[\textit{Flickr}]
    {\includegraphics[width=0.23\textwidth]
      {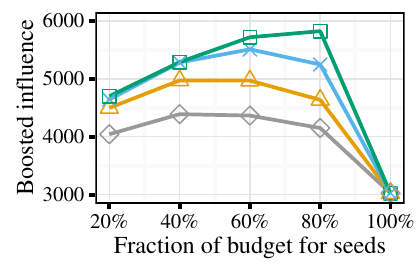}\label{fig:bd_flickr}}
    \caption{Budget allocation between seeding and boosting.}\label{fig:budget}
\end{figure}

\Cref{fig:budget} shows the results for \textit{Flixster} and \textit{Flickr}.
Spending a mixed budget among initial adopters and boosting users
	achieves higher final influence spread than spending all budget on initial adopters.
For example, for cost ratio of $800$ between seeding and boosting, if we choose $80\%$ budget for seeding and
	$20\%$ for boosting, we would achieve around $20\%$ and $92\%$ higher influence spread than 
	pure seeding, for \textit{Flixster} and \textit{Flickr} respectively.
Moreover, the best budget mix is different for different networks and different cost ratio, suggesting the
	need for specific tuning and analysis for each case.

\section{Experiments on Bidirected Trees}\label{sec:experiments_bid}

We conduct extensive experiments to test the proposed algorithms on bidirected trees.
In our experiments, \texttt{DP-Boost} can efficiently approximate the
  $k$-boosting problem for bidirected trees with thousands of nodes.
We also show that \texttt{Greedy-Boost} returns solutions that are near-optimal.
All experiments were conduct on same environment as in \Cref{sec:experiments}.

We use synthetic bidirected trees to test algorithms for bidirected trees
  in \Cref{sec:bidirected}.
For every given number of nodes $n$,
  we construct a complete (undirected) binary tree with $n$ nodes,
  then we replace each undirected edge by two directed edges,
  one in each direction.
We assign influence probabilities $\{p_{uv}\}$'s on edges according to the 
  \textit{Trivalency} model.
Moreover, for every edge $e_{uv}$, let $p'_{uv}=1-{(1-p_{uv})}^2$.
For every tree, we select $50$ seeds using the \texttt{IMM} method.
We compare \texttt{Greedy-Boost} and \texttt{DP-Boost}. 
The boost of influence of the returned sets are computed exactly.
We run each experiment five times with randomly assigned influence probabilities
  and report the averaged results.

\mynoindent\textbf{\texttt{Greedy-Boost} versus \texttt{DP-Boost} with varying $\epsilon$.}
For \texttt{DP-Boost}, 
  the value of $\epsilon$ controls the tradeoff between the 
  accuracy and computational costs.
\Cref{fig:bid_eps} shows results for
  \texttt{DP-Boost} with varying $\epsilon$ and \texttt{Greedy-Boost}.
For \texttt{DP-Boost}, when $\epsilon$ increases,
  the running time decreases dramatically, but the boost is almost unaffected.
Because \texttt{DP-Boost} returns $(1-\epsilon)$-approximate solutions,
  it provides a benchmark for the greedy algorithm.
\Cref{fig:bid_eps_inf} shows that the greedy algorithm
  \texttt{Greedy-Boost} returns near-optimal solutions in practice.
Moreover, \Cref{fig:bid_eps_time} shows
  \texttt{Greedy-Boost} is orders of magnitude faster than
  \texttt{DP-Boost} with $\epsilon=1$ where the theoretical guarantee is in fact lost.

\begin{figure}[ht]
    \centering
    \captionsetup[subfigure]{font=scriptsize,oneside,margin={0.6cm,0.0cm}}
    \subfloat{\includegraphics[width=0.5\textwidth]{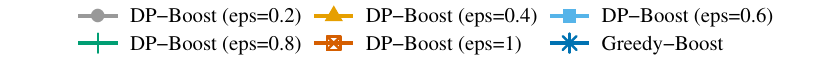}}\\
    \setcounter{subfigure}{0}%
    \subfloat[Boost of influence]
    {\includegraphics[width=0.23\textwidth]{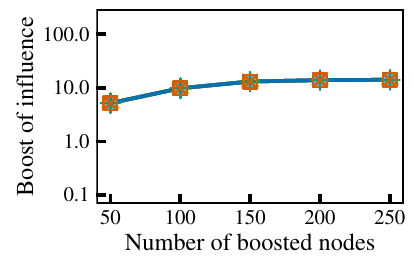}\label{fig:bid_eps_inf}}
    \subfloat[Running time]
    {\includegraphics[width=0.23\textwidth]{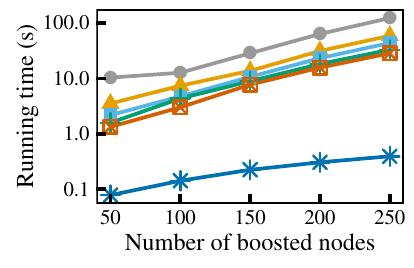}\label{fig:bid_eps_time}}
    \caption{The greedy algorithm versus the rounded dynamic
      programming on random bidirected trees with $2000$ nodes.}\label{fig:bid_eps}
\end{figure}

\mynoindent\textbf{\texttt{Greedy-Boost} versus \texttt{DP-Boost} with varying tree sizes.}
We set $\epsilon=0.5$ for \texttt{DP-Boost}.
\Cref{fig:bid_n} compares \texttt{Greedy-Boost} and \texttt{DP-Boost} for
  trees with varying sizes.
Results for smaller values of $k$ are similar. 
In \Cref{fig:bid_n_inf}, for every $k$, lines representing
  \texttt{Greedy-Boost} and \texttt{DP-Boost} are completely overlapped,
  suggesting that \texttt{Greedy-Boost} always return near-optimal solutions
  on trees with varying sizes. 
\Cref{fig:bid_n_time} demonstrates the efficiency of \texttt{Greedy-Boost}.
Both \Cref{fig:bid_eps} and \Cref{fig:bid_n} suggest that
  \texttt{Greedy-Boost} is very efficient and it returns near-optimal solutions in
  practice.

\begin{figure}[ht]
    \centering
    \captionsetup[subfigure]{font=scriptsize,oneside,margin={0.6cm,0.0cm}}
    \subfloat{\includegraphics[width=0.5\textwidth]{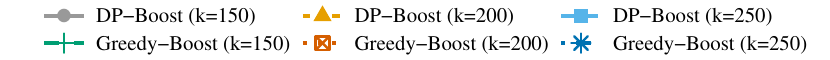}}\\
    \setcounter{subfigure}{0}%
    \subfloat[Boost of influence]
    {\includegraphics[width=0.23\textwidth]{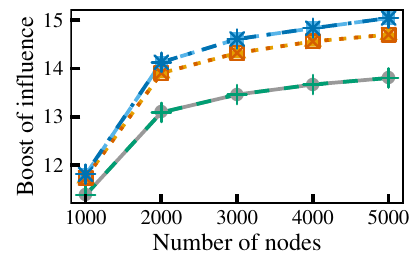}\label{fig:bid_n_inf}}
    \subfloat[Running time]
    {\includegraphics[width=0.23\textwidth]{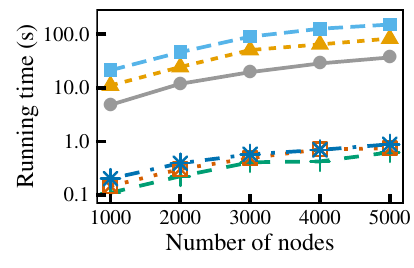}\label{fig:bid_n_time}}
    \caption{The greedy algorithm versus the rounded dynamic
      programming on random bidirected trees with varies sizes.}\label{fig:bid_n}
\end{figure}

%% file: conclusion.tex
\section{Conclusion}\label{sec:conclusion}

In this work, we address a novel $k$-boosting problem that asks 
  how to \textit{boost} the influence spread by offering $k$
  users incentives so that they are more likely to be influenced by friends.
For the $k$-boosting problem on general graphs,
  we develop efficient approximation algorithms,
  \texttt{PRR-Boost} and \texttt{PRR-Boost-LB},
  that have \textit{data-dependent} approximation factors.
Both \texttt{PRR-Boost} and \texttt{PRR-Boost-LB} are delicate integration
  of \textit{Potentially Reverse Reachable Graphs} and
  the state-of-the-art techniques for influence maximization problems.
For the $k$-boosting problem on bidirected trees,
  we present an efficient greedy algorithm \texttt{Greedy-Boost} based on
  a linear-time exact computation of the \textit{boost of influence spread},
  and we also present \texttt{DP-Boost} which is shown to be a fully
  polynomial-time approximation scheme.
We conduct extensive experiments on real datasets using \texttt{PRR-Boost} and
\texttt{PRR-Boost-LB}.
In our experiments, we consider both the case where the seeds are highly
influential users, and the case where the seeds are randomly selected users.
Results demonstrate the superiority of our proposed algorithms 
  over intuitive baselines.
Compared with \texttt{PRR-Boost}, experimental results show that
  \texttt{PRR-Boost-LB} returns solution with comparable quality
  but has significantly lower computational costs.
On real social networks,
  we also explore the scenario where we are allowed to determine 
  how to spend the limited budget on both targeting initial adopters and boosting users.
Experimental results demonstrate the importance of studying the problem of
  targeting initial adopters and boosting users with a mixed strategy.
We also conduct experiments on synthetic bidirected trees using
  \texttt{Greedy-Boost} and \texttt{DP-Boost}.
Results show the efficiency and effectiveness
  of our \texttt{Greedy-Boost} and \texttt{DP-Boost}.
In particular, we show via experiments that \texttt{Greedy-Boost}
  is extremely efficient and returns near-optimal solutions in practice.

The proposed ``boosting'' problem has several more future directions.
One direction is to design more efficient approximation algorithms or
  effective heuristics for the $k$-boosting problem.
This may requires new techniques about how to tackle the non-submodularity
  of the objective function.
These new techniques may also be applied to solve other existing or future
  questions in the area of influence maximization.
Another direction is to investigate similar problems under other influence
  diffusion models,
  for example the well-known \textit{Linear Threshold} (LT) model.
We believe the general question of to how to boost the spread of information
  is of great importance and it deserves more attention.

%% file: appendix_generalgraph.tex
\onlytech{
\lemmaprrboostselect*
}

\onlypaper{
\begin{proof}[\bf Proof of \Cref{lemma:prr_boost_select}]
}
\onlytech{
\begin{proof}
}
Let $B$ be a boost set with $k$ nodes, we say that $B$ is a \textit{bad} set if
  $\Delta_{S}(B)<(1-1/e-\epsilon)\cdot OPT_{\mu}$.
To prove \Cref{lemma:prr_boost_select},
  we first show that each \textit{bad} boost set with $k$ nodes
  is returned by \Cref{algo:prr_boost} with a probability of
  at most $(1-n^{-\ell'})/\binom{n}{k}$.
Let $B$ be an arbitrary \textit{bad} set with $k$ nodes, we have
\begin{align}
  \Delta_{S}(B)<(1-1/e-\epsilon)\cdot OPT_{\mu}.\label[ineq]{eq:bad_solution}
\end{align}
If we return $B$ as the $B_{sa}$,
  we must have $\hat{\Delta}_{\mathcal{R}}(B)>\hat{\Delta}_{\mathcal{R}}(B_{\mu})$.
Therefore, the probability that $B$ is returned is upper bounded by
$\Pr[\hat{\Delta}_{\mathcal{R}}(B)>\hat{\Delta}_{\mathcal{R}}(B_{\mu})]$.
From \Cref{eq:prr_boost_lb2} and \Cref{eq:bad_solution}, we have
\begin{align*}
  & n\cdot \hat{\mu}_{\mathcal{R}}(B_{\mu})-\Delta_{S}(B) \\
  \geq & (1-1/e)\cdot(1-\epsilon_1)\cdot OPT_{\mu}
  - (1-1/e-\epsilon)\cdot OPT_{\mu} \\
  = &(\epsilon-(1-1/e)\cdot \epsilon_1)\cdot OPT_{\mu}.
\end{align*}
Let $\epsilon_2 = \epsilon - (1-1/e)\cdot \epsilon_1$, we have
\begin{align*}
  & \Pr[\hat{\Delta}_{\mathcal{R}}(B) > \hat{\Delta}_{\mathcal{R}}(B_{\mu})] \\
\leq & \Pr[\hat{\Delta}_{\mathcal{R}}(B) > \hat{\mu}_{\mathcal{R}}(B_{\mu})] \\
\leq & \Pr[n\cdot \hat{\Delta}_{\mathcal{R}}(B) - \Delta_{S}(B) 
        > n\cdot \hat{\mu}_{\mathcal{R}}(B_{\mu}) - \Delta_{S}(B)] \\
\leq & \Pr[n\cdot \hat{\Delta}_{\mathcal{R}}(B) - \Delta_{S}(B) 
       > \epsilon_2\cdot OPT_{\mu}].
\end{align*}
Let $p=\Delta_{S}(B)/n$,
  we know $p=\mathbb{E}[f_R(B)]$ from \Cref{lemma:prr_delta}.
Recall that $\theta\!=\!|\mathcal{R}|$ and 
  $\hat{\Delta}_{\mathcal{R}}(B)\!=\!\big(\sum_{R\in\mathcal{R}}f_R(B)\big)/\theta$.
Let $\delta = \frac{\epsilon_2\cdot OPT_{\mu}}{np}$, by Chernoff bound,
  we have
\begin{align*}
& \Pr\left[n\cdot \hat{\Delta}_{\mathcal{R}}(B) - \Delta_{S}(B) 
       > \epsilon_2\cdot OPT_{\mu}\right] \\
= & \Pr\left[\sum_{R\in\mathcal{R}}f_R(B) - \theta p 
        > \frac{\epsilon_2\cdot OPT_{\mu}}{np} \cdot \theta p\right]\\
\leq & \exp\Big(-\frac{\delta^2}{2+\delta}\cdot \theta p\Big)
= \exp\Big(-\frac{\epsilon_2^2\cdot OPT_{\mu}^{2}}
       {2n^2p + \epsilon_2\cdot OPT_{\mu} \cdot n}\cdot \theta \Big) \\
\leq & \exp\Big({-}\frac{\epsilon_2^2{\cdot} OPT_{\mu}^{2}}
                       {2n(1{-}1/e{-}\epsilon){\cdot} OPT_{\mu} 
                        {+}\epsilon_2{\cdot} OPT_{\mu} {\cdot} n}
       {\cdot} \theta \Big) \text{(by Ineq~\eqref{eq:bad_solution})}\\
\leq & \exp\Big(-\frac{(\epsilon-(1-1/e)\cdot\epsilon_1)^2\cdot OPT_{\mu}}
                       {n\cdot (2-2/e)}
       \cdot \theta \Big) \\
\leq & \exp\Big(-\log\big( \tbinom{n}{k}\cdot (2n^{\ell'}) \big) \Big)
       \quad\text{(by Ineq~\eqref{eq:prr_boost_lb1})}\\
\leq & n^{-\ell'} / \tbinom{n}{k}.
\end{align*}
Because there are at most $\binom{n}{k}$ \textit{bad} sets $B$ with $k$ nodes,
  by union bound,
  the probability that \Cref{algo:prr_boost} returns a \textit{bad}
  solution is at most $n^{-\ell'}$.
Because $OPT_{\mu}\geq \mu(B^*)$, 
  the probability that \Cref{algo:prr_boost} returns a solution so that
  $\Delta_{S}(B_{sa})<(1-1/e-\epsilon)\cdot \mu(B^*)$
  is also at most $n^{-\ell'}$.
\end{proof}

%% file: appendix_bidirected.tex
\onlytech{
\lemmaapuv*
}

\onlypaper{
\begin{proof}[\bf Proof of \Cref{lemma:apuv}]
}
\onlytech{
\begin{proof}
}
Suppose we boost the node set $B$.
\Cref{eq:apu} holds from the fact that node $u$ does not get influenced
  if and only if all its neighbors (i.e., $N(u)$) fail to influence it.
\Cref{eq:apuv_slow} holds from the fact that
  node $u$ does not get influenced
  if and only if all of its neighbors in $G_{\udelv}$ (i.e., $N(u)\backslash\{v\}$)
  fail to influence it.
\Cref{eq:apuv_fast} is a direct result from \Cref{eq:apuv_slow},
  and we have $1-ap_B(v\backslash u)\cdot\edgep{v}{u}{B}>0$ because we have assumed that
  non-seed nodes are activated with probability less than one.
\end{proof}

\onlytech{
\lemmaguv*
}

\onlypaper{
\begin{proof}[\bf Proof of \Cref{lemma:guv}]
}
\onlytech{
\begin{proof}
}
In $G_{\udelv}$, suppose we insert the node $u$ into the seed set.
The activation probability of node $u$ itself increases from $ap_B(\udelv)$ to $1$,
  increases by $1-ap_B(\udelv)$.
Let node $w$ be a neighbor of $u$ in $G_{\udelv}$.
The probability that $u$ could be activated by nodes other than $w$ increases 
  from $1-\prod_{x\in N(u)\backslash\{v,w\}}(1-ap_B(x\backslash u)\cdot\edgep{x}{u}{B})$
  to $1$,
  increases by
  \begin{align*}
    \prod_{x\in N(u)\backslash\{v,w\}}\left(1-ap_B(x\backslash u)\cdot\edgep{x}{u}{B}\right)
    =\frac{1-ap_B(\udelv)}{1-ap_B(w\backslash u)\cdot\edgep{w}{u}{B}}.
  \end{align*}
In the above equation,
  we have $1-ap_B(w\backslash u)\cdot\edgep{w}{u}{B}>0$
  because we assume that non-seed nodes are activated with probability less than one.
By definition of $g_B(\udelv)$, we have
\onlypaper{
  \begin{align*}
    g_B(\udelv) = \left(1-ap_B(\udelv)\right) 
    + \sum_{\substack{w\in N(u)\\w\neq v}}
    \frac{1-ap_B(\udelv)} {1-ap_B(w\backslash u)\cdot\edgep{w}{u}{B}}\cdot
    \edgep{u}{w}{B} \cdot g_B(w\backslash u),
  \end{align*}
  }
\onlytech{
  \begin{align*}
    g_B(\udelv) = & \left(1-ap_B(\udelv)\right) \\
    & + \sum_{\substack{w\in N(u)\\w\neq v}}
    \frac{1-ap_B(\udelv)} {1-ap_B(w\backslash u)\cdot\edgep{w}{u}{B}}\cdot
    \edgep{u}{w}{B} \cdot g_B(w\backslash u),
  \end{align*}
  }
  and \Cref{eq:guv} holds.
\Cref{eq:guv_fast} can be derived directly from \Cref{eq:guv}.
\end{proof}

\onlytech{
\lemmasigmau*
}

\onlypaper{
\begin{proof}[\bf Proof of \Cref{lemma:sigmau}]
}
\onlytech{
\begin{proof}
}
If $u\in B$ or $u\in S$, it is obvious that $\sigma_S(B\cup\{u\})=\sigma_S(B)$.
Now, consider $u\in V\backslash(B\cup S)$,
  and we insert node $u$ into the boost set $B$.
The value of $\Delta ap_B(u)$ is the increase of
  the activation probability of node $u$ itself.
The value of $\Delta ap_B(\udelv)$ is the increase of 
  the activation probabilities of node $u$ in $G_{\udelv}$.
Suppose $v$ is a neighbor of $u$.
Let $V_{v\backslash u}$ be the set of nodes in $G_{v\backslash u}$.
In graph $G$, when we insert node $u$ into $B$,
  the expected number of activated nodes in $V_{v\backslash u}$ increases by
  $\edgep{u}{v}{B}\cdot \Delta ap_B(\udelv)\cdot g_B(v,u)$.
Thus, we have \Cref{eq:sigmau} holds.
\end{proof}

\onlytech{
\thtreedp*
}

The proof of \Cref{th:tree_dp} relies on the following lemma.
\begin{lemma}\label{lemma:product_diff_corollary}
  We have $a_{1}a_{2}-b_{1}b_{2}\leq (a_{1}-b_{1})+(a_{2}-b_{2})$ for 
  $0\leq a_{i}\leq 1$, $0\leq b_{i}\leq 1$.
\end{lemma}

\onlytech{
\begin{proof}
}
\onlypaper{
  \begin{proof}[\bf Proof of \Cref{th:tree_dp}]
}
The complexity of \texttt{DP-Boost} has been analyzed in the paper.
Let $B^*$ be the optimal solution of the $k$-boosting problem,
  and assume $\Delta_{S}(B^*)\geq 1$.
Let $\tilde{B}$ be the solution returned by \texttt{DP-Boost}.
To prove \Cref{th:tree_dp}, we only need to prove that we have 
\begin{align}
  \Delta_{S}(\tilde{B})\geq (1-\epsilon)\Delta_{S}(B^*).
\end{align}

Suppose we boost nodes in $B^*$.
For every node $v$,
  let $\kappa_v^*=|B^*\cap T_v|$ be the number of nodes boosted in $T_v$,
  let $c_v^*$ be the probability that node $v$ is
  activated in $T_v$ and let $f_v^*$ be the probability that node $v$'s parent
  is activated in $G\backslash T_v$.
If $v$ is the root node, let $f^*_v=0$.
For every node $v$, let $u$ be the parent of $v$ and let $v_i$ be the $i$-th
  child of $v$, we have
\onlypaper{
\begin{align}
  &g(v,\kappa^*_v,c^*_v,f^*_v) =  \Big(\sum_{v_i} g(v_i,\kappa^*_{v_i},c^*_{v_i},f^*_{v_i})\Big)
   + 1-(1-f^*_v\cdot \edgep{u}{v}{B^*})(1-c^*_v)-ap_\emptyset(v),\\
  &\Delta_{S}(B^*)=
  \sum_{v\in V} \Big(1-(1-f^*_v\cdot \edgep{u}{v}{B^*})(1-c^*_v)-ap_\emptyset(v)\Big).
\end{align}
}
\onlytech{
\begin{align}
  &g(v,\kappa^*_v,c^*_v,f^*_v) = \Big(\sum_{v_i} g(v_i,\kappa^*_{v_i},c^*_{v_i},f^*_{v_i})\Big) \nonumber\\
  & + 1-(1-f^*_v\cdot \edgep{u}{v}{B^*})(1-c^*_v)-ap_\emptyset(v),\\
  &\Delta_{S}(B^*)=
  \sum_{v\in V} \Big(1-(1-f^*_v\cdot \edgep{u}{v}{B^*})(1-c^*_v)-ap_\emptyset(v)\Big).
\end{align}
}
Define $p^*(x\rightsquigarrow y)$ as the probability that node $x$ can
  influence node $y$, given that $B^*$ is boosted.
If $x=y$, we let $p^*(x\rightsquigarrow y)=1$.

We now assign rounded values $\tilde{c}_{v}$ and $\tilde{f}_{v}$ for all node $v$.
The assignment guarantees that for every internal node $v$,
  $(\kappa^*_{v_i},\tilde{c}_{v_i},\tilde{f}_{v_i},\forall i)$ is a consistent subproblem of
  $g'(v,\kappa^*_v,\tilde{c}_v,\tilde{f}_v)$.

We first assign values $\tilde{c}_v$ for every node $v$ in $G$.
We will show that for every node $v$ we have 
\begin{align}\label[ineq]{eq:proof_tree_dp_c}
  c_v^*-\tilde{c}_v\leq \delta \cdot \sum_{x\in T_v}p^*(x\rightsquigarrow v).
\end{align} 
We assign values of $\tilde{c}$ from leaf nodes to the root node.
For a leaf node $v$, let $\tilde{c}_v=\rdelta{c^*_v}$ and
  \Cref{eq:proof_tree_dp_c} holds.
For every internal seed node $v$,
  let $\tilde{c}_v=1$ and \Cref{eq:proof_tree_dp_c} holds.
For every non-seed internal node $v$,
  let $\tilde{c}_v=\rdelta{1-\prod_{v_i} (1-\tilde{c}_{v_i}\cdot
    \edgep{v_i}{v}{B^*})}$,
  and \Cref{eq:proof_tree_dp_c} can be verified as follows.
\begin{align*}
  & c^*_v{-}\tilde{c}_v = 
    1{-}\prod_{v_i} \left(1{-}c^*_{v_i}\cdot \edgep{v_i}{v}{B^*}\right) {-}
    \bigrdelta{1{-}\prod_{v_i} \left(1{-}\tilde{c}_{v_i}\cdot \edgep{v_i}{v}{B^*}\right)} \\
  \leq & \delta{+}\sum_{v_i} (c^*_{v_i}{-}\tilde{c}_{v_i})\cdot \edgep{v_i}{v}{B^*}
  \leq \delta {+} \sum_{v_i}\Big(\delta \sum_{x\in T_{v_i}}p^*(x\rightsquigarrow v_i)
         \cdot\edgep{v_i}{v}{B^*} \Big)\\
  \leq &\delta \sum_{x\in T_v}p^*(x\rightsquigarrow v)
\end{align*}
The first inequality holds from \Cref{lemma:product_diff_corollary},
  and the second inequality holds by induction.

Now, we assign values $\tilde{f}_v$ for every node $v$ from root to leaf.
For every non-root node $v$, denote its parent by $u$, we will show that our
  assignment satisfies
\begin{align}\label[ineq]{eq:proof_tree_dp_f}
  f_{v}^*-\tilde{f}_{v}\leq \delta
  \cdot \sum_{x\in G\backslash T_{v}}p^*(x\rightsquigarrow u).
\end{align} 
For the root node $r$, we have $f^*_r=0$, and we let $\tilde{f}_r=0$.
For an internal seed node $v$ and the $i$-th child $v_i$ of $v$,
  let $\tilde{f}_{v_i}=1$ and \Cref{eq:proof_tree_dp_f} holds
  for $v_i$ because $f^*_{v_i}=1$.
For an internal non-seed node $v$ and its child $v_i$,
  denote the parent of $v$ by $u$,
  let $\tilde{f}_{v_i} =
  \bigrdelta{1-(1-\tilde{f}_v\cdot \edgep{u}{v}{B^*})\prod _{j\neq i}
    \left(1-\tilde{c}_{v_j}\cdot \edgep{v_j}{v}{B^*}\right)}$.
Then, \Cref{eq:proof_tree_dp_f} can be verified as follows.
\begin{align*}
    & f^*_{v_i}-\tilde{f}_{v_i} 
    \leq \delta + (f^*_v-\tilde{f}_v)\cdot \edgep{u}{v}{B^*} 
      + \sum_{j\neq i} (c^*_{v_j}-\tilde{c}_{v_j})\cdot \edgep{v_j}{v}{B^*} \\
    &\leq \delta + \delta\cdot\sum_{x\in G\backslash T_v}p^*(x\rightsquigarrow v) 
      + \delta\cdot\sum_{j\neq i} \sum_{x\in T_{v_j}}p^*(x\rightsquigarrow v) \\
    & \leq \delta\cdot\sum_{x\in G\backslash T_{v_i}}p^*(x\rightsquigarrow v)
\end{align*}
The first inequality holds from \Cref{lemma:product_diff_corollary},
  and the second inequality holds by induction.
  
For every internal node $v$,
  from how we assign the values of $\tilde{c}_{v}$ and $\tilde{f}_{v}$,
  we can conclude that $(\kappa^*_{v_i},\tilde{c}_{v_i},\tilde{f}_{v_i},\forall i)$
  is a consistent subproblem of $g'(v,\kappa^*_v,\tilde{c}_v,\tilde{f}_v)$.
   
Let $\tilde{B}$ be the set of nodes returned by \texttt{DP-Boost}, we have
\begin{align*}
  \Delta_{S}(\tilde{B}) \geq \sum_{v\in V} \max
  \Big\{1-(1-\tilde{f}_v\cdot \edgep{u}{v}{B^*})(1-\tilde{c}_v)-ap_\emptyset(v),0\Big\},
\end{align*}
where we use $u$ to denote the parent of $v$.
Moreover, we have
\begin{align*}
  & \Delta_{S}(B^*)-\Delta_{S}(\tilde{B}) 
  \leq \sum_{v\in V} \big( (f^*_v-\tilde{f}_v)\cdot \edgep{u}{v}{B^*}
         +(c^*_v-\tilde{c}_v)\big)\\
  & \leq \delta\sum_{v\in V}\big(\!\!\!\!\sum_{x\in G\backslash T_v}\!\!\!p^*(x\rightsquigarrow v) 
         {+} \sum_{x\in T_v} p^*(x\rightsquigarrow v)  \big)\\
  & \leq \delta \sum_{v\in V}\sum_{x\in V}p^*(x\rightsquigarrow v)
         \leq \delta \sum_{v\in V}\sum_{x\in V}p^{(k)}(x\rightsquigarrow v).
\end{align*}
Finally, recall that the rounding parameter $\delta$ of \texttt{DP-Boost} is
  $\delta = \frac{\epsilon\cdot \max(LB,1)}
  {\sum_{v\in V}\sum_{x\in V}p^{(k)}(x\rightsquigarrow v)}$,
where $LB$ is a lower bound of $\Delta_S(B^*)$.
We can conclude that
  $\Delta_{S}(\tilde{B})\geq (1-\epsilon)\cdot \Delta_{S}(B^*)$.
\end{proof}

%% file: appendix_problem.tex
\hardness*

Lemma~\ref{lemma:np_hard} proves the NP-hardness of the $k$-boosting problem,
  and Lemma~\ref{lemma:sharp_p_hard} shows the \#P-hardness of the
  boost computation.

\begin{lemma}\label{lemma:np_hard}
The $k$-boosting problem is NP-hard.
\end{lemma}

\onlytech{
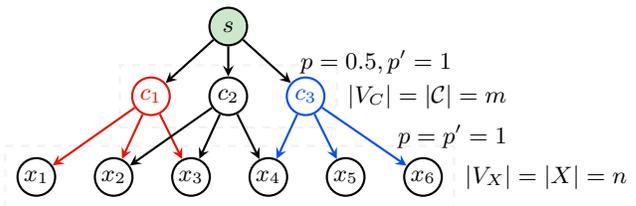
\begin{figure}[htb]
\centering
\begin{tikzpicture}[thick, node distance = 0.4cm and 0.5cm]

  \node[seed node] (s) {$s$};%
  \node[main node] (c2) [below=of s] {$c_2$};%
  \node[red, main node] (c1) [left=of c2] {$c_1$};%
  \node[blue, main node] (c3) [right=of c2] {$c_3$};%
  \node[main node] (x1) [below left=0.7cm and 1.15cm of c1] {$x_1$};%
  \node[main node] (x2) [right=of x1] {$x_2$};%
  \node[main node] (x3) [right=of x2] {$x_3$};%
  \node[main node] (x4) [right=of x3] {$x_4$};%
  \node[main node] (x5) [right=of x4] {$x_5$};%
  \node[main node] (x6) [right=of x5] {$x_6$};
  
  \node[draw, fit=(c1) (c3), inner sep=0.15cm, dashed, gray!10,
    label={right:$|V_C|=|\mathcal{C}|=m$}] {};
  \node[draw, fit=(x1) (x6), inner sep=0.15cm, dashed, gray!10,
    label={right:$|V_X|=|X|=n$}] {};

  \draw [->] (s) -- (c1); 
  \draw [->] (s) -- (c2); 
  \draw [->] (s) -- (c3) node [midway, right] {$\quad p=0.5,p'=1$};%
  \draw [red, ->] (c1) -- (x1);
  \draw [red, ->] (c1) -- (x2); 
  \draw [red, ->] (c1) -- (x3); 
  \draw [->] (c2) -- (x2); 
  \draw [->] (c2) -- (x3); 
  \draw [->] (c2) -- (x4); 
  \draw [blue, ->] (c3) -- (x4); 
  \draw [blue, ->] (c3) -- (x5); 
  \draw [blue, ->] (c3) -- (x6) node [black, midway, right] {$\quad p=p'=1$};

\end{tikzpicture}
\caption{Illustration of the NP-hardness graph construction:
  $X=\{x_1,x_2,\ldots,x_6\}$, $\mathcal{C}=\{C_1,C_2,C_3\}$,
  $C_1=\{x_1,x_2,x_3\}$, $C_2=\{x_2,x_3,x_4\}$,
  $C_3=\{x_4,x_5,x_6\}$, with $m=3$, $n=6$.}\label{fig:np_hard}
\end{figure}
}

\begin{proof}[\unskip Proof]
We prove Lemma~\ref{lemma:np_hard} by a reduction from the NP-complete
  \textit{Set Cover} problem~\cite{karp1972reducibility}.
The \textit{Set Cover} problem is as follows:
Given a ground set $X=\{e_1,e_1,\ldots,e_n\}$
  and a collection $\mathcal{C}$ of subsets $\mathcal{C}=\{C_1, C_2,\ldots,C_m\}$ of $X$,
  we want to know whether there exist $k$ subsets in $\mathcal{C}$ so that their
  union is $X$.
We assume that every element in $X$ is covered by at least one set
  in $\mathcal{C}$.
We reduce the \textit{Set Cover} problem to the $k$-boosting problem as follows.

Given an arbitrary instance of the \textit{Set Cover} problem.
We define a corresponding directed tripartite graph $G$ with $1+m+n$ nodes.
\onlytech{
Figure~\ref{fig:np_hard} shows how we construct the graph $G$.
}
Node $s$ is a seed node.
Node set $V_C=\{c_1,c_2,\ldots,c_m\}$ contains $m$ nodes,
  where node $c_i$ corresponds to the set $C_i$ in $\mathcal{C}$.
Node set $V_X=\{x_1,x_2,\ldots,x_n\}$ contains $n$ nodes,
  where node $x_i$ corresponds to the element $e_i$ in $X$.
For every node $c_i\in V_C$,
  there is a directed edge from $s$ to $c_i$ with
  an influence probability of $0.5$
  and a boosted influence probability of $1$.
Moreover, if a set $C_i$ contains an element $e_j$ in $X$,
  we add a directed edge from $c_i$ to $x_j$ with
  both the influence probability and the boosted influence probability 
  on that edge being $1$.
Denote the degree of the node $x_i\in V_X$ by $d_{x_i}$.
When we do not boost any nodes in $G$ (i.e., $B=\emptyset$),
  the expected influence spread of $S=\{s\}$ in $G$ can be computed as 
$\sigma_{S}(\emptyset)=1+m+\sum_{x_i\in V_X}\Big(1-{0.5}^{d_{x_i}}\Big)$.
The \textit{Set Cover} problem is equivalent to deciding if there is a set $B$
of $k$ boosted nodes in graph $G$ so that
$\sigma_{S}(B)=\sigma_{S}(\emptyset)+\Delta_{S}(B)=1+n+m$.
Because the \textit{Set Cover} problem~\cite{karp1972reducibility} is NP-complete,
  the $k$-boosting problem is NP-hard.
\end{proof}

\begin{lemma}\label{lemma:sharp_p_hard}
Computing $\Delta_{S}(B)$ given $S$ and $B$ is \#P-hard.
\end{lemma}

\onlytech{
\begin{figure}[htb]
\centering
\begin{tikzpicture}[thick, node distance = 1cm and 2cm, graph
  node/.style={draw, ellipse, minimum height=1.5cm,
    minimum width=3cm,fill=gray!15, align=center}]

  \node[graph node] (G) {$G_1$\\$p=0.5$};
  \node[seed node] (s) [left= -8mm of G] {$s$};
  \node[main node] (t) [right= -8mm of G] {$t$};
  \node[main node] (tt) [right= 2.2cm of t] {$t'$};
  \node[draw, fit=(s) (tt), inner sep=0.7cm, dashed, 
  label={[xshift=2.4cm, yshift=-0.5cm]{$G_2$}}] {};

  \draw [->] (t) -- (tt)
    node [midway, above] {$p_{tt'}=0.5$}
    node [midway, below] {$p_{tt'}'=1$};%

\end{tikzpicture}
\caption{Illustration of the \#P-hardness graph
  construction.}\label{fig:sharp_p_hard}
\end{figure}
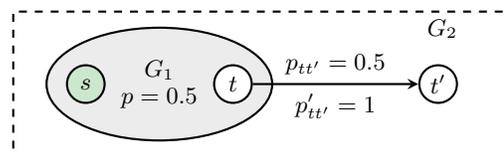
}

\begin{proof}[\unskip Proof]
We prove Lemma~\ref{lemma:sharp_p_hard} by a reduction from the
  \#P-complete counting problem of $s$-$t$
  connectedness in a directed graph~\cite{valiant1979complexity}.
An instance of the $s$-$t$ connectedness is a directed graph $G_1=(V,E)$ and
  two nodes $s$ and $t$ in the graph.
The problem is to count the number of subgraphs of $G_1$ in which $s$
  is connected to $t$.
This problem is equivalent to computing the probability that $s$ is
  connected to $t$ when each edge in $G_1$ has an independent probability
  of $0.5$ to be connected,
  and another $0.5$ to be disconnected.

We reduce this problem to the computation of the boost of influence spread.
Let $G_2=(V\cup\{t'\},E\cup\{e_{tt'}\})$ be a directed graph.
\onlytech{
Figure~\ref{fig:sharp_p_hard} shows the construction of $G_2$.
}
In $G_2$, let $S=\{s\}$, $B=\{t'\}$.
Moreover, let $p_{uv}=0.5$ and $p'_{uv}=1$ for every edge $e_{uv}$ in $G_2$.
We compute $\Delta_{S}(B)$ for graph $G_2$.
Then, $\Delta_{S}(B)/(p_{tt'}'-p_{tt'})=2\Delta_{S}(B)$ is
  the probability that $s$ is connected to $t$ in $G_1$,
  when each edge in $G_1$ has an independent probability of $0.5$ to be connected.
Thus, we solve the $s$-$t$ connectedness counting problem.
Because the $s$-$t$ connectedness problem is \#P-complete,
  the computation of the boost of influence spread
  (i.e., $\Delta_{S}(B)$) is \#P-hard.
\end{proof}

%% file: appendix_dpboost.tex
In this section, we extend \texttt{DP-Boost} in \Cref{sec:bidirected_fptas}
  to tackle the $k$-boosting problem on general bidirected trees.
On general bidirected trees,
  there is no restriction of the number of children of nodes.
For the exact dynamic programming,
  the description in \Cref{sec:bidirected_fptas} naturally
  works for general bidirected trees.
However, the generalization of \texttt{DP-Boost} is
  non-trivial and the definition of \textit{consistent
  subproblem} is much more involved.
Formally, the general \texttt{DP-Boost} is as follows.

\begin{definition}[General \texttt{DP-Boost}]\label{def:dpboost_general}
  Let $v$ be a node. Denote the parent node of $v$ by $u$.
  \begin{itemize}[leftmargin=*]
  \item \textit{Base case.} Suppose $v$ is a leaf node.
  If $c\neq \indicator{v\in S}$, let $g'(v,\kappa,c,f)=-\infty$;
  otherwise, let 
  \begin{align*}
    g'(v,\kappa,c,f)=
      \max \big\{1{-}(1{-}c)(1{-}f\cdot \edgep{u}{v}{\indicator{\kappa>0}}){-}ap_{\emptyset}(v),0\big\}.
  \end{align*}
  
  \item \textit{Recurrence formula.} Suppose $v$ is an internal node.
  If $v$ is a seed node, we let $g'(v,\kappa,c,f)=-\infty$ for $c\neq 1$, and let
  \begin{align*} 
    g'(v,\kappa,1,f) = \max_{\kappa=\sum\kappa_{v_i}} \sum_i g'(v_i,\kappa_{v_i},c_{v_i},1).
  \end{align*}
  If $v$ is a non-seed node with $d\geq 1$ children.
  We use $C'(v,\kappa,c,f)$ to denote the set of
  \textit{consistent subproblems} of $g'(v,\kappa,c,f)$.
  For $1\leq i\leq d$, define
  \begin{align*}
    \delta_v(i)=
    \begin{cases}
    \frac{\delta}{d-2}& 1<i<d,\\
    0& \text{otherwise}.
    \end{cases}
  \end{align*}
  Subproblems $(\kappa_{v_i},c_{v_i},f_{v_i},\forall i)$
  are consistent with $g'(v,\kappa,c,f)$ if they satisfy the following conditions.
  \begin{itemize}
  \item \textit{About $\kappa$ and $\kappa_{v_i}$}:
  $\kappa = \sum_{v_i} \kappa_{v_i} + b$ where $b\in \{0,1\}$.
  \item \textit{About $c$ and $c_{v_i}$}:
  $c=\rdelta{x_d}$, where $x_0=0$ and $x_i= \rdeltav{1-(1-x_{i-1}) (1-c_{v_i}\cdot
    \edgep{v_i}{v}{b})}{\delta_v(i)}$ for $1\leq i\leq d$.
  \item \textit{About $f$ and $f_{v_i}$}:
  $f_{v_i} = \bigrdelta{1-(1-x_{i-1})(1-y_i)}$ for $1\leq i\leq d$, where
  $y_d = f\cdot \edgep{u}{v}{b}$ and
  $y_{i} = \rdeltav{1-(1-y_{i+1})
    (1-c_{v_{i+1}}\cdot\edgep{v_{i+1}}{v}{b})}{\delta_v(i)}$ for $1\leq i\leq d$.
  \end{itemize}

  If $C'(v,\kappa,c,f)=\emptyset$, let $g'(v,\kappa,c,f)=\!-\!\infty$;
  otherwise, let 
  \begin{align*} 
    g'(v,\kappa,c,f) = \!\!\!\!\!\!\!\max_{\substack{(\kappa_{v_i},f_{v_i},c_{v_i},\forall i)\\
    \in C'(v,\kappa,c,f),\\b=k-\sum_i \kappa_{v_i}}}
    \left(\substack{\sum_i g'(v_i,\kappa_{v_i},c_{v_i},f_{v_i})+ \\
    \max\{1{-}(1{-}c)(1{-}f\cdot \edgep{u}{v}{b}){-}ap_{\emptyset}(v),0\}} \right).
  \end{align*}
  \end{itemize}
\end{definition}

For a bidirected tree where every node has at most two children,
  Definition~\ref{def:dpboost_general} degenerates to
  Definition~\ref{def:dpboost}.
Definition~\ref{def:dpboost_general} defines consistent subproblems
  of $g'(v,\kappa,c,f)$,
  with some intermediate variables $x_i$ and $y_i$.
Intuitively, when there is no rounding,
  $x_i$ is the probability that $v$ is activated in $T_v\backslash(\cup_{j>i}T_{v_j})$,
  and $y_i$ is the probability that $v$ is activated in $G\backslash(\cup_{j\leq
  i}T_{v_j})$.
To prevent the number of possible values of $x_i$ and
  $y_i$ from growing exponentially with $d$,
  we also round $x_i$ and $y_i$ for all $1<i<d$.

\Cref{algo:rounded_dp_general} depicts the framework of the general
  \texttt{DP-Boost}.
First, Lines~\ref{algo:dpboost_delta_start}-\ref{algo:dpboost_delta_end}
  determine the rounding parameter $\delta$.
With the rounding parameter $\delta$, we compute the values of $g'(\ldots)$ bottom-up.
There are four subroutines that compute the values of $g'(v,\ldots)$
  for different types of nodes.
We will describe each subroutine in detail.
For notational convenience, in the remaining of this section,
  we use $u$ to denote the parent of node $v$ when the context is clear.

\begin{algorithm}[t]
\SetInd{0.5em}{0.5em}
\SetKw{KwAnd}{and}%
\SetKwFunction{algo}{Rounded\_DP}%
\SetKwFunction{subleaf}{Leaf}%
\SetKwFunction{subseed}{InternalSeed}%
\SetKwFunction{subnonseedchild}{NonseedWithChild}%
\SetKwFunction{subnonseedchildren}{NonseedWithChildren}%
\SetKwProg{myalgorithm}{Algorithm}{}{}%
$B_{greedy}=\textit{Greedy-Boost}(G,S,k)$\label{algo:dpboost_delta_start}\;
$\delta = \frac{\epsilon\cdot \max(\Delta_{S}(B_{greedy}),1)}
{2\sum_{v\in V}\sum_{x\in V}p^{(k)}(x\rightsquigarrow v)}$\tcp*{rounding parameter}\label{algo:dpboost_delta_end}
\For(\tcp*[f]{compute $g'(\dots)$}){nodes $v$ from leaf to root} {%
  \lIf{$v$ is a leaf} {\subleaf{$v$}}%
  \lElseIf{$v$ is a seed} {\subseed{$v$}}%
  \lElseIf{$v$ has non child} {\subnonseedchild{$v$}}%
  \lElse{\subnonseedchildren{$v$}}%
}%
\KwRet{node set $B$ corresponding to $\max_c g'(r,k,c,0)$}\;
\caption{General DP-Boost $(G,S,k,\epsilon)$}\label{algo:rounded_dp_general}
\end{algorithm}

\mynoindent\textbf{\texttt{Leaf($v$)}}:
Suppose node $v$ is a leaf node, we assign $g'(v,\kappa,c,f)$ for all $\kappa$,
  and rounded value of $c$ and $f$ by \Cref{def:dpboost_general}.
There are $O(k/\delta^2)$ entries, assigning each entry takes $O(1)$ time.
Thus, the complexity of this subroutine is $O(k/\delta^2)$.

\mynoindent\textbf{\texttt{InternalSeed($v$)}}: %
Suppose node $v$ is a internal seed node with $d$ children ($d\geq 1$). 
We first initialize $g'(v,\kappa,c,f)=-\infty$ for all $\kappa$ and rounded
  values of $c$ and $f$.
Because $v$ is a seed node, node $v$ is activated in $T_v$
  with probability $c=1$ and we must have $f_{v_i}=1$
  for the $i$-th child $v_i$ of $v$.
However, the number of assignments of 
  $(\kappa_{v_i},c_{v_i},f_{v_i}=1,\forall i)$ can still
  grow exponentially with the number of children,
  because $c_{v_i}$ for each children $v_i$ could have up to $1/\delta$ possible values.
To tackle this problem, we define a \textit{helper function}
  $h(i,\kappa)$ for $1\leq i\leq d$ and $0\leq \kappa\leq k$.
When there is no rounding,
  the value of $h(i,\kappa)$ is the maximum increase of the expected number of
  activated nodes in $\cup_{j\leq i}T_{v_j}$ (i.e., the first $i$ subtrees of $v$),
  given that we boost at most $\kappa$ nodes in the first $i$ subtrees of $v$.
Formally, $h(i,\kappa)$ is defined as follows:
\begin{align*}
  h(i,\kappa) = &\max_{\kappa = \sum_{j=1}^i \kappa_{j}}
                \left\{\sum_{j=1}^i 
                  \max_{c_{v_j}}\big\{g'(v_{j},\kappa_{v_j},c_{v_j},1)\big\}
                \right\}.
\end{align*}
Assuming $h(0,\kappa)=0$ for all $\kappa$.
The value of $h(i,\kappa)$ for $i>1$ can be efficiently computed as follows:
\begin{align}
  h(i,\kappa) = &\max_{0\leq \kappa_{v_i}\leq \kappa} \left\{ h(i-1, \kappa-\kappa_{v_i}) 
                + \max_{c_{v_i}}g'(v_i,\kappa_{v_i}, c_{v_i}, 1)\right\}.
\end{align}
\Cref{algo:subroutine_seed} depicts \texttt{InternalSeed($v$)}.
The complexity of this subroutine is $O(d\cdot k\cdot (k+1/\delta))$.

\begin{algorithm}[t]
\SetKw{KwAnd}{and}%
Initialize $g'(v,\kappa,c,f)\leftarrow -\infty$ for all $0\leq\kappa\leq k$ and rounded $c$ and $f$\;%
Initialize $h(i,\kappa)\leftarrow 0$, for $0\leq i\leq d$ and $0\leq \kappa\leq k$\;%
\For(\tcp*[f]{$d$ is the number of children of $v$}){$i\leftarrow 1$ \KwTo $d$}{%
  \For{$\kappa_{v_i}\leftarrow 0$ \KwTo $k$} {%
    $\text{max}_g\leftarrow \max_{c_{v_i}}g'(v_i,\kappa_{v_i},c_{v_i},1)$%
    \tcp*{$f_{v_i}=1$}%
    \For{$\kappa\leftarrow \kappa_{v_i}$ \KwTo $k$} {%
      $h(i,\kappa) = \max\big\{h(i,\kappa), h(i-1,\kappa-\kappa_{v_i}) +
      \text{max}_g\big\}$\;%
    }%
  }%
}%
\ForAll{$0\leq\kappa\leq k$ \KwAnd rounded $f$}{%
  $g'(v,\kappa,1,f) \leftarrow h(d,\kappa)$%
  \tcp*{boost $\kappa$ nodes among the first $d$ subtrees}%
}%
\caption{\texttt{InternalSeed($v$)}}\label{algo:subroutine_seed}%
\end{algorithm}

\mynoindent\textbf{\texttt{NonseedWithChild($v$)}}: %
Suppose $v$ is a non-seed node with one child.
We first initialize $g'(v,\kappa,c,f)=-\infty$ for all $\kappa$ and rounded
  values of $c$ and $f$.
Then, we compute $g'(v,\kappa,c,f)$ by Definition~\ref{def:dpboost_general}.
\Cref{algo:subroutine_child} depicts this subroutine.
The complexity of this subroutine is $O(k/\delta^2)$.

\begin{algorithm}[t]
Initialize $g'(v,\kappa,c,f)\leftarrow -\infty$ for all $0\leq\kappa\leq k$ and rounded $c$ and $f$\;%
\ForAll{$b\in\{0,1\}$, rounded $c_{v_1},f$}{%
  $c\leftarrow \rdelta{c_{v_1}\cdot\edgep{v_1}{v}{b}}$, %
  $f_{v_1}\leftarrow \rdelta{f\cdot\edgep{u}{v}{b}}$\; %
  $boost_v\leftarrow \max\big\{1-(1-c)(1-f\cdot \edgep{u}{v}{b})-ap_\emptyset(v),0\big\}$\;
  \For{$\kappa=b$ \KwTo $k$}{%
    $g'(v,\kappa,c,f)\leftarrow \max\big\{g'(v,\kappa,c,f),
    g'(v_1,\kappa\!-\!b,c_{v_1},f_{v_1})+boost_v\big\}$\;%
  }%
}%
\caption{\texttt{NonseedWithChild($v$)}}\label{algo:subroutine_child}%
\end{algorithm}

\mynoindent\textbf{\texttt{NonseedWithChildren($v$)}}: %
Suppose node $v$ is a non-seed node with $d\geq 2$ children.
Similar to the subroutine \texttt{InternalSeed($\cdot$)},
  we need a \textit{helper function}.
Let $h(b,i,\kappa,x_i,z_i)$ be the \textit{helper function}, where
  $b\in\{0,1\}$, $2\leq i\leq d$.
Moreover, for $h(b,i,\kappa,x_i,z_i)$, 
  we only consider values of $x_i$ and $z_i$ that are multiples of $\delta_v(i)$.
The helper function $h(b,i,\kappa,x_i,y)$ is formally defined as follows.
\begin{align*}
  h(b,i,&\kappa,x_i,z_i) =\max \quad \sum_{j=1}^{i} g(v_j,\kappa_{v_j},c_{v_j},f_{v_j}) \\
  s.t.\quad &\kappa = \sum_{j=1}^i \kappa_{v_j} + b,  x_0=0, \\
  & x_j= \rdeltav{1{-}(1{-}x_{j{-}1})
    (1{-}c_{v_j}\cdot \edgep{v_j}{v}{b})}{\delta_v(j)}\ (1\leq j\leq i),\\
  &y_i = z_i\cdot \edgep{u}{v}{b} \text{ if } i=d \text{ and } y_i=z_i \text{ otherwise}\\
  &y_{j} = \rdeltav{1{-}(1{-}y_{j+1})
    (1{-}c_{v_{j+1}}\cdot\edgep{v_{j+1}}{v}{b})}{\delta_v(j)} (1\leq j<i)\\
  &f_{v_j} = \rdelta{1{-}(1{-}x_{j{-}1})(1{-}y_j)}\quad (1\leq j\leq i)\\
\end{align*}
When there is no rounding, $h(b,i,\kappa,x,y)$ is the maximum boost of first $i$
  subtrees of $v$ given that
  (1) $b$ indicates whether we boost node $v$;
  (2) $\kappa-b$ nodes are boosted in $\cup_{j\leq i}T_{v_j}$;
  (3) $x$ is the probability that $v$ is activated in
  $T_v\backslash(\cup_{j>i}T_{v_j})$; and
  (4) $y$ is the probability that $v$ is activated in $G\backslash(\cup_{j\leq
  i}T_{v_j})$.
Comparing the above definition of the helper function to 
  Definition~\ref{def:dpboost_general}, we know 
\begin{align*}
  g'(v,\kappa,c,f) = &\max_{0\leq b\leq \indicator{\kappa>0}} \left(\substack{h(b,d,\kappa-b,c,f) \\
  + \max\{1{-}(1{-}c)(1{-}f\cdot\edgep{u}{v}{b}){-}ap_\emptyset(v),0\}}\right)
\end{align*}
For the boundary case where $i=2$, the helper function is computed by its definition.
When $2<i\leq d$, the helper function is computed efficiently as follows.
\begin{align*}
  h(b,i,&\kappa,x_i,z_i) = \max \left( \substack{h(b,i-1,\kappa - \kappa_{v_i},x_{i-1},z_{i-1})\\+g'(v_i,\kappa_{v_i},c_{v_i},f_{v_i})} \right)\\
  s.t.\quad & 0\leq \kappa_{v_i}\leq \kappa, \ %
              x_i = \rdeltav{1-(1-x_{i-1})(1-c_{v_i}\cdot \edgep{v_i}{v}{b})}{\delta_v(i)}\\
          & y_i = z_i\cdot \edgep{u}{v}{b} \text{ if } i=d 
            \text{ and } y_i=z_i \text{ otherwise}\\
          & z_{i-1} = \rdeltav{1-(1-c_{v_i}\cdot\edgep{v_i}{v}{b})(1-y_i)}{\delta_v(i-1)}\\
          & f_{v_i} = \rdelta{1-(1-x_{i-1})(1-y_i)}
\end{align*}
\Cref{algo:subroutine_children} depicts this subroutine.
The complexity of initializing $g'(v,\dots)$ is $O(k/\delta^2)$.
The complexity of
  Lines~\mbox{\ref{algo:subroutine_children_dp_s}-\ref{algo:subroutine_children_dp_e}}
  is $O(d\cdot k^2\cdot \frac{1}{\delta}\cdot {\left(
  \frac{d-2}{\delta} \right)}^2)=O(k^2d^3/\delta^3)$.
The complexity of
  Lines~\ref{algo:subroutine_children_g_s}-\ref{algo:subroutine_children_g_e}
  is $O(k/\delta^2)$.
Therefore, the complexity of \Cref{algo:subroutine_children}
  is $O(k^2d^3/\delta^3)$.

\begin{algorithm}[t]
Initialize $g'(v,\kappa,c,f)\leftarrow -\infty$ for all $0\leq\kappa\leq k$ and rounded $c$ and $f$\;%
\For{$b\leftarrow 0$ \KwTo $1$}{%
  Initialize $h(b,i,\kappa,x_i,z_i)\leftarrow 0$, $\forall
  i,\kappa,x_i,z_i$\;%
  \ForAll(\tcp*[f]{boundary case})%
  {rounded $c_{v_1}, c_{v_2}, z_{2}$}{\label{algo:subroutine_children_dp_s}%
    $x_2\leftarrow \rdeltav{1-(1-c_{v_1}\cdot\edgep{v_1}{v}{b})
      (1-c_{v_2}\cdot\edgep{v_2}{v}{b})}{\delta_v(2)}$\;%
    \leIf{$d=2$}{$y_2\leftarrow z_{2}\cdot\edgep{u}{v}{b}$}{$y_2\leftarrow
      z_{2}$}%
    $f_{v_1} \leftarrow \rdelta{1-(1-c_{v_2}\cdot \edgep{v_2}{v}{b})(1-y_2)}$,\ %
    $f_{v_2} \leftarrow \rdelta{1-(1-c_{v_1}\cdot \edgep{v_1}{v}{b})(1-y_2)}$\;%
    \ForAll{$\kappa_{v_1}+\kappa_{v_2}+b\leq k$}{%
      $\kappa \leftarrow \kappa_{v_1} + \kappa_{v_2} + b$\;%
      $h(b,2,\kappa,x_2,z_2)\leftarrow \max\big\{h(b,2,\kappa,x_2,z_2),
      g'(v_1,\kappa_{v_1},c_{v_1},f_{v_1})+g'(v_2,\kappa_{v_2},c_{v_2},f_{v_2})\big\}$\;%
    }%
  }%
  \For(\tcp*[f]{helper function for $2< i\leq d$}){$i=3$ \KwTo $d$}{%
    \ForAll{rounded $x_{i-1},c_{v_i},z_i\in[0,1]$}{%
      $x_i\leftarrow
      \rdeltav{1-(1-x_{i-1})(1-c_{v_i}\cdot\edgep{v_i}{v}{b})}{\delta_v(i)}$\;%
      \leIf{$i=d$}{$y_i\leftarrow z_i\cdot\edgep{u}{v}{b}$}%
      {$y_i\leftarrow z_i$}%
      $z_{i-1} \leftarrow
      \rdeltav{1-(1-c_{v_i}\cdot\edgep{v_i}{v}{b})(1-y_i)}{\delta_v(i-1)}$\;%
      $f_{v_i} \leftarrow \rdelta{1-(1-x_{i-1})(1-y_i)}$\;%
      \ForAll{$\kappa$ \KwAnd $0\leq\kappa_{v_i}\leq \kappa$}{%
        $h(b,i,\kappa,x_i,z_i)\leftarrow\max\big\{h(b,i,\kappa,x_i,z_i),
        h(b,i-1,\kappa-\kappa_{v_i},x_{i-1},z_{i-1})+g'(v_i,\kappa_{v_i},c_{v_i},f_{v_i})\big\}$
      }%
    }%
  }\label{algo:subroutine_children_dp_e}%
  \ForAll(\tcp*[f]{$g'(v,\ldots)$}){$b\leq \kappa\leq k$ and rounded $c,f$}%
  {\label{algo:subroutine_children_g_s}%
    $boost_v \leftarrow \{1-(1-c)(1-f\cdot\edgep{u}{v}{b})-ap_\emptyset(v),0\}$\;%
    $g'(v,\kappa,c,f)\leftarrow \max\big\{g'(v,\kappa,c,f),
    h(b,d,\kappa,c,f)+boost_v\big\}$\label{algo:subroutine_children_g_e}\;%
  }
}%
\caption{\texttt{NonseedWithChildren($v$)}}\label{algo:subroutine_children}%
\end{algorithm}

\mynoindent\textbf{Complexity of the general \texttt{DP-Boost}:}
In the general \texttt{DP-Boost},
  we first determine the rounding parameter $\delta$.
The time complexity is $O (kn^2)$, as in \Cref{sec:bidirected_fptas}.
With the rounding parameter,
  we compute the values of $g' (\cdot)$ bottom-up.
The most expensive subroutine is the
  \texttt{NonseedWithChildren} subroutine, which
  runs in $O (k^2d^3/\delta^3)$.
Therefore, the total complexity of the general \texttt{DP-Boost}
  is $O (k^2/\delta^3 \cdot \sum_v d_v^3)$, where $d_v$ is the number of
  children of node $v$.
In the worst case,
  we have $O (1/\delta) = O (n^2/\epsilon)$ and $O (\sum_v{d_v}^3)=O ({n}^3)$.
Therefore, the complexity for the general \texttt{DP-Boost} is
  $O (k^2n^9/\epsilon^3)$.
For the special case where the number of children for
  every node is bounded by a constant (e.g., two),
  we have $O(\sum_v d_v^3) = O (n)$
  and the complexity of the general \texttt{DP-Boost} is $O (k^2n^7/\epsilon^3)$.
To conclude, we have the following theorem.

\begin{theorem}\label{th:boostdp_general}
Assuming the optimal boost of influence is at least one,
  the general \texttt{DP-Boost} is a fully-polynomial time approximation scheme,
  it returns a $(1-\epsilon)$-approximate solution in $O(k^2n^9/\epsilon^3)$.
For bidirected trees where the number of children of nodes is upper bounded by a constant,
  the general \texttt{DP-Boost} runs in $O (k^2n^7/\epsilon^3)$.
\end{theorem}

We have analyzed the complexity of the general \texttt{DP-Boost}. 
To prove Theorem~\ref{th:boostdp_general}, we only need to prove the following lemma
  about the approximation ratio.

\begin{lemma}\label{lemma:boostdp_general}
Let $B^*$ be the optimal solution of the $k$-boosting problem,
  and assume $\Delta_{S}(B^*)\geq 1$.
Let $\tilde{B}$ be the solution returned by the general \texttt{DP-Boost},
  we have 
\begin{align}
  \Delta_{S}(\tilde{B})\geq (1-\epsilon)\Delta_{S}(B^*).
\end{align}
\end{lemma}

Proving \Cref{lemma:boostdp_general} relies on the following lemma,
  which can be proved by induction.
\Cref{lemma:product_diff_corollary_general} is also a direct corollary of
  Lemma~4 in~\cite{bharathi2007competitive}.
\begin{lemma}\label{lemma:product_diff_corollary_general}
For any $a_{1},\ldots,a_n$ and $b_{1},\ldots,b_n$,
  where $0\leq a_{i}\leq 1$, $0\leq b_{i}\leq 1$, we have
  $\prod_{i=1}^n a_{i} - \prod_{i=1}^n b_{i} \leq \sum_{i=1}^n (a_{i}-b_{i})$.
\end{lemma}

Now, we prove \Cref{lemma:boostdp_general}.
\begin{proof}
Let $B^*$ be the optimal solution for the $k$-boosting problem.
Suppose we boost nodes in $B^*$, for every node $v$,
  let $\kappa_v^*=|B\cap T_v|$ be the number of nodes boosted in $T_v$,
  let $c_v^*$ be the probability that node $v$ is activated in $T_v$,
  and let $f_v^*$ be the probability that node $v$'s parent is activated
  in $G\backslash T_v$.
For the root node $r$, let $f_r^*=0$.
For every node $v$, denote its parent by $u$, we have
\begin{align}
  & g(v,\kappa^*_v,c^*_v,f^*_v)= \Big( \sum_{v_i} g(v_i,\kappa^*_{v_i},c^*_{v_i},f^*_{v_i}) \Big) \nonumber\\
  & \quad\quad\quad + 1{-}(1{-}f^*_v\cdot \edgep{u}{v}{B^*})(1{-}c^*_v){-}ap_\emptyset(v),\\
  &\Delta_{S}(B^*)=
  \sum_{v\in V} \Big(1{-}(1{-}f^*_v\cdot \edgep{u}{v}{B^*})(1{-}c^*_v){-}ap_\emptyset(v)\Big).
\end{align}
Define $p^*(x\rightsquigarrow y)$ as the probability that node $x$ can
  influence node $y$ when we boost nodes in $B^*$. 
If $x=y$, define $p^*(x\rightsquigarrow y)=1$.

Now, we assign rounded values $\tilde{c}_{v}$ and $\tilde{f}_{v}$ for all
  node $v$.
The assignment guarantees that for every internal node $v$,
$(\kappa^*_{v_i},\tilde{c}_{v_i},\tilde{f}_{v_i},\forall i)$ is a consistent subproblem of
$g'(v,\kappa^*_v,\tilde{c}_v,\tilde{f}_v)$.

We first assign values $\tilde{c}_v$ for every node $v$ in $G$.
For every node $v$, we will show that our assignment of $\tilde{c}_v$ satisfies
\begin{align}\label[ineq]{eq:proof_fptas_general_c}
  c_v^*-\tilde{c}_v\leq 2\delta \cdot \sum_{x\in T_v}p^*(x\rightsquigarrow v).
\end{align} 
We assign values of $\tilde{c}$ from leaf nodes to the root node.
For every leaf node $v$,
  let $\tilde{c}_v=\rdelta{c^*_v}$,
  then \Cref{eq:proof_fptas_general_c} holds.
For an internal seed node $v$,
  let $\tilde{c}_v=1$,
  then \Cref{eq:proof_fptas_general_c} holds because $c^*_v=1$.
For an internal non-seed node $v$ with $d$ children,
  we compute $\tilde{c}_v$ as follows.
First, let $\tilde{x}_0=0$ and compute 
  and $\tilde{x}_{i}= \rdeltav{1-(1-\tilde{x}_{i-1})(1-\tilde{c}_{v_i}\cdot
  \edgep{v_i}{v}{B^*})}{\delta_v(i)}$ for $1\leq i\leq d$.
Then, let $\tilde{c}_v=\rdelta{\tilde{x}_d}$.
\Cref{eq:proof_fptas_general_c} can be verified as follows.
Define $x^*_0=0$ and $x^*_i=1-\prod_{j=1}^i(1-c^*_{v_j}\cdot\edgep{v_j}{v}{B^*})$ for $1\leq i\leq d$.
We have $x^*_{0}=\tilde{x}_{0}$.
Moreover, for $1\leq i\leq d$, we have
\begin{align*}
  &x^*_i-\tilde{x}_{i} %
  \leq (x^*_{i-1}-\tilde{x}_{i-1}) + (c^*_{v_i}-\tilde{c}_{v_i})\cdot\edgep{v_i}{v}{B^*} + \delta_v(i) \\
  &\leq \sum_{j=1}^i (c^*_{v_j}-\tilde{c}_{v_j})\cdot\edgep{v_j}{v}{B^*} +\sum_{j=1}^i \delta_v(j). 
\end{align*}
The first inequality holds from \Cref{lemma:product_diff_corollary_general},
  and the second inequality holds by induction.
Then, the difference between $c^*_v$ and $\tilde{c}_v$ can be bounded as follows.
\begin{align*}
  & c^*_v-\tilde{c}_v = x^*_d-\rdelta{\tilde{x}_d}
  \leq \sum_{i=1}^d(c^*_{v_i}-\tilde{c}_{v_i})\cdot \edgep{v_i}{v}{B^*}
         + \sum_{i=2}^d \delta_v(i) + \delta \\
  & \leq 2\delta\cdot\sum_{v_i}\cdot \Big(
         \sum_{x\in T_{v_i}}p^*(x\rightsquigarrow v_i)\cdot\edgep{v_i}{v}{B^*}
         \Big) + 2\delta  \\
  & \leq 2\delta\cdot\sum_{x\in T_v}p^*(x\rightsquigarrow v) 
\end{align*}

Now, we assign values $\tilde{f}_v$ for every node $v$ from root to leaf.
For every non-root node $v$, denote its parent by $u$, we will show that our
  assignment satisfies
\begin{align}\label[ineq]{eq:proof_fptas_general_f}
  f_{v}^*-\tilde{f}_{v}\leq 2\delta
  \cdot \sum_{x\in G\backslash T_{v}}p^*(x\rightsquigarrow u).
\end{align} 
For the root node $r$, we have $f^*_r=0$, and we let $\tilde{f}_r=0$.
Suppose $v$ is an internal seed node,
  for every child $v_i$, let $\tilde{f}_{v_i}=1$ and
  \Cref{eq:proof_fptas_general_f}
  holds for $v_i$ because $f^*_{v_i}=1$.
Now, suppose $v$ is an internal non-seed node with $d$ children,
  and denote the parent of $v$ by $u$.
We compute $\tilde{f}_{v_i}$ for child $v_i$ as follows.
First, let $\tilde{y}_d = \tilde{f_v}\cdot \edgep{u}{v}{B^*}$ and
  $\tilde{y}_{i} = \rdeltav{1-(1-\tilde{c}_{v_{i+1}}\cdot
  \edgep{v_{i+1}}{v}{B^*})(1-\tilde{y}_{i+1})}{\delta_v(i)}$ for $1\leq i< d$.
Then, let $\tilde{f}_{v_i}=\rdelta{1-(1-\tilde{x}_{i-1})(1-\tilde{y}_{i})}$ for $1\leq i\leq d$.
\Cref{eq:proof_fptas_general_f} can be verified as follows.
Define $y^*_d = \tilde{f_v}\cdot \edgep{u}{v}{B^*}$ and
  $y^*_{i} = 1-(1-\tilde{c}_{v_{i+1}}\cdot
  \edgep{v_{i+1}}{v}{B^*})(1-y^*_{i+1})$ for $1\leq i< d$.
For $i=d$, we have $y^*_{i}-\tilde{y}_{i}=(f^*_v-\tilde{f}_v)\cdot\edgep{u}{v}{B^*}$.
For $1\leq i<d$, the difference between $y^*_{i}$ and $\tilde{y}_{i}$ 
  can be bounded as follows.
\begin{align*}
  & y^*_{i}-\tilde{y}_{i} \leq (c^*_{v_{i+1}}-\tilde{c}_{v_{i+1}})\cdot\edgep{v_{i+1}}{v}{B^*} 
                               + (y^*_{i+1}-\tilde{y}_{i+1}) + \delta_v(i) \\
  & \leq (f^*_v-\tilde{f}_v)\cdot\edgep{u}{v}{B^*}
         +\sum_{j=i+1}^d(c^*_{v_j}-\tilde{c}_{v_j})\cdot\edgep{v_j}{v}{B^*}
         + \sum_{j=i}^d \delta_v(i). 
\end{align*}
The first inequality holds from \Cref{lemma:product_diff_corollary_general},
  and the second inequality holds by induction.
Then, for $1\leq i\leq d$, the difference between $f^*_{v_i}$ and $\tilde{f}_{v_i}$ can be
bounded as follows.
\begin{align*}
  & f^*_{v_i}-\tilde{f}_{v_i} \leq \Big( 1-(1-x^*_{i-1})(1-y^*_{i})\Big) -\tilde{f}_{v_i} \\
  &\leq (x^*_{i-1}-\tilde{x}_{i-1}) + (y^*_{i} - \tilde{y}_{i}) + \delta \\ 
  &\leq (f^*_v-\tilde{f}_v)\cdot\edgep{u}{v}{B^*}
         +\sum_{j\neq i}(c^*_{v_j}-\tilde{c}_{v_j})\cdot\edgep{v_j}{v}{B^*} 
         +\sum_{j\neq i}\delta_v(i) + \delta \\
  & \leq 2\delta\cdot\sum_{x\in G\backslash T_v}p^*(x\rightsquigarrow v) 
    + 2\delta\cdot\sum_{j\neq i} \sum_{x\in T_{v_j}}p^*(x\rightsquigarrow v) + 2\delta \\
  & \leq 2\delta\cdot\sum_{x\in G\backslash T_{v_i}}p^*(x\rightsquigarrow v)
\end{align*}

For every internal node $v$,
  from how we assign the values of $\tilde{c}_{v}$ and $\tilde{f}_{v}$,
  we can conclude that $(\kappa^*_{v_i},\tilde{c}_{v_i},\tilde{f}_{v_i},\forall i)$
  is a consistent subproblem of $g'(v,\kappa^*_v,\tilde{c}_v,\tilde{f}_v)$.

Let $\tilde{B}$ be the set of nodes returned by the general \texttt{DP-Boost}, we have
\begin{align*}
  \Delta_{S}(\tilde{B}) \geq \sum_{v\in V} \max
  \Big\{1-(1-\tilde{f}_v\cdot \edgep{u}{v}{B^*})(1-\tilde{c}_v)-ap_\emptyset(v),0\Big\},
\end{align*}
where we use $u$ to denote the parent of $v$.
Moreover, we have
\begin{align*}
  & \Delta_{S}(B^*)-\Delta_{S}(\tilde{B}) %
    \leq \sum_{v\in V} \Big( (f^*_v-\tilde{f}_v)\cdot \edgep{u}{v}{B^*} +(c^*_v-\tilde{c}_v)\Big)\\
  & \leq 2\delta\cdot\sum_{v\in V}\Big(\sum_{x\in G\backslash T_v}p^*(x\rightsquigarrow v) 
         + \sum_{x\in T_v} p^*(x\rightsquigarrow v)  \Big) \\
  & \leq 2\delta\cdot\sum_{v\in V}\sum_{x\in V}p^{(k)}(x\rightsquigarrow v).
\end{align*}
Finally, recall that the rounding parameter $\delta$ is
  $\delta = \frac{\epsilon\cdot \max(LB,1)}
  {2\sum_{v\in V}\sum_{x\in V}p^{(k)}(x\rightsquigarrow v)}$,
where $LB$ is a lower bound of $\Delta_S(B^*)$,
we can conclude that
  $\Delta_{S}(\tilde{B})\geq (1-\epsilon)\cdot \Delta_{S}(B^*)$.
\end{proof}

\mynoindent\textbf{Refinements.}
In the implementation of the general \texttt{DP-Boost},
  we also apply the refinements that we have discussed
  in \Cref{sec:bidirected_fptas}.
Recall that in \texttt{NonseedWithChildren}$(v)$ that
  computes $g'(v,\ldots)$ for a non-seed internal node $v$ with multiple children,
  we have to compute a helper function $h(b,i,\kappa,x,z)$.
In our implementation of \texttt{NonseedWithChildren}$(v)$,
  we also compute lower and upper bounds for the values of
  $x$ and $z$ for $h(b,i,\kappa,x,z)$.
The lower bound is computed assuming that we do not boost any node.
The upper bound is computed assuming that all nodes are boosted.